\pdfoutput=1
\documentclass[letterpaper,abstract=true,DIV=14,titlepage=false]{scrartcl}
\usepackage[T1]{fontenc}
\usepackage[utf8]{inputenc}
\usepackage[tt=false]{libertine}
\usepackage{xcolor}
\usepackage{amsmath,amssymb,amsthm}
\usepackage{microtype}
\usepackage{graphicx}
\usepackage{booktabs,multirow,array,makecell}
\usepackage{bbm,dsfont}
\usepackage{float,placeins}
\usepackage{algorithmic}
\usepackage[ruled,linesnumbered,norelsize]{algorithm2e}
\usepackage{thmtools,thm-restate}
\usepackage{caption,subcaption}
\usepackage{tikz}
\usetikzlibrary{arrows.meta,positioning,decorations.pathreplacing,calc,shapes.geometric}
\usepackage[normalem]{ulem}
\usepackage{floatsettings}

\usepackage[datamodel=acmdatamodel,style=acmauthoryear,backend=biber,
  giveninits=true,uniquename=init,uniquelist=false]{biblatex}

\usepackage{hyperref}
\hypersetup{colorlinks=true,citecolor=blue,linkcolor=blue,urlcolor=blue,
  pdftitle={Stackelberg POMDP: Learning to Lead via Reinforcement Learning},
  pdfauthor={Matthias Gerstgrasser, Gianluca Brero, Alon Eden, Darshan Chakrabarti,
    Nicolas Lepore, Vincent Li, Eric Mibuari, Amy Greenwald, David C. Parkes}}

\newcommand{\Description}[1]{}

\theoremstyle{plain}

\DeclareMathOperator{\E}{\mathbb{E}}
\renewcommand{\Pr}{\mathrm{Pr}}
\newcommand{\Prob}{\Pr} 
\newcommand{\tran}{{Tr}} 

\newcommand{\eqalg}{\mathcal{A}}
\newcommand{\type}{\theta}
\newcommand{\types}{{\type}} 

\newtheorem{definition}{Definition}[section]
\newtheorem{proposition}{Proposition}[section]

\newtheorem{theorem}{Theorem}[section]
\newtheorem{lemma}[theorem]{Lemma}
\theoremstyle{plain}
\newtheorem{example}[theorem]{Example}
\newtheorem{examplecontinuation}{Example}
\newenvironment{continuedexample}[1]
  {\renewcommand{\theexamplecontinuation}{\ref{#1}}%
   \begin{examplecontinuation}[continued]}
  {\end{examplecontinuation}}

\newcommand{\PlatformN}{25}
\newcommand{\PlatformSellerN}{25}
\newcommand{\PlatformMean}{0.6972}
\newcommand{\PlatformSE}{0.0154}
\newcommand{\PlatformLow}{0.8047}
\newcommand{\PlatformLowSE}{0.0132}
\newcommand{\PlatformHigh}{0.5897}
\newcommand{\PlatformHighSE}{0.0206}
\newcommand{\PlatformFixed}{0.6431}
\newcommand{\PlatformClipped}{0.7356}
\newcommand{\PlatformUniform}{0.4636}
\newcommand{\PlatformMedian}{0.7430}
\newcommand{\PlatformLowMedian}{0.8429}
\newcommand{\PlatformHighMedian}{0.6431}
\newcommand{\PlatformGain}{0.0540}
\newcommand{\PlatformGainPercent}{8.4}

\newcommand{\PlatformFixedLow}{0.6431}
\newcommand{\PlatformFixedHigh}{0.6431}
\newcommand{\PlatformClippedLow}{0.8280}
\newcommand{\PlatformClippedHigh}{0.6431}
\newcommand{\PlatformUniformLow}{0.5144}
\newcommand{\PlatformUniformHigh}{0.4129}

\newcommand{\authorblock}[3]{%
  \begin{minipage}[t]{0.315\textwidth}\centering\normalfont\small
    \setlength{\parskip}{0pt}%
    {\large #1\par}
    \vspace{2pt}%
    {\small #2\par}
    \vspace{1pt}%
    {\footnotesize\hypersetup{urlcolor=black}\urlstyle{same}%
      \href{mailto:#3}{\nolinkurl{#3}}\par}
      \vspace{10pt}%
  \end{minipage}}

\title{Stackelberg POMDP:\\Learning to Lead via Reinforcement Learning\thanks{This article unifies and extends two conference papers, published at NeurIPS 2022 \citep{BreroLMP22} and ICML 2023 \citep{gerstgrasser2023oracles}, and a preliminary manuscript presented at the ICML Workshop on Reinforcement Learning Theory in 2021 \citep{brero2021learning}. The extensions include a common framework, expanded theoretical analysis and experiments, and additional reproducibility materials. Copyright 2026 held by the owner/author(s). This work is licensed under a \href{https://creativecommons.org/licenses/by/4.0/}{Creative Commons Attribution 4.0 International License}.}}
\author{%
  \begin{tabular}{@{}c@{\hspace{0.015\textwidth}}c@{\hspace{0.015\textwidth}}c@{}}
  \authorblock{Matthias Gerstgrasser\textsuperscript{\textdagger}}{OpenAI}{matthias@gerstgrasser.net} &
  \authorblock{Gianluca Brero\textsuperscript{\textdagger}}{Bryant University}{gbrero@bryant.edu} &
  \authorblock{Alon Eden\textsuperscript{\textdagger}}{Hebrew University of Jerusalem}{alon.eden@mail.huji.ac.il} \\[1.2em]
  \authorblock{Darshan Chakrabarti}{Columbia University}{dc3595@columbia.edu} &
  \authorblock{Nicolas Lepore}{Google}{nlepore33@gmail.com} &
  \authorblock{Vincent Li}{Adaptive Innovations}{vincentli2014@gmail.com} \\[1.2em]
  \authorblock{Eric Mibuari}{Partnership on AI}{eric@partnershiponai.org} &
  \authorblock{Amy Greenwald}{Brown University}{amy_greenwald@brown.edu} &
  \authorblock{David C. Parkes}{Harvard University}{parkes@seas.harvard.edu}
  \end{tabular}\\[0.7em]
  {\small\textsuperscript{\textdagger}Authors contributed equally; order is randomized.}}
\date{October 7, 2026}

\begin{document}
\maketitle
 \vspace{-20pt}
\begin{abstract}

Many real-world domains---including e-commerce platform design, security planning, and multi-agent coordination---feature leader--follower problems where one decision-maker commits to a policy and others react strategically. We develop a reinforcement learning framework for such interactions in sequential environments with partial observations and multiple followers. Followers may adapt through no-regret learning or reinforcement learning, potentially departing from equilibrium behavior. The framework embeds follower adaptation into the leader's environment to construct a single-agent partially observable Markov decision process---the \textit{Stackelberg POMDP}. For \textit{policy-interactive response algorithms}, which access the leader's policy through queries, we prove that an optimal policy based only on the leader's game history yields an optimal commitment under the specified response procedure. We use proximal policy optimization with a centralized critic and train contextual meta-followers to respond across leader policies. In indirect mechanism design, mechanisms using buyer messages achieve higher social welfare than optimal standard sequential price mechanisms across all tested type counts, with responses certified as approximate Bayesian coarse correlated equilibria. In platform design, learned display rules increase mean consumer surplus by \PlatformGainPercent\% over an optimized fixed price cap while accommodating hidden seller costs. In Atari bilateral trade, meta-learned follower responses support joint learning of visual gameplay and economic decisions; assigning leadership to the seller or buyer shifts transaction prices and payoffs in that agent's favor. Controlled ablations examine how response credit, policy consistency, and reward timing affect learning.
\end{abstract}

\section{Introduction}\label{sec:intro}

Today, AI agents operate at the heart of global commerce, logistics, and infrastructure. The pandemic revealed how easily supply chain bottlenecks could disrupt production and delivery, prompting widespread adoption of AI systems for resilience and efficiency \citep{alicke2021covid, javaid2022supply}. Similar transformations are underway in electricity markets, where AI manages real-time decisions in decentralized grids with prosumer participation \citep{abada2023artificial}, and on e-commerce platforms, where machine learning already drives dynamic pricing in fast-moving, highly competitive marketplaces \citep{calvano2020protecting}. The rapid progress of large language models is further accelerating these trends, endowing AI agents with stronger capabilities for reasoning, communication, and autonomous decision-making \citep{wang2024survey}.

Despite the potential benefits, concerns have emerged about unintended consequences. Both the U.S. Federal Trade Commission \citep{FTC} and the European Commission \citep{OECD} have warned that automated agents may create opaque market dynamics. These concerns are well-founded: \citet{calvano2020artificial} showed in simulation that reinforcement learning algorithms can learn to sustain collusive pricing without explicit programming, while \citet{fish2024algorithmic} found that LLM-based pricing agents can autonomously collude in oligopoly settings to the detriment of consumers. Complementing these simulation results, \citet{assad2020algorithmic} found empirically that the adoption of algorithmic pricing software in retail gasoline markets is associated with higher margins.

Such opacity concerns extend beyond markets. ``Lights-out'' factories combine robotic material handling with real-time optimization, operating fulfillment centers with minimal human oversight \citep{NationalAcademies2024Options, BoschSDS2024LightsOut}. These settings similarly raise questions about how to govern interacting autonomous systems whose behavior can be difficult to anticipate.

These examples motivate a common design question: how can one decision-maker shape the behavior of autonomous agents by committing to rules of interaction? We address this question using the \textit{Stackelberg game model}, in which a designated leader commits first and followers act strategically in response \citep{stackelberg1934marktform}. This structure captures the commitment power of a platform designer or regulator. An e-commerce platform, for instance, sets fee structures, display ordering, and search rankings that shape sellers' incentives, while sellers choose how to operate within those rules.

Learning a commitment in these settings requires handling several complexities together. The platform's decisions and sellers' responses unfold over many rounds of interaction rather than once, with today's state shaping tomorrow's. Sellers hold private information that the platform never directly observes---costs, inventory, demand forecasts---known in the game theory literature as the sellers' \textit{types}. And sellers themselves may be learning algorithms, increasingly reinforcement learners~\citep{calvano2020artificial, assad2020algorithmic, rolf2023review}, adapting through trial and error rather than immediately playing an equilibrium.

Our aim in this paper is to address the leader's problem when these complexities are present at once. We formulate the leader--follower interaction as a \textit{partially observable Markov game} (POMG)~\citep{HansenBZ04}---a multi-agent extension of \textit{partially observable Markov decision processes} (POMDPs)~\citep{smallwood1973optimal} in which the state evolves through joint actions and each agent might have private information about the state, e.g., a private type. In our running example, the state of the POMG captures market conditions such as inventory levels; the platform's fee structures and sellers' pricing decisions jointly determine transitions; and each seller's private cost and inventory information is exactly the kind of partial observability the POMG framework is designed for. Within this framework, an agent acts according to a \textit{policy}: a function mapping the agent's observation history to a distribution over actions.

Given a POMG modeling the interaction between all agents, the Stackelberg structure is imposed by granting the leader commitment power: the leader commits to a policy in this game, and followers then adapt through multi-stage response dynamics rather than instantly playing an equilibrium. Crucially, we do not prescribe a specific equilibrium concept for the followers. Instead, we model their response through a general \textit{follower-response oracle}---a mapping from leader policies to distributions over follower strategy profiles. A leader policy that maximizes the leader's objective against the corresponding follower response constitutes a \textit{behavioral Stackelberg equilibrium} (Definition~\ref{def:leaderProblem}). Our central question is: can a leader \textit{learn} such a policy through interactions with followers who hold private information and whose strategies adapt in response to the leader's own behavior?

We study this question through \textit{policy-interactive (PI) response algorithms}, which implement the follower-response oracle by querying the leader's policy on observation histories and observing the actions it returns. Follower learning through repeated gameplay fits this interface because the game generates queries at the histories reached during play. Examples include response procedures based on reinforcement learning (RL)~\citep{sutton2018reinforcement} and no-regret learning~\citep{cesabianchi2006prediction}. RL can also be paired with rewards inferred through inverse reinforcement learning~\citep{ng2000algorithms} to model human-like followers.

Because a PI algorithm adapts to the leader's policy through these trial interactions, we can embed its behavior in the leader's environment, yielding a single-agent POMDP that we call the \textit{Stackelberg POMDP} (Definition~\ref{def:StackPOMDP}). To simulate a query, the POMDP presents the queried game history as its observation; during play, it presents the leader's actual observation history. An unrestricted policy could also remember earlier queries and respond differently to the same game history across phases. We therefore restrict attention to memoryless policies in this POMDP, enforcing a single mapping from game histories to action distributions. Under this restriction, an optimal policy in the Stackelberg POMDP yields an optimal leader commitment under the specified follower-response procedure (Theorem~\ref{thm:main}).

The leader's actions influence its return both by shaping the followers' response and through subsequent gameplay. We show that a policy-gradient update that accounts only for gameplay can point, in expectation, in a direction that decreases the leader's objective (Proposition~\ref{prop:response_credit}). For practical optimization, we use proximal policy optimization (PPO)~\citep{schulman2017proximal} with a centralized critic in all three applications. Because the follower-response algorithm's internal state is hidden from the leader, we follow centralized training and decentralized execution (CTDE)~\citep[see, e.g.,][]{LoweWTHAM17}: the critic receives privileged simulator information, while the actor uses only the leader's game-observation history.

Repeatedly computing follower responses can make leader training expensive. To reduce this cost, we use meta-reinforcement learning~\citep{duan2016rl}: a contextual follower trains across leader policies and then responds using a context built from a compact set of policy queries. Once trained, this model provides another PI response algorithm without requiring follower retraining during leader optimization. Section~\ref{sec:scaling} develops this construction for a single follower.

Our three applications examine different challenges for learning leader commitments. In indirect mechanism design, the leader must allocate goods using reports from buyers with private values and strategic reporting incentives. We study Message Sequential Price Mechanisms (MSPMs), which use limited messages to choose the order and prices of offers. For platform design, we build on the algorithmic-collusion model of \citet{calvano2020artificial} and the demand-steering interventions of \citet{johnson2023platform}. We introduce hidden seller costs and study how a platform can improve consumer surplus by choosing which sellers to display as their pricing algorithms adapt. Finally, in an Atari task based on \textit{Space Invaders}~\citep{bellemare2013arcade}, agents trade ammunition, and follower responses require both visual gameplay and continuous economic decisions. This tests whether meta-learning can make leader training practical by reusing follower-response models across commitments. We complement these applications with controlled ablations of the framework's design choices (Appendix~\ref{app:results_extra}).

\section{Related Work}
\label{sec:related}

We organize related work by how follower responses enter the design problem. We first review commitment problems with explicit best-response or equilibrium models, progressing from single followers to multiple followers, private information, and dynamic games. We then discuss incentive design through agents' learning and adaptation. These strands overlap when learning is used to approximate equilibrium responses.

\subsection{Stackelberg Commitment and Equilibrium Responses}
\label{subsec:related:stackelberg}

The basic Stackelberg model features a leader and a single follower. The leader moves first, committing to a policy. The follower then observes the leader's commitment and plays a best response---if multiple best responses exist, the follower breaks ties either in the leader's favor (optimistic) or against the leader (pessimistic). Foundational work analyzes the value of commitment to mixed strategies in finite games~\citep{von2004leadership}. \citet{ConitzerS06} give a linear-programming approach for two-player normal-form games and complexity results for Bayesian games with private payoff information.

This basic, single-follower model has been widely applied in security domains. In Stackelberg security games, defenders act as leaders who commit to allocations of policing resources such as patrols and checkpoints, and attackers act as followers choosing targets. Algorithmic advances in this setting \citep{paruchuri2008playing} have powered real-world deployments such as airport security scheduling at LAX \citep{sinha2018stackelberg} and wildlife-protection patrol planning \citep{fang2016deploying}. In economics, Stackelberg models also capture carbon-tax design and fiscal competition~\citep{zhou2019would, wang1999commodity}.

Learning methods address uncertainty about payoffs and responses within these commitment problems. \citet{letchford2009learning} study one-shot finite-action games with unknown follower payoffs. They learn leader commitments by querying a follower best-response oracle, with query complexity depending on the sizes of the follower's best-response regions. In security games, \citet{blum2014learning} and \citet{peng2019learning} learn commitments from observed attacker responses when attacker payoffs are unknown. \citet{BaiJWX21} study noisy reward feedback in bandit and bandit-RL games. In the latter, the leader chooses one action that determines an MDP in which the follower acts sequentially. Our construction also allows the leader to act throughout the game, committing to a policy that conditions its actions on its observation history.

With multiple followers, a commitment induces a game among them, making their interaction part of the leader's problem. For one-shot games with known payoffs, computing an optimal commitment under pessimistic selection among pure follower Nash equilibria is NP-hard even with only two followers~\citep{coniglio2020computing}. Computational approaches impose follower-equilibrium constraints through mathematical programming~\citep{basilico2017methods} or differentiate through Nash equilibrium responses, including stochastic equilibrium selection~\citep{wang2022coordinating}. \citet{conitzer2011commitment} additionally allow the leader to commit to correlated recommendations, expanding the available commitment.

Equilibrium response models also support learning incentive rules across a family of problem instances. \citet{thoma2026deep} train neural networks that map each instance to incentives, such as contract payments, defining a normal-form game in which followers choose actions simultaneously. They optimize these networks by differentiating through pretrained neural approximations to selected maximum-entropy correlated or coarse correlated equilibria.

Within games with multiple followers, learning methods also address private information and uncertainty about follower payoffs. \citet{personnat2026learning} derive regret bounds for repeated rounds of Bayesian Stackelberg games with known utility functions and newly drawn private types. Followers best respond within each round, while the leader learns under an unknown type distribution using type or action feedback. Their main results assume that each follower's payoff is independent of the other followers' actions, although some type-feedback bounds extend to interacting followers.

\citet{maddux2025noregret} instead learn from realized leader costs without observing follower utilities. For each leader choice, followers learn an approximate Nash equilibrium of a static game; strong monotonicity ensures that the target equilibrium is unique. Their leader learns its cost by Gaussian-process regression, with a complexity bound on the class of cost functions being learned. Both studies thus learn commitments through repeated static games.

We next turn to decisions across multiple stages. Dynamic Stackelberg models allow actions to affect future states, with either one or multiple followers. \citet{zhong2021can} study general-sum Markov games in which the leader optimizes cumulative reward while followers maximize only immediate reward. At each state, followers play a Nash equilibrium of the stage game induced by the leader's policy. For two-player stochastic games with a known environment and computable follower best responses, \citet{kudo2025policy} give a policy-iteration method that monotonically improves the leader's value under those responses. Both studies assume full observability. Our construction also accommodates partial observability, with leader commitments that condition on private observation histories.

Learning algorithms for dynamic incentive design can also target best-response or equilibrium models. \citet{thoma2024contextual} study contextual MDP design for potentially multiple followers, estimating hypergradients from follower trajectories. They prove convergence to a stationary point under entropy-regularized follower optimization, controlled response approximation, smoothness, and access to derivatives of the configured model. \citet{zeng2025meanfield} give finite-time and finite-sample guarantees for actor--critic learning in Stackelberg mean field games with an infinite population of homogeneous followers. Their convergence to a stationary point of the Stackelberg objective relies on regularity and a gradient-alignment condition.

Equilibrium tracking addresses the additional difficulty that followers may continue learning as the leader changes its decisions. \citet{frankel2025finite} give finite-time guarantees using smooth, strongly monotone follower games and contracting stochastic response algorithms. Their guarantee for approximating a global Stackelberg equilibrium additionally assumes strong convexity of the leader's anticipated objective.

Our construction evaluates commitments against a specified policy-interactive response procedure that need not yield an equilibrium. It accommodates sequential decisions, interacting followers, and private observation histories without requiring a differentiable response.

\subsection{Incentive Design with Learning Agents}
\label{subsec:related:bilevel}

A complementary line of work learns incentives by modeling or influencing agents' adaptation. This problem forms part of the broader Social Environment Design agenda, which also considers collective choice over policy objectives~\citep{zhang2024social}.

The AI Economist~\citep{zheng2022ai} aims to maximize welfare subject to agents optimizing their own utility under the tax policy, and trains a planner and economic agents concurrently. Meta-gradient incentive design~\citep{yang2022adaptive} evaluates welfare using trajectories generated after a finite update of the current follower policies and differentiates through that update using follower parameters and gradients. These parameters continue evolving during training, and the updated policies need not constitute an equilibrium response. The meta-gradient accounts for response dependence without representing follower updates as actions in the leader's RL trajectory. Other methods learn models of worker behavior to guide task assignments and incentive payments~\citep{ShuT19, ShiYWWZLA20}.

The timing of reward evaluation distinguishes leader objectives. \citet{balaguer2022good} use nested optimization to evaluate mechanisms over follower learning trajectories, with evolutionary search when gradients are unavailable. Including transient rewards can favor commitments that are suboptimal under the post-response objective, even if the combined objective is optimized exactly. Appendix~\ref{sec:ap:limits:reward} examines the effect of including these rewards on learned policies. Valuing transient performance may be appropriate when the leader cares about outcomes while followers learn, for example when interacting with human participants.

Related incentive-design methods provide guarantees for different targets. \citet{ivanov2024principal} study sequential outcome-contingent contracts. For finite-horizon models with hidden agent actions, their principal--agent meta-algorithm finds a subgame-perfect equilibrium under exact policy optimization. Such an equilibrium excludes commitments to future actions that the principal would prefer to revise when reached, and can therefore yield less principal utility than an ex ante Stackelberg commitment. In their observed-action model, any subgame-perfect equilibrium is also a Stackelberg equilibrium. Equilibrium steering instead takes a target outcome as given. \citet{zhang2023steering} show that a non-player mediator with knowledge of the players' payoffs can steer no-regret learners toward a predetermined pure-strategy equilibrium in an extensive-form game with sublinear total payments. Their extension to general correlated targets also uses action recommendations.

Much of this literature develops methods for shaping agents' learning and evaluates the resulting incentives empirically. Our contribution complements this work with a reduction from commitment optimization to a single-agent POMDP. Compared with concurrent planner--agent training, our formulation evaluates commitments through independently initialized calls to a fixed follower-response procedure; unlike meta-gradient methods, it does not require differentiating through follower updates. A finite learning procedure can serve as a response oracle even when it does not reach equilibrium. We establish that optimal memoryless policies in the resulting POMDP correspond to optimal leader commitments under the specified procedure. This guarantee concerns the formulation; we evaluate its actor--critic optimization empirically.

\section{Preliminaries} \label{sec:preliminaries}

We model sequential interactions using \textit{partially observable Markov games} (POMGs)~\citep{HansenBZ04}. We first define a general game, observation histories and policies, and the game induced by fixing some agents' policies. We then specialize these objects to the Stackelberg problem, where a leader commits to a policy before its followers respond.
\begin{definition}[Partially Observable Markov game]
\label{def:pomg}
A {\em partially observable Markov game (POMG)} $\mathcal G$ on a finite, nonempty agent set $N$ is defined by the following components:
\begin{itemize}
    \item A set of {\em states}, $S$, with each state $s \in S$ containing the necessary information on game history to predict the game's progression based on agents' actions. 
    \item A set of {\em actions}, $A_i$, for each agent $i\in N$. An {\em action profile} is denoted by $a = (a_i)_{i\in N} \in A$, where $A = \prod_{i\in N} A_i$.
    \item A set of {\em observations}, $O_i$, for each agent $i\in N$. An {\em observation profile} is denoted by $o = (o_i)_{i\in N} \in O$, where $O = \prod_{i\in N} O_i$.
    \item An {\em initial distribution}, $\rho_0$, over $S\times O$, from which the initial state and observation profile $(s_0,o_0)$ are drawn.
    \item A {\em state transition function}, $\tran: S \times A \times S \to [0, 1]$, which defines the probability of transitioning from one state to another conditioned on the current state $s$ and  action profile $a$: $\tran(s, a, s') = \Prob(s'|s, a)$. 
    \item An {\em observation generation rule}, $\mathcal{O}: S \times A \times O \to [0,1]$, which gives the distribution of the next observation profile $o'$ conditional on the reached state $s'$ and preceding action profile $a$: $\mathcal{O}(s', a, o') = \Prob(o'|s', a)$.
    \item A {\em reward function}, $R_i: S \times A \to \mathbb{R}$, which maps  for each agent $i\in N$ the current state $s$ and  action profile $a$  to the immediate reward received by the agent, $R_i(s, a)$.
    \item A finite time horizon, $T\geq 1$, indicating that the game proceeds for $T$ steps. At the terminal step $T$, no further actions are taken and, following standard notation, no reward is received.
\end{itemize}
At each step $t=0,\ldots,T-1$, each agent $i$ chooses $a_{i,t}$ using its observations through time $t$ and receives reward $R_i(s_t,a_t)$. The game then draws $s_{t+1}\sim\tran(s_t,a_t,\cdot)$, followed by $o_{t+1}\sim\mathcal O(s_{t+1},a_t,\cdot)$; agent $i$ observes $o_{i,t+1}$.
\end{definition}

When $|N|=1$, this formulation is a single-agent POMDP.

We define an \textit{episode} as a complete run of the game up to horizon $T$, during which each agent seeks to maximize its expected cumulative reward. Let $o_{i,t}$ denote agent $i$’s observation at time $t$, and let the \textit{observation history} up to time $t$ be $h_{i,t} = \bigl(o_{i,0},\, o_{i,1},\, \ldots,\, o_{i,t}\bigr)$. We write $H_i$ for the \emph{set of all possible observation histories} that agent $i$ can experience in this game, over all possible plays and times $t=0,\ldots,T$.

Each agent $i\in N$ follows a stochastic policy $\pi_i$, with
$$
a_{i,t} \sim \pi_i (\cdot | h_{i,t}), 
\qquad 
\pi_i : H_i \to \Delta({A}_i),
$$
Here and throughout, $\Delta(X)$ denotes the set of probability distributions over $X$.\footnote{This formulation accommodates memory-based policies \citep[e.g.][]{hausknecht2015drqn, heess2015memory} as well as policies defined over histories of observations that are compressed into sufficient statistics \citep[e.g.][]{smallwood1973optimal, lovejoy1991computational}.}
Let $\Pi_i$  denote the set of policies for agent $i$.
Fixing the policies of some agents induces a game for the remaining agents:
\begin{definition}[Subgame of a POMG, informal]
\label{def:inducedGame}
    Given a POMG $\mathcal{G}$ on agent set $N$, a proper subset $I\subsetneq N$, and a policy profile $\pi_I=\{\pi_i\}_{i\in I}$ for agents in $I$, we define the {\em subgame POMG}, $\mathcal{G}_{\pi_I}$, as the game induced among agents $N\setminus I$ by fixing policy profile $\pi_I$. Actions, rewards, and observations remain the same as in $\mathcal G$, while states are augmented to include observation histories $h_I=(h_i)_{i\in I}$.
    State transitions in $\mathcal G_{\pi_I}$ are induced by the state transition function in $\mathcal{G}$, with the actions of the agents in $I$ in $\mathcal G$ computed ``internally'' by the game $\mathcal{G}_{\pi_I}$ by applying policies $\pi_I$ on observation histories $h_I$.
\end{definition}

Appendix~\ref{app:subgame} gives the formal definition.
For the Stackelberg problem, let $N=\{\ell,1,\ldots,n\}$, with leader $\ell$ and $n\geq1$ followers. We write $N_{-\ell}=N\setminus\{\ell\}$ for the follower set. Fixing the leader's policy $\pi_\ell$ corresponds to $I=\{\ell\}$ and yields the follower game $\mathcal G_{\pi_\ell}$. Fixing all follower policies corresponds to $I=N_{-\ell}$ and leaves a POMDP for the leader. We now add the leader's commitment stage:
\begin{definition}[Stackelberg Partially Observable Markov Game]
\label{def:BSGame}
    Let $\mathcal{G}$ be a POMG with agent set $N=\{\ell,1,\ldots,n\}$. The \textit{Stackelberg POMG}, $\mathcal{S}_\mathcal{G}$, is defined by the following interaction:
    \begin{itemize}
        \item \textbf{Leader's Commitment}: The leader selects a policy $\pi_\ell$ in  game $\mathcal{G}$.
        \item \textbf{Induced Subgame}: The followers play the subgame POMG, $\mathcal{G}_{\pi_\ell}$, induced by the leader's policy. The result is a state-action trajectory, $\tau = (s_0, a_0, \ldots, s_{T-1}, a_{T-1}, s_{T})$, in $\mathcal{G}$, for time horizon $T\geq 1$.
        \item \textbf{Leader's Reward}: The leader receives the reward, $R_\ell(\tau) = \sum_{t=0}^{T-1} R_\ell(s_t, a_t)$, based on  trajectory $\tau$, where $R_\ell$ is the leader's reward in  game $\mathcal{G}$.
    \end{itemize}
\end{definition}

\section{The Leader's Problem}\label{sec:leaderProblem}
In a Stackelberg game, the leader optimizes its reward by anticipating how followers will respond to its strategy. In Stackelberg POMGs, followers may exhibit a wide range of response behaviors. These include various equilibrium strategies—such as Nash equilibria \citep{fudenberg1991game}, correlated equilibria \citep{aumann1974subjectivity}, and coarse correlated equilibria \citep{roughgarden2015intrinsic}—which correspond to {\em Stackelberg-Nash, Stackelberg-CE}, and {\em Stackelberg-CCE} solution concepts, respectively. They may also include non-equilibrium behaviors, such as the seemingly collusive behaviors observed in simulation by \citet{calvano2020artificial} and \citet{abada2023artificial}. To capture a range of possible follower behaviors, we make use of the following oracle abstraction: 
\begin{definition}[Follower–response oracle] \label{def:oracle}
Let $\Pi_{-\ell}=\prod_{i\in N_{-\ell}}\Pi_i$ denote the space of follower policy profiles. A \emph{follower-response oracle} assigns to each leader policy $\pi_\ell$ a distribution over response profiles in the induced game $\mathcal G_{\pi_\ell}$:
$$
\pi^*_{-\ell}\sim\mathcal B(\cdot\mid\pi_\ell),
\qquad
\mathcal B:\Pi_\ell\to\Delta(\Pi_{-\ell}).
$$
\end{definition}

Definition~\ref{def:oracle} allows random selection of a follower policy profile, in addition to any action randomization within the selected policies. Responses need not be equilibria. We use this oracle to state the leader's decision problem.
For this, let $(\tau_\ell, \tau^*_{-\ell}) = (s_0, (a_{\ell,0},a^*_{-\ell,0}), \ldots, s_{T-1}, (a_{\ell,T-1},a^*_{-\ell,T-1}),s_{T})$
be a {\em state-action trajectory} in POMG $\mathcal G$. Let $\Prob(\cdot\mid\pi_\ell,\mathcal B)$ denote its distribution when a profile $\pi^*_{-\ell}\sim\mathcal B(\cdot\mid\pi_\ell)$ is first drawn and then held fixed during a fresh play of $\mathcal G$ against $\pi_\ell$. This distribution averages over both response selection and gameplay randomness. We formulate the leader's problem:
\begin{definition} [Optimal leader strategy]
\label{def:leaderProblem}
    Let $\mathcal{S}_\mathcal{G}$ be a Stackelberg POMG defined on game $\mathcal{G}$ and $\mathcal{B}$ the follower–response oracle. The leader’s decision problem is:
    \begin{equation*}
        \pi^*_\ell \in \underset{\pi_\ell\in \Pi_\ell}{\arg \max } \underset{(\tau_\ell, \tau^*_{-\ell}) \sim 
        \Prob(\cdot| \pi_\ell, \mathcal{B})
        }{E}\left[\sum_{t=0}^{T-1} R_\ell\left(s_t, (a_{\ell, t}, a^*_{-\ell, t})\right)\right].
    \end{equation*}
    The pair $\bigl(\pi^*_\ell,\,\mathcal B(\cdot\mid\pi_\ell^*)\bigr)$ is a \emph{behavioral Stackelberg equilibrium} (under behavior $\mathcal B$) of Stackelberg POMG $\mathcal S_\mathcal G$. \end{definition}

In words, the leader objective is to choose a strategy that maximizes its expected cumulative reward, accounting for the fact that followers respond according to $\mathcal B(\cdot\mid\pi_\ell)$.

A possibly randomized algorithm---e.g., reinforcement learning, best-response optimization, or equilibrium search---implements $\mathcal B$ if, for every input policy $\pi_\ell$, its output profile has distribution $\mathcal B(\cdot\mid\pi_\ell)$. Different oracle implementations may ``access'' leader policy $\pi_\ell$ in different ways, such as full gradients, roll-outs, or black-box evaluations. Our main positive result (Theorem~\ref{thm:main}) will rely on an oracle implementation that treats the leader policy as a black box: these algorithms may query the leader's policy on any leader observation history, $h_\ell$, and observe the action it returns, but cannot  inspect or know the policy’s ``internals,'' such as its parameterization, functional form, or the procedure by which actions are generated.
\begin{definition}[Policy-Interactive Response Algorithm] \label{def:PIAlgorithm}
Given a Stackelberg POMG $\mathcal S_\mathcal G$ on game $\mathcal G$, a \emph{policy-interactive (PI) response algorithm} $\eqalg$ accesses a supplied leader policy $\pi_\ell$ exclusively through queries and returns a follower policy profile for $\mathcal G_{\pi_\ell}$. Each query specifies a leader observation history $h_\ell$ and receives a fresh action sample $a_\ell\sim\pi_\ell(\cdot\mid h_\ell)$. The algorithm has the following components, fixed independently of the supplied policy:
\begin{itemize}
    \item A set of \emph{algorithm states} $S_b$, where $s_b=(\pi_{-\ell},h_\ell,s_b^+)$ records a current follower policy profile, the next query history, and auxiliary information used by the algorithm.
    \item An \emph{initial-state distribution} $\rho_b$, from which $s_{b,0}$ is drawn independently at the start of each oracle call.
    \item A \emph{state transition function}, $\tran_b(s_b,a_\ell,s_b')=\Prob(s_b'\mid s_b,a_\ell)$, which updates the algorithm state after a query answer.
    \item A finite \emph{query horizon} $E$, bounding the number of queries before the algorithm returns. If it returns earlier, we pad to $E$ with queries at a fixed history whose answers are ignored; the returned profile is held fixed.
\end{itemize}
The final state $s_{b,E}$ contains the returned response profile $\pi^*_{-\ell}=(\pi^*_1,\ldots,\pi^*_n)$. Its distribution, induced by initialization, query samples, and state transitions, defines $\mathcal B(\cdot\mid\pi_\ell)$.
\end{definition}

We next illustrate PI response algorithms with a simple trial-and-error example. Further instantiations using no-regret learning, Q-learning, and contextual meta-followers appear in Section~\ref{sec:scaling} and Appendices~\ref{app:mw_alg} and~\ref{app:environments}.

\begin{example}
Consider a single-follower setting where the follower uses a simple trial-and-error procedure: it tests $K$ candidate policies $\pi_f^{(1)}, \pi_f^{(2)}, \ldots, \pi_f^{(K)}$ against the leader’s strategy $\pi_\ell$ in $\mathcal{G}_{\pi_\ell}$ and selects the one that yields the highest reward. This can be expressed as a PI response algorithm:
\begin{itemize}
    \item {Algorithm States:} A state $s_b$ records the current candidate index $k\in\{1,\ldots,K\}$, the leader history $h_\ell$, and auxiliary information $s_b^+$ containing the follower history $h_f$, the current accumulated reward, and the best reward and its policy index from completed trials. During a trial, $\pi_{-\ell}=\pi_f^{(k)}$.

    \item {Algorithm State Transitions:} Each game step updates the histories and accumulated follower reward using the queried leader action and the follower action sampled from $\pi_f^{(k)}(\cdot\mid h_f)$. At the end of each trial, the algorithm retains the highest total reward so far and its policy index, keeping the first maximizer on ties. If $k<K$, it increments $k$, resets the accumulated reward to zero, and starts a fresh trial; otherwise, it returns the best policy.

    \item {Query Horizon:} The bound is $E=K\cdot T$, since each of the $K$ policies is tested for at most $T$ steps.
\end{itemize}
\end{example}

\section{Learning an Optimal Leader Policy} 
\label{sec:learningPolicies}

We now embed a PI response algorithm into the leader's environment to construct the \textit{Stackelberg POMDP}, shown in Figure~\ref{fig:stackpomdp}. We restrict the leader to the same mapping from game histories to action distributions during response queries and gameplay, and evaluate its return in a fresh play against the resulting follower profile. We show that this formulation preserves the value of every leader commitment.

\begin{figure*}[!t]
    \centering
    \begin{tikzpicture}[
    every node/.style={font=\small},
    state/.style={circle, draw, minimum size=1.6cm},
    box/.style={rectangle, draw, minimum width=1.6cm, minimum height=1.6cm},
    >={Stealth[round]},
    node distance=1.6cm and 1.8cm
    ]
    
        \node[state] (sb0) {$s_{\mathcal S, 0} = s_{b,0}$};
        \node[state, above=.4cm of sb0] (o00) {$o_{\mathcal S, 0} = h_{\ell,0}$};
        \node[state, right=.8cm of o00] (a00) {$a_{\mathcal S,0} = a_{\ell,0}$};
        \node[state, right=2.4cm of sb0] (sb1) {$s_{\mathcal S, 1} = s_{b,1}$};
        \node[state, right=2.4cm of sb1] (sbe) {$s_{\mathcal S, E} = s_{b,E}$};
        \node at ($(sb1)!0.5!(sbe)$) (dots) {$\cdots$};
        \node at ($(sb1 |- o00)$) (r1) {$R_1 = 0$};
        \node at ($(sbe |- o00)$) (re) {$R_E = 0$};
        
        \node[box, right=2.4cm of sbe] (G) {$\mathcal{G}$};
        \node at ($(G |- o00)$) (ret1) {$R_\ell(\tau)$};
        \node at ($ (G) - ($ (sb1) - (a00) $) $) (policy) {$\pi_{\mathcal S}$};
        
        \draw[->] (sb0) -- (o00);
        \draw[->] (o00) -- (a00)node[midway,above] {$\pi_{\mathcal S}$};
        \draw[->] (a00) -- (sb1);
        \draw[->] (sb0) -- (sb1);
        \draw[->] (a00) -- (sb1);
        \draw[->] (sb1) -- (dots);
        \draw[->] (dots) -- (sbe);
        \draw[->] (sb1) -- (r1);
        \draw[->] (sbe) -- (re);
        \draw[->] (G) -- (ret1);
        \draw[->] (policy) -- (G);
        \draw[->] (sbe) -- (G) node[midway,above] {$\pi^*_{-\ell}$};

        \draw [decorate,decoration={brace,amplitude=8pt,mirror},thick]
          ($(sb0.south) - (0,4pt)$) -- ($(sbe.south) - (0,4pt)$) node[midway,below,yshift=-8pt] {follower response phase};
        
        \draw [decorate,decoration={brace,amplitude=8pt,mirror},thick]
       ($(sbe.south) - (0,4pt)$) -- ($(G.south) - (0,4pt)$) node[midway,below,yshift=-8pt] {reward phase};
    
    \end{tikzpicture}
    \caption{A Stackelberg POMDP episode. In the \textit{follower response phase}, algorithm $\eqalg$ is simulated for $E$ steps. At each step $e$, $\eqalg$ issues a query (a leader observation history $h_{\ell,e}$) and receives the leader's response $a_{\ell,e} \sim \pi_{\mathcal{S}}(h_{\ell,e})$, updating its state from $s_{b,e}$ to $s_{b,e+1}$. By step $E$, the algorithm state $s_{b,E}$ contains the follower response profile $\pi^*_{-\ell}$. In the \textit{reward phase}, the leader plays the game $\mathcal{G}$ against these follower policies, yielding total return $R_\ell(\tau)$.}
    \Description{Flow diagram of a Stackelberg POMDP episode: zero-reward follower-response queries update the response-algorithm state, followed by a reward game against the resulting follower policy.}
    \label{fig:stackpomdp}
\end{figure*}

\begin{definition}[Stackelberg POMDP]\label{def:StackPOMDP} Given a Stackelberg POMG, $\mathcal{S}_\mathcal{G}$ and a PI response algorithm $\eqalg$ implementing the follower-response oracle $\mathcal B$, the {\em Stackelberg POMDP} is constructed as a finite-horizon, episodic POMDP {that models the leader's problem}, where an episode of the Stackelberg POMDP is divided into two phases:
\begin{enumerate}
    \item A {\bfseries follower response phase} of $E$ query steps in which algorithm $\eqalg$ is initialized with $s_{b,0}\sim\rho_b$ and simulated inside the leader's environment (see Figure~\ref{fig:stackpomdp}). At each step $e=0,\ldots,E-1$, the state is $s_{\mathcal{S}, e} = s_{b,e}$ (the state of $\eqalg$), the observation $o_{\mathcal{S}, e}$ is the $e$-th query to the leader, the action $a_{\mathcal{S}, e}$ is the leader’s response to that query, and the state transition follows the update rule of $\eqalg$. By step $E$, the algorithm state $s_{b,E}$ contains the follower response profile $\pi^*_{-\ell} = (\pi^*_1, \ldots, \pi^*_n)$, whose distribution is $\mathcal B(\cdot\mid\pi_\ell)$ under the corresponding leader policy. This sampled profile governs follower behavior in the next phase. No reward is given during this phase.
    \item A {\bfseries reward phase} of $T$ steps, corresponding to a single play of $\mathcal{G}_{\pi^*_{-\ell}}$, where followers use the sampled profile $\pi^*_{-\ell}$ from the final state $s_{b,E}$. The hidden state retains this profile, the game state, and the agents' game-observation histories. The game is initialized afresh, independently of the response phase. Conditional on the retained profile, transitions and rewards follow $\mathcal G_{\pi^*_{-\ell}}$. Note that the POMDP observation $o_{\mathcal{S}, e}$ provided to the leader at each step of this phase is precisely the leader's game observation history $h_\ell$ in $\mathcal{G}$ up to the current reward-phase step. This phase reduces to a POMDP, as it involves only a single agent (the leader).
\end{enumerate} 
\end{definition}
\noindent
We let $\Pi_{\mathcal{S}}$ denote the set of policies in the Stackelberg POMDP: each $\pi_{\mathcal{S}} \in \Pi_{\mathcal{S}}$ maps a POMDP observation history $h_{\mathcal{S}, t} = (o_{\mathcal{S}, 0}, \ldots, o_{\mathcal{S}, t})$ to a distribution over POMDP actions. 
Appendix~\ref{app:environments} details the follower-response procedures and Stackelberg POMDP constructions for each application.

As for POMGs, the goal in the Stackelberg POMDP is to find a policy that maximizes the expected cumulative reward. We restrict attention to \emph{memoryless} leader policies---those that depend only on the current POMDP observation. This observation is itself a leader observation history in the original game. Thus, the restriction preserves dependence on the leader's game history while excluding additional information from the sequence of follower-response queries.

\begin{definition}[Memoryless leader policies]\label{def:memoryless}
Let $\Pi_{\mathcal S}$ be the set of leader policies in the Stackelberg POMDP (defined above), and let
$$
h_{\mathcal S,t}=(o_{\mathcal S,0},\ldots,o_{\mathcal S,t})
$$
denote the leader’s observation history up to time~$t$.
We define the memoryless subclass
$$
\Pi_{\mathcal S}^{0} \;=\;
\Bigl\{
\pi_{\mathcal S}\in\Pi_{\mathcal S}\;:\;
\forall\,t_1,t_2,\;
o_{\mathcal S,t_1}=o_{\mathcal S,t_2}
\;\Longrightarrow\;
\pi_{\mathcal S}(\,\cdot\mid h_{\mathcal S,t_1})
\,=\,
\pi_{\mathcal S}(\,\cdot\mid h_{\mathcal S,t_2})
\Bigr\}.
$$
That is, $\pi_{\mathcal S}$ is memoryless if any two histories ending in the
\emph{same} observation induce the \emph{same} probability distribution over leader actions.
\end{definition}
\noindent
For memoryless policies, we write $\pi_{\mathcal S}(\cdot\mid o_{\mathcal S})$, since the action distribution depends only on the current POMDP observation.
The policy correspondence preserves expected returns and therefore optimal commitments.

\begin{theorem} \label{thm:main}
    Every leader policy corresponds to a memoryless Stackelberg POMDP policy with the same expected return, and conversely. Consequently, an optimal memoryless policy yields an optimal leader commitment under $\mathcal B$.
\end{theorem}
\begin{proof}
    The policy classes correspond in both directions through
    $$\pi_{\mathcal S}(\cdot\mid h_\ell)=\pi_\ell(\cdot\mid h_\ell),\qquad h_\ell\in H_\ell,$$
    because the current POMDP observation is a leader game-observation history. Each memoryless policy therefore specifies a leader commitment, and every leader commitment defines a memoryless POMDP policy.

    Fix corresponding policies. The response phase uses the same initialization, conditional action distributions, and state transitions as $\eqalg$. Induction over its $E$ steps shows that its final follower profile has law $\mathcal B(\cdot\mid\pi_\ell)$. Conditional on that profile, the reward phase has the same trajectory law as a fresh play of $\mathcal G$ against those follower policies. The memoryless restriction ensures that the leader uses the same $\pi_\ell$ in both phases. Response-phase rewards are zero, so averaging the conditional gameplay return over the response distribution gives exactly the objective in Definition~\ref{def:leaderProblem}, by the law of total expectation. Thus every corresponding policy pair has the same expected return. Since the policy correspondence works in both directions, maximizing over memoryless Stackelberg-POMDP policies is equivalent to optimizing the leader's commitment.
\end{proof}

Note that optimal policies in POMDPs generally require access to the full observation history, and memoryless policies can lead to suboptimal overall reward. However, in our setting it is crucial that the policy behaves consistently across both the follower response phase and the reward phase. If the policy were allowed to differ between these phases, the leader could ``trick'' followers into adopting a certain behavior during the response phase and then switch to a different policy to maximize its own rewards. However, such behavior would not constitute a response to the leader policy actually used in the final game, thereby violating the definition of a Stackelberg equilibrium. We demonstrate this failure mode experimentally in Appendix~\ref{sec:ap:limits:noninvariant}. 

The PI interface separates the choice of follower-response algorithm from the method used to optimize the leader's policy. To optimize over memoryless policies, we use PPO with a centralized critic in a centralized training and decentralized execution (CTDE) architecture \citep[see, e.g.,][]{LoweWTHAM17}. During training, the critic can access simulator state or privileged state features, while the actor relies only on the current POMDP observation, preserving the memoryless restriction. The augmented state includes the follower-response algorithm's internal state, together with the phase and time index. For a fixed response algorithm, this gives stationary Markov transitions even while followers learn, as their adaptation is part of the state evolution. A critic that observes this state can distinguish learning stages that produce the same leader observation. The actor remains partially informed, as required for a valid commitment. Each application specifies the full state or selected privileged features supplied to its critic.

\paragraph{Response-phase credit assignment.}
Computing follower responses can require many leader-policy queries, creating a long sequence of zero-reward transitions before gameplay begins. When followers adapt through reinforcement learning, repeating this response computation for each candidate leader policy can dominate the training cost, especially in high-dimensional environments. Because response queries receive no immediate leader reward, it may seem natural to update the leader using only gameplay transitions. However, actions taken during these queries shape the follower's eventual response and therefore the leader's payoff. The following result shows that ignoring this dependence can make the expected policy-gradient update point in a direction that decreases the leader's expected payoff.

To isolate this credit-assignment issue, we analyze REINFORCE, a canonical policy-gradient method whose full-trajectory estimator is unbiased \citep{williams1992simple}. This makes it a natural reference for identifying the bias introduced by gameplay-only updates.

\begin{proposition}[Response-phase credit assignment, informal]
\label{prop:response_credit}
Actions taken during the response phase can influence the final payoff by shaping follower behavior. A REINFORCE update based only on gameplay can overlook this influence and may, even on average, point in a direction that decreases the leader's Stackelberg objective.
\end{proposition}

Appendix~\ref{app:response_credit} gives the full statement, proof, and a one-query example using the likelihood-ratio identity \citep{schulman2015stochastic}. Appendix~\ref{sec:response_credit_ablation} reports a controlled ablation of response-query credit using this example.

Having established why response queries matter for learning, we next examine how the choice of commitment class can reduce the number of policy calls.

\paragraph{Deterministic commitments and cached queries.}
Restricting the leader to deterministic policies may reduce its optimal expected return, but permits reuse of answers to repeated policy queries. Any PI response algorithm admits a cached PI implementation that stores query answers in its auxiliary state and handles repeated requests internally. For deterministic leader policies, this preserves the distribution of follower responses while requiring at most one policy evaluation per distinct observation history encountered. The policy correspondence of Theorem~\ref{thm:main} also holds within this restricted class. During training, a stochastic actor can explore deterministic commitments by sampling an action when each history is first queried and reusing it throughout the response and reward phases. The cache is cleared between episodes. Appendix~\ref{app:cached_commitments} formalizes this construction.

The next section develops another way to shorten response computation: contextual meta-followers trained across leader policies.

\section{Scaling to High-Dimensional Environments}
\label{sec:scaling}

In many practical settings, the family of subgames $\{\mathcal{G}_{\pi_\ell}\}_{\pi_\ell \in \Pi_\ell}$ shares enough structure for followers to generalize across leader policies. While each subgame's dynamics and rewards depend on the leader's policy at hand, the state and action spaces remain fixed across subgames, and the dependence on $\pi_\ell$ is typically structured (e.g., the leader's policy parameterizes a payment, threshold, or selection rule). This suggests that followers trained on a \textit{distribution} of leader policies could adapt their responses through query-based conditioning, amortizing the cost of follower learning across commitments.

In this section, we use meta-reinforcement learning to construct a PI response algorithm for a single follower, denoted $f$. A learned follower-response model conditions its behavior on answers to leader-policy queries. Constructing such an oracle for strategically interacting followers, whose responses depend on each other, raises additional questions in context design and meta-training that we leave to future work.

We adopt the framework of \textit{contextual policies}~\citep{rakelly2019efficient}, a meta-RL approach where task-specific information is encoded in a context vector $\omega\in \Omega$ that conditions the agent's behavior.

\begin{definition}
We define a \textit{contextual follower policy} as a mapping:
\begin{equation}
    \pi_f^{\text{meta}}: O_f \times \Omega \to \Delta(A_f),
\end{equation}
where $\Omega$ is a context space encoding information about the leader's strategy. 
    
\end{definition}

The key insight is that we can construct the context $\omega \in \Omega$ by \textit{querying the leader policy}, thereby satisfying the requirements of a policy-interactive response algorithm.

Fix $k$ query histories $h_\ell^{(1)},\ldots,h_\ell^{(k)}\in H_\ell$, where $H_\ell$ is the set of possible leader observation histories in the original game. We construct the context from the leader's responses to these queries:
\begin{equation}
    \omega(\pi_\ell) = \left( a_\ell^{(1)}, \ldots, a_\ell^{(k)} \right), \quad \text{where } a_\ell^{(j)} \sim \pi_\ell(\cdot \mid h_\ell^{(j)}).
\end{equation}
The notation $\omega(\pi_\ell)$ denotes a random context induced by fresh query samples. In words, $\omega$ acts as a low-dimensional summary of the leader's policy. While a small set of queries cannot fully identify an arbitrary leader policy, in many practical settings (such as pricing or resource allocation), the leader's strategic intent can be effectively approximated by observing their behavior in a few key situations.

\begin{proposition}
\label{prop:meta_rl_pi}
For a fixed contextual follower policy and a fixed list of query histories, query-based context construction constitutes a valid PI response algorithm (Definition~\ref{def:PIAlgorithm}) with algorithm time horizon $E=k$.
\end{proposition}

\begin{proof}
Initialize the query index and answer list to fixed empty values at each call. The state contains this index and list, the next query history, and a follower-policy component, using an arbitrary fixed policy until all answers are available. Each transition records a fresh leader action sample and advances to the next query. After $k$ queries, the follower-policy component is set to $\pi_f^{\text{meta}}(\cdot\mid\cdot,\omega)$; as a policy on histories, it uses the current follower observation. A fixed placeholder history completes the terminal state. The random query answers induce a distribution over these output policies, defining $\mathcal B(\cdot\mid\pi_\ell)$. Thus the procedure satisfies Definition~\ref{def:PIAlgorithm}.
\end{proof}

Since the Meta-RL oracle satisfies Definition~\ref{def:PIAlgorithm}, Theorem~\ref{thm:main} applies: an optimal memoryless policy in the resulting Stackelberg POMDP solves the leader's problem under the response distribution induced by the fixed contextual policy and query procedure.

Training proceeds in two stages. First, the contextual follower $\pi_f^{\text{meta}}$ is trained to maximize expected follower reward across leader policies sampled from a distribution over $\Pi_\ell$, with contexts constructed by querying each sampled policy. The resulting model is intended to approximate best responses over that training distribution. We then freeze this model and train the leader in the resulting Stackelberg POMDP.

These training stages are distinct from the two phases within each leader-training episode. Each episode contains a response phase of $k$ queries to the leader policy, followed by gameplay against the frozen meta-follower conditioned on the resulting context. This eliminates follower retraining within those episodes. Domain-specific response algorithms and training details appear alongside their environments in Appendix~\ref{app:environments}.

\section{Applications of Stackelberg POMDP}\label{sec:experiments}

We evaluate the framework in three applications where a leader must account for how other agents respond to its decisions. We start with indirect mechanism design, where the leader allocates goods using limited reports about agents' private preferences. The goal is to maximize the total value of the allocated goods to the receiving agents---a quantity called \emph{social welfare}. 

We then study online-platform design with adaptive pricing sellers and hidden costs, building on work on algorithmic collusion and demand steering~\citep{calvano2020artificial,johnson2023platform}. The platform learns which sellers to display from their posted prices. It aims to mitigate algorithmic collusion and improve consumer surplus---the expected utility consumers obtain from their available choices, net of prices.

Last, we consider Atari settings, where agents trade ammunition in a highly dynamic visual environment, learning how much it is worth while developing complex policies for using it effectively. The leader's payoff combines game rewards with money received from selling ammunition or spent on buying it. This application tests whether meta-learned follower responses make the framework practical in such environments. Table~\ref{tab:application_overview} summarizes the leader's choices and goals, and how followers respond across these applications.

Across all three applications, we train leaders using PPO with a centralized critic, as described in Section~\ref{sec:learningPolicies}. Appendix~\ref{app:environments} provides the shared training setup and application-specific architectures, hyperparameters, and evaluation protocols.

Comparisons aggregate independent runs of leader training; captions specify averaging and uncertainty. Appendix~\ref{app:statistical_comparisons} collects statistical tests for the main comparisons and controlled ablations. The mechanism-design and platform experiments use deterministic commitments with cached queries, as described in Section~\ref{sec:learningPolicies}. Our code and reproducibility materials are publicly available in the StackelbergPOMDP repository~\citep{GlcBrero_StackelbergPOMDP}.

\begin{table}[!t]
\centering
\caption{Leader choices, follower responses, and goals in the three application families.}
\label{tab:application_overview}
\small
\setlength{\tabcolsep}{4pt}
\renewcommand{\arraystretch}{1.3}
\begin{tabular}{@{}>{\raggedright\arraybackslash}p{0.17\linewidth}>{\raggedright\arraybackslash}p{0.24\linewidth}>{\raggedright\arraybackslash}p{0.18\linewidth}>{\raggedright\arraybackslash}p{\dimexpr0.41\linewidth-6\tabcolsep\relax}@{}}
\toprule
Application & Leader's choices & How followers respond & Leader's goal \\
\midrule
Indirect mechanism design & Allocation and pricing rules & Multiplicative weights & Maximize social welfare \\
\addlinespace
Platform intervention & Price-contingent display & Q-learning & Maximize consumer surplus by mitigating algorithmic collusion \\
\addlinespace
Atari bilateral trade & How to play and trade ammunition & Contextual meta-follower & Maximize the leader's payoff from gameplay and trade \\
\bottomrule
\end{tabular}
\Description{A table mapping the three application families to leader choices, follower responses, and goals. The mechanism-design row covers MSPMs.}
\end{table}

\subsection{Indirect Mechanism Design with No-Regret Followers}
\label{subsec:indirect_mechanism_noregret}

We study \emph{message sequential price mechanisms (MSPMs)}, which bring together several central features of indirect mechanism design. Buyers have private values and communicate through a limited set of messages; the designer chooses prices and the order in which buyers receive offers. An initial messaging round is followed by sequential offers and acceptance decisions, which determine allocation and payments. These features make MSPMs a compact setting for learning a mechanism as a policy over successive rounds of interaction while accounting for strategic reporting. Our objective is to improve social welfare. To explain what communication adds, we first describe the benchmark without messages.

Our benchmark is a \emph{sequential price mechanism (SPM)} for a single item; SPMs more generally can allocate multiple items~\citep{brero2021reinforcement}. In our setting, the mechanism visits buyers in an order of its choosing and makes each a take-it-or-leave-it offer. Buyer $i$ privately observes its value $v_i$ and obtains utility $v_i-p_i$ by accepting price $p_i$ and zero by rejecting. For these single-item mechanisms, the same welfare definition gives the receiving buyer's value, or zero if the item remains unsold. Prices affect which buyer receives the item; payments themselves are transfers and do not enter welfare. Acceptance ends the mechanism; rejection allows a visit to another unvisited buyer. Because each buyer is visited at most once, it cannot obtain a better offer by rejecting. It therefore maximizes its payoff by accepting when $v_i>p_i$ and rejecting otherwise, regardless of the other buyers' behavior; we break ties by rejection. This strategic simplicity motivates sequential offers \citep{sandholm2006sequences} and learned SPMs \citep{brero2021reinforcement}. The mechanism's limitation is that it observes no signal about realized values before making its first offer.

We make the SPM benchmark concrete using a scaled instance of Example~1 of \citet{AgrawalSZ20}, an optimal-stopping problem concerning the choice of observation order.

\begin{example}\label{ex:mspm-two-type}
Buyer~1's value is $0.2$ with probability $0.8$ and $1$ otherwise; Buyer~2's value is $0$ or $0.4$ with equal probability. If the mechanism observed both realized values, it would assign the item to a highest-value buyer; this ex post efficient allocation has expected welfare $0.440$. An optimal SPM instead offers Buyer~2 price $0.2$, then Buyer~1 price zero after a rejection. It is inefficient only at $(v_1,v_2)=(1,0.4)$, where Buyer~2 accepts and welfare is $0.4$ rather than $1$. Since this profile has probability $0.1$, expected welfare is $0.380$.
\end{example}

An MSPM relaxes the SPM information constraint by adding a messaging round before the offers. Buyers observe their values and send labels $\mu_i\in\{1,\ldots,m\}$; the mechanism observes the joint label profile, but not the values, and conditions its visit order and prices on it. Trade then proceeds as in an SPM. This acceptance rule remains optimal, but reporting is strategic because a label can change when a buyer is visited and the price it receives.

\begin{continuedexample}{ex:mspm-two-type}
Give each buyer two labels, $L$ and $H$, with low values reporting $L$ and high values reporting $H$. An optimal MSPM offers Buyer~1 price zero whenever Buyer~2 reports $L$. When Buyer~2 reports $H$, it offers Buyer~2 price $0.2$ if Buyer~1 reports $L$ and Buyer~1 price $0.4$ if Buyer~1 reports $H$. After any rejection, it offers the other buyer price zero. These reports allocate to a highest-value buyer at every profile and attain expected welfare $0.440$.
\end{continuedexample}

The remaining question is how strategic buyers would use these labels. We model their reporting behavior with the multiplicative-weights (MW) algorithm, defined in Appendix~\ref{app:mw_alg}. For each candidate MSPM, MW updates each buyer type's distribution over labels. With vanishing regret for every buyer type, empirical play approaches the set of \emph{Bayesian coarse correlated equilibria (BCCE)} \citep{hartline2015no}: no buyer type can improve its expected utility by committing, before seeing its recommendation, to a different fixed label (Definition~\ref{def:bcce}).

\begin{samepage}
\begin{continuedexample}{ex:mspm-two-type}
Compare each type's expected utility under its prescribed label with its utility from switching. Buyer~1's low type obtains $0.10$ from either label, while its high type obtains $0.80$ from $H$ and $0.50$ from $L$. Buyer~2's low type obtains zero from either label, while its high type obtains $0.16$ from $H$ and zero from $L$. No type has a profitable fixed-label deviation, so the low--high reporting strategy is a BCCE.
\end{continuedexample}
\end{samepage}

The response in Example~\ref{ex:mspm-two-type} is stronger than a BCCE: at every realized value profile, neither buyer gains by switching labels. It is therefore an ex post equilibrium, and hence a pure Bayes--Nash equilibrium. We often observe these stronger conditions, but certify BCCE because that is what no-regret convergence guarantees.

In finite runs, we use an approximate BCCE certificate, continuing MW until the maximum fixed-message deviation gain is at most $0.05$ on the normalized payoff scale. MW therefore constructs the response during the Stackelberg POMDP's response phase; the reward phase evaluates welfare under that certified response, and PPO optimizes the leader's mechanism. Appendix~\ref{app:stack_mspm} gives the MSPM construction and certification procedure. We report expected welfare $W$ directly, alongside the ex post efficient welfare $W^*$ for each value distribution.

\begin{figure}[!t]
    \centering
    \begin{minipage}{0.535\linewidth}
        \centering\sffamily\small (a) Two-type training
    \end{minipage}
    \hfill
    \begin{minipage}{0.445\linewidth}
        \centering\sffamily\small (b) Welfare across type counts
    \end{minipage}
    \par\smallskip
    \begin{minipage}[c]{0.535\linewidth}
        \centering
        \includegraphics[width=\linewidth]{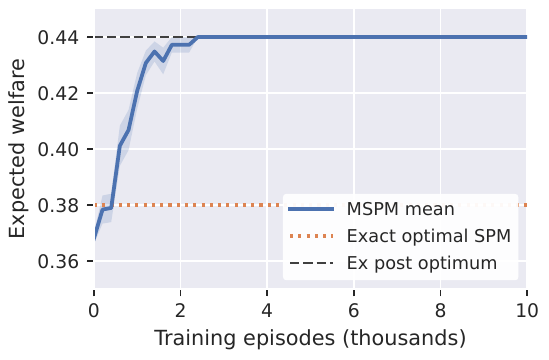}
    \end{minipage}
    \hfill
    \begin{minipage}[c]{0.445\linewidth}
        \centering
        {\footnotesize
        \setlength{\tabcolsep}{2pt}
        \renewcommand{\arraystretch}{1.3}
        \begin{tabular}{@{}cccc@{}}
            \toprule
            \makecell{Types per\\buyer} & $W^*$ & \makecell{Optimal\\SPM} &
            \makecell{MSPM best\\welfare (SE)} \\
            \midrule
            2 & $0.4400$ & $0.3800$ & $0.4400\;(0.0000)$ \\
            3 & $0.7222$ & $0.6667$ & $0.7222\;(0.0000)$ \\
            4 & $0.7083$ & $0.6667$ & $0.7058\;(0.0014)$ \\
            5 & $0.7000$ & $0.6500$ & $0.6896\;(0.0014)$ \\
            6 & $0.6944$ & $0.6500$ & $0.6838\;(0.0003)$ \\
            \bottomrule
        \end{tabular}
        }
    \end{minipage}
    \caption{MSPM welfare with two messages and 25 fresh seeds per type count. (a) Training on the scaled two-type benchmark of \citet{AgrawalSZ20}. Welfare of the current mechanism is evaluated every 200 episodes, without smoothing; shading is one standard error. Horizontal lines mark exact optimal SPM welfare $0.380$ and ex post efficient welfare $0.440$. (b) Comparison across type counts. $W^*$ is ex post efficient welfare, and the SPM optima are exact. MSPM entries give the mean (sample SE) of each seed's best certified evaluation over 10,000 episodes. Types 3--6 use independent uniform value grids.}
    \Description{The left panel shows mean two-type welfare rising to the ex post optimum of 0.44 by 2,400 episodes and remaining there through 10,000 episodes, with a one-standard-error band and an optimal standard-SPM benchmark at 0.38. The right panel lists exact welfare benchmarks and MSPM best-welfare means and standard errors for two through six types per buyer. MSPM means exceed the standard-SPM optima at every type count.}
    \label{fig:mspm_results}
\end{figure}

Figure~\ref{fig:mspm_results}(a) shows welfare of the current mechanism on the scaled two-type instance of \citet{AgrawalSZ20}. All 25 seeds reach ex post efficient welfare $0.440$. The exact optimal standard SPM attains $0.380$.

To test whether the value of communication extends beyond the specific Agrawal distribution, the other instances retain two buyers, one item, and the same reward, but draw each value independently and uniformly from $\{0,1/(\bar\theta-1),\ldots,1\}$ for $\bar\theta\in\{3,4,5,6\}$, while keeping two messages. For any interior value of the first-visited buyer, acceptance can be inefficient when the other value is $1$, whereas rejection can be inefficient when the other value is $0$. Thus, the standard-SPM optimum is below $W^*$ for every grid. Figure~\ref{fig:mspm_results}(b) compares MSPM welfare with the exact SPM optimum for each grid. Figure~\ref{fig:mspm_multitype_training} in Appendix~\ref{subapp:noregret_impl_mspm} shows how current-policy welfare evolves during training on these grids.

Panel (b) includes every validation seed and averages each seed's best certified welfare over scheduled checkpoints. These best-welfare means exceed the exact standard-SPM optima for all five type counts. Best evaluations reach $W^*$ in 100\%, 100\%, 88\%, 20\%, and 0\% of runs for types 2 through 6, respectively. These are best-checkpoint summaries; final checkpoints reach $W^*$ in 100\%, 96\%, 84\%, 0\%, and 0\% of runs, respectively.

\subsection{Platform Intervention against Algorithmic Collusion}
\label{subsec:price_collusion_rl}

Pricing algorithms can learn to sustain collusive prices in repeated markets, as demonstrated by \citet{calvano2020artificial} with Q-learning sellers. \citet{johnson2023platform} study how platforms can encourage competition by directing consumer attention toward lower-priced sellers. We build on this approach by learning a platform's display rule: after sellers post their prices, the rule determines which products consumers can purchase. Our extension introduces seller costs that are hidden from the platform. The platform must therefore commit to a single rule that encourages competitive pricing across different cost conditions.

\paragraph{Market and platform intervention.}
We retain the differentiated-product demand model used in these studies. Two sellers repeatedly post prices, and consumers choose among the displayed products and an outside option according to logit demand. Each seller's Q-learning algorithm conditions its price on the previous round's price pair. Following \citet{johnson2023platform}, we use 15 equally spaced prices from $p^{(0)}=1.0$ to $p^{(14)}=2.1$.

We replace the common marginal cost of 1.0 in their baseline with a common cost $c\in\{1.0,1.5\}$, with the two values equally likely. The cost is fixed during each seller response and known to both sellers, but hidden from the platform. Both costs use the same price grid, so prices below the high cost remain available. The platform must accommodate these different production costs using only observed prices.

The platform commits to its display rule before sellers adapt to it. In each round, the sellers post prices simultaneously, and the rule determines the displayed set, or \emph{buy box}. Only displayed sellers receive demand. As in \citet{johnson2023platform}, the platform maximizes consumer surplus (CS), the expected utility consumers obtain from the displayed products net of their prices. At fixed prices, displaying both products increases CS by giving consumers more choice. Restricting display can encourage lower prices as sellers adapt, but also reduces the available product variety. This tradeoff motivates learning the display rule. Appendix~\ref{subapp:econ_model} gives the demand, profit, and consumer-surplus equations.

\paragraph{Learning a display rule.}
The platform commits to a symmetric rule mapping each submitted price pair to a ceiling on the price grid, and displays every seller whose price is at most that ceiling. The actor observes only the current prices, sorted to impose symmetry; cost and seller-learning state are hidden. Appendix~\ref{subapp:collusion_interventions} gives the finite-grid encoding of display decisions, including ties and minimum-price cases. The whole mapping is committed before seller adaptation, so reacting to posted prices implements the platform's prior commitment.

For each committed display rule, the response algorithm trains the two Q-learning sellers, then freezes their policies. The reward is average consumer surplus along the resulting recurrent price cycle under these frozen policies. We evaluate each sampled commitment at both costs and average the two rewards, integrating over the cost distribution exactly. Appendix~\ref{subapp:stack_instantiation} specifies the query interface and seller response; Appendix~\ref{subapp:collusion_hyperparams} gives the training settings.

\paragraph{Benchmarks and evaluation.}
We compare learned rules with two threshold policies. A \emph{fixed cap} displays each seller whose price is at most a constant ceiling. A \emph{clipped-minimum rule} uses the lower of the two posted prices as its ceiling, truncated to chosen lower and upper bounds. We select the fixed cap and the clipped-rule bounds by evaluating all such rules on a separate validation cohort. The selected cap is $p^{(7)}=1.55$, and the selected clipped rule uses $p^{(4)}\approx1.3143$ as its lower bound:
\[
\tau_{\mathrm{clip}}(p)=\min\{p^{(7)},\max\{p^{(4)},\min(p_1,p_2)\}\}.
\]
The clipped rule combines three incentives. When both prices are at most $p^{(4)}$, it displays both products, preserving variety. Above this lower region, it rewards undercutting by displaying only the cheaper seller when prices differ and at least one meets the upper cap; ties within the cap retain both products. If both prices exceed $p^{(7)}$, it displays neither. The lower bound therefore limits the pressure to undercut, while the upper bound makes high prices unprofitable.

We also compare PPO with random search, which samples policies uniformly from the space of distinct display rules implemented by these price ceilings. One search evaluates 10,000 distinct rules, using one seller-learning seed per rule at both costs, and selects the rule with the highest equal-cost consumer surplus. This uses 20,000 seller responses, matching one PPO training run. Appendix~\ref{subapp:collusion_hyperparams} gives the sampling and evaluation procedure.

Table~\ref{tab:hidden_cost_platform} reports final consumer surplus for PPO and the benchmarks, evaluating each policy on \PlatformSellerN{} fresh seller seeds per cost. PPO standard errors are computed across leader-seed means. Figure~\ref{fig:hidden_cost_learning} shows current-policy validation performance throughout training, with horizontal lines for the fixed cap, clipped rule, selected random-search policy, and no intervention, which always displays both sellers.

\paragraph{Consumer-surplus results.}
PPO's final policies average $\PlatformMean\pm\PlatformSE$, compared with $\PlatformFixed$ for the fixed cap: a gain of $\PlatformGain$, or $\PlatformGainPercent\%$. Five runs score approximately $0.538$--$0.579$, while the other twenty score $0.646$--$0.755$ and beat the fixed cap. These lower-performing runs reduce the mean, so we also report the median, $\PlatformMedian$, which exceeds the clipped rule's $\PlatformClipped$. PPO's mean nevertheless remains below the clipped rule. The policy selected by random search scores $\PlatformUniform$, below the fixed cap. Appendix~\ref{app:statistical_comparisons} supports the mean differences from all three comparators after accounting for training and seller-response variation. The random-search comparison concerns its single selected policy. This comparison suggests that finding effective interventions within this budget benefits from learning to guide the search.

\begin{table}[!t]
\centering
\caption{Consumer surplus on fresh seller responses. Every policy is evaluated at both hidden costs on \PlatformSellerN{} fresh seller seeds per cost. PPO rows summarize the final policies from \PlatformN{} leader runs. Means and medians are computed across run-level scores, with each run's two costs averaged first in the last column; $\pm$ denotes one standard error across leader runs. Validation selects the fixed and clipped rules; fresh scores never select a policy. Random search supplies one selected policy, evaluated on an independent fresh cohort.}
\label{tab:hidden_cost_platform}
\small
\begin{tabular}{@{}lccc@{}}
\toprule
Intervention / summary & $c=1.0$ & $c=1.5$ & Equal-cost average \\
\midrule
Fixed threshold $p^{(7)}=1.55$ & \PlatformFixedLow & \PlatformFixedHigh & \PlatformFixed \\
Clipped minimum $[p^{(4)},p^{(7)}]$ & \PlatformClippedLow & \PlatformClippedHigh & \PlatformClipped \\
Random search (one run) & \PlatformUniformLow & \PlatformUniformHigh & \PlatformUniform \\
\midrule
PPO: mean & $\PlatformLow\pm\PlatformLowSE$ & $\PlatformHigh\pm\PlatformHighSE$ & $\PlatformMean\pm\PlatformSE$ \\
PPO: median & \PlatformLowMedian & \PlatformHighMedian & \PlatformMedian \\
\bottomrule
\end{tabular}
\end{table}

\begin{figure}[!t]
\centering
\includegraphics[width=0.40\linewidth]{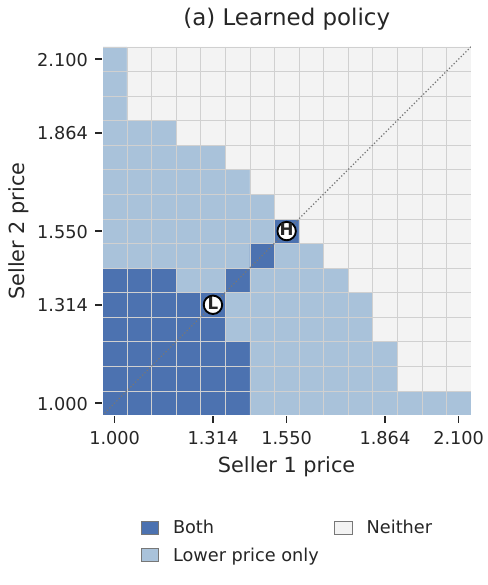}\hfill
\includegraphics[width=0.58\linewidth]{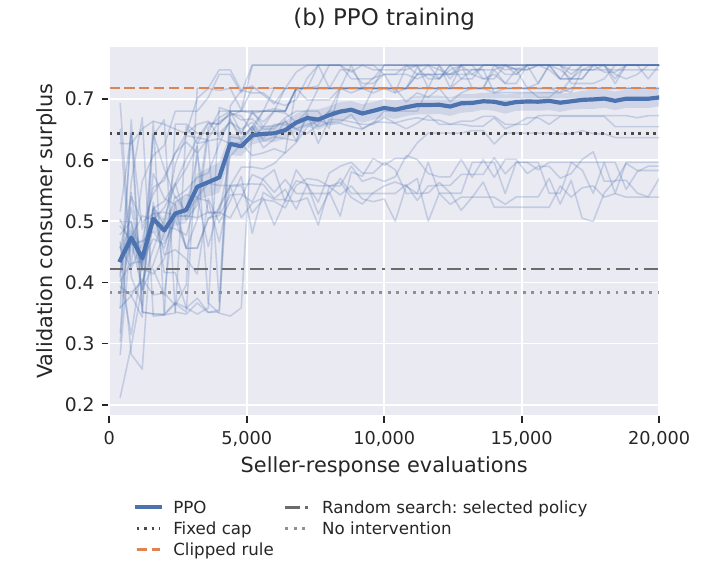}
\caption{Learned commitment and normalized PPO training. Left: a top-performing final policy (seed 801), selected among all \PlatformN{} runs by highest final validation score, with smallest-seed ties. Dark blue admits both sellers, light blue only the lower-price seller, and near-white neither. L and H mark observed long-run price pairs at costs 1.0 and 1.5: all \PlatformSellerN{} fresh responses per cost reach the respective marked pair, with both sellers displayed. Right: current-policy validation CS throughout training. Thin lines show individual runs; shading is one standard error around their mean. All curves and benchmark lines average both costs and the same five validation seller seeds. The random-search line evaluates its selected policy. The horizontal axis counts seller-response computations; see Appendix~\ref{app:collusion_details} for response and evaluation protocols.}
\Description{The left panel shows all 225 admission decisions of a top-performing learned policy using dark blue for both sellers, light blue for the lower-price seller only, and near-white for neither. L and H identify its observed low-cost and high-cost long-run price pairs. The right panel shows 25 PPO training curves and their mean, with fixed-cap, clipped-rule, selected random-search, and no-intervention validation benchmarks.}
\label{fig:hidden_cost_policy}
\label{fig:hidden_cost_learning}
\end{figure}

\paragraph{How a successful rule works.}
The illustrated learned rule (seed 801) ties for the highest fresh score among the final policies: six attain $0.7550$. At low cost, all \PlatformSellerN{} fresh seller responses reach $(p^{(4)},p^{(4)})$, with both products displayed and CS $0.8669$. At high cost, all reach $(p^{(7)},p^{(7)})$, again with both displayed and CS $0.6431$. The same rule thus accommodates two production costs without observing them. Direct enumeration of one-period deviations also identifies these two price pairs as the unique pure-strategy equilibria of the clipped-minimum benchmark at their respective costs. Their coincidence helps explain the learned policy, but equilibrium incentives alone do not determine where Q-learning converges.

Under both rules, cutting the observed low-cost price from $1.3143$ to $1.2357$ increases demand but reduces profit from $0.1522$ to $0.1326$: the smaller margin outweighs the demand gain. Raising price leads to exclusion and zero profit. Although both sellers charging the lowest price above cost, $1.0786$, would give consumers more surplus, that pair is unstable: a seller can raise its price to $1.3143$, remain displayed, and increase profit from $0.0388$ to $0.0865$. At high cost, $1.55$ is the lowest grid price above cost; cutting price generates a loss, while raising it leads to exclusion.

At low cost, the seed-801 learned rule steers seller Q-learning more reliably toward $(p^{(4)},p^{(4)})$: all 25 fresh evaluations reach this outcome, compared with 20 of 25 under the clipped-minimum benchmark. The benchmark reaches $(p^{(6)},p^{(6)})$ in two evaluations and $(p^{(7)},p^{(7)})$ in three; at high cost, both rules always reach $(p^{(7)},p^{(7)})$. This difference accounts for the learned rule's higher average CS, $0.7550$ versus $0.7356$. The two rules agree on admission at every equal-price pair and differ at 70 ordered unequal-price pairs, changing the profits sellers experience during exploration. Training against Q-learning lets the platform account for these learning dynamics when choosing its intervention. Whether the advantage persists under other seller-learning algorithms remains a question for future work.

\subsection{Bilateral Trade in Atari Space Invaders}
\label{subsec:atari_experiments}

This experiment tests whether the Meta-RL response oracle of Section~\ref{sec:scaling} can make StackPOMDP practical when a follower response includes high-dimensional visual control. \textit{Space Invaders} is an Atari 2600 arcade game in which a player moves a cannon and fires at descending aliens while avoiding their attacks \citep{bellemare2013arcade}. We modify this original single-player game into a bilateral market for scarce ammunition and study it under both seller-led and buyer-led commitments. Relearning a pixel-level follower for every candidate leader commitment would be prohibitive, making this a demanding test of amortized meta-response.

\paragraph{Economic environment.}
The environment couples two independent single-player \textit{Space Invaders} instances, one controlled by the seller and one by the buyer. Both agents begin with zero ammunition and play for a fixed horizon of $H=200$ gameplay decisions. At five exogenously sampled times, gameplay pauses and one bullet arrives to the seller. At event $j$, the seller posts a price $p_j\in[0,1]$, the buyer posts a willingness-to-pay threshold $t_j\in[0,1]$, and the bullet and payment transfer iff $p_j\leq t_j$; otherwise the seller retains the bullet.

In the basic POMG, agent $i$ observes $o_{i,t}=(x_{i,t},u_{i,t})$, where $x_{i,t}\in[0,1]^{4\times84\times84}$ is the standard four-frame Atari image stack and $u_{i,t}\in[0,1]^9$ is a low-dimensional state. Specifically, $u_{i,t}=(b_{i,t},z_{i,t},\tau_t,d_t,e_t)$ contains the agent's normalized ammunition count $b_{i,t}\in[0,1]$, an indicator $z_{i,t}\in\{0,1\}$ of whether a projectile is active,\footnote{An active projectile prevents the agent from firing again until that projectile leaves the screen.} normalized game time $\tau_t\in[0,1]$, a trade-pause indicator $d_t\in\{0,1\}$, and a five-entry one-hot encoding $e_t\in\{0,1\}^5$ of the next or current trade event. Each policy emits a composite action $a=(a^{\mathrm{game}},a^{\mathrm{econ}})$: $a^{\mathrm{game}}$ selects one of six Atari actions, whereas $a^{\mathrm{econ}}\in[0,1]$ supplies the price or threshold used only during a paused trade.

\begin{figure}[!t]
\centering
\includegraphics[width=\textwidth]{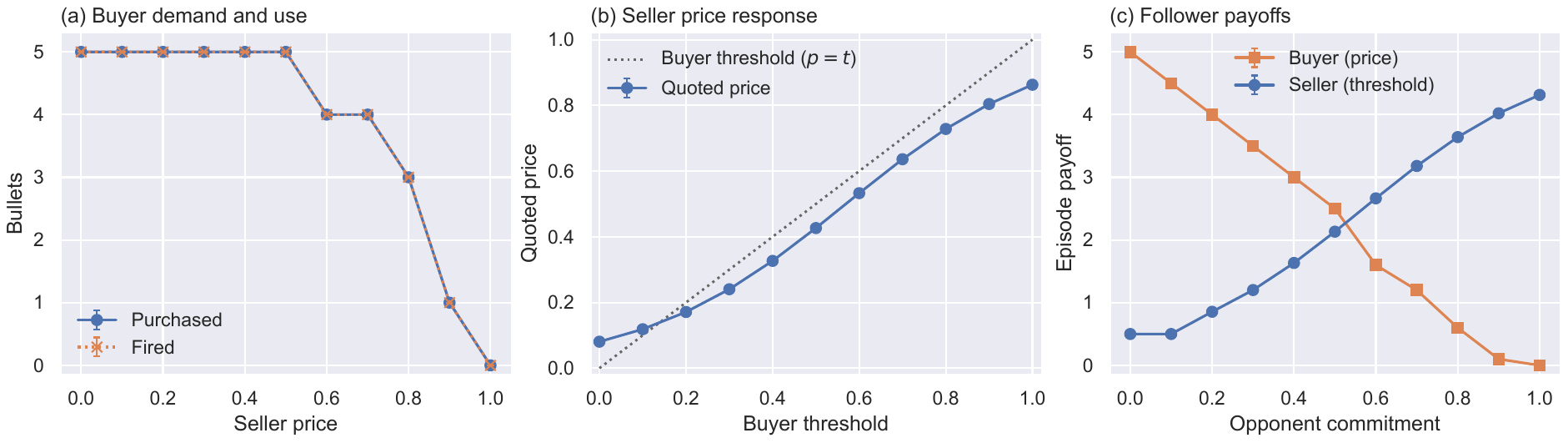}
\caption{Fixed-context diagnostics for the two contextual meta-followers. Panel (a) shows buyer purchases and shots against constant seller-price vectors $p\mathbf{1}$. Panel (b) shows seller quotes against constant buyer-threshold vectors $t\mathbf{1}$; the dotted line $p=t$ marks the acceptance boundary. Panel (c) combines role-specific follower payoffs. The buyer curve varies the seller price, whereas the seller curve varies the buyer threshold, so the two curves come from separate evaluations. Each point averages 20 held-out episodes using the schedule $(20,50,80,110,140)$; error bars show one sample standard error and are visually negligible.}
\Description{Three response panels: buyer purchases and shots fall together at high committed seller prices; seller quotes rise with the committed buyer threshold and remain below it once trade begins; buyer payoff falls with seller price, while seller payoff rises with buyer threshold.}
\label{fig:atari_response_diagnostics}
\end{figure}

Let $G_S$ and $G_B$ denote the agents' total clipped gameplay rewards, and let $I_j=\mathbf{1}[p_j\leq t_j]$. The episode payoffs are
\begin{equation}
    R_S = 0.1G_S + \sum_{j=1}^{5}p_j I_j,
    \qquad
    R_B = G_B - \sum_{j=1}^{5}p_j I_j.
    \label{eq:atari_payoffs}
\end{equation}
The coefficient $0.1$ scales down the seller's gameplay payoff relative to the buyer's, creating gains from transferring ammunition. A bullet's realized value nevertheless varies across events with its arrival time, current inventory, and the subsequent game states.

\paragraph{Role assignments and meta-response.}
We evaluate two separate role assignments. In the seller-led configuration, seller $S$ is the leader and commits to a policy $\pi_\ell$ governing both gameplay and pricing; the buyer response policy selects gameplay actions and willingness-to-pay thresholds. In the buyer-led configuration, buyer $B$ is the leader and commits to a policy $\pi_\ell$ governing both gameplay and threshold setting; the seller response policy selects gameplay actions and prices. Leadership remains fixed throughout each training run and episode; role reversal refers to comparing these two configurations.

Following Section~\ref{sec:scaling}, we train each contextual meta-follower offline over role-specific economic contexts and freeze it during leader training. Both the contextual meta-followers and the StackPOMDP leaders are trained with PPO; Appendix~\ref{app:atari_pomg} gives the architectures, hyperparameters, and evaluation procedure. The leader formally commits to its complete policy, but the Atari oracle summarizes its economic component for the follower through a five-entry vector $\omega$: prices for the meta-buyer and thresholds for the meta-seller.

Each StackPOMDP episode constructs $\omega$ through five zero-reward queries, one for each trade event $j\in\{1,\ldots,5\}$. At query $j$, the leader receives a canonical event observation that identifies event $j$ but contains no live gameplay information.\footnote{Appendix~\ref{app:atari_pomg} gives the exact observation encodings and full implementation details.} The resulting vector is
\[
    \omega=(a^{L,\mathrm{econ}}_1,\ldots,a^{L,\mathrm{econ}}_5).
\]

\begin{figure}[!t]
\centering
\includegraphics[width=\textwidth]{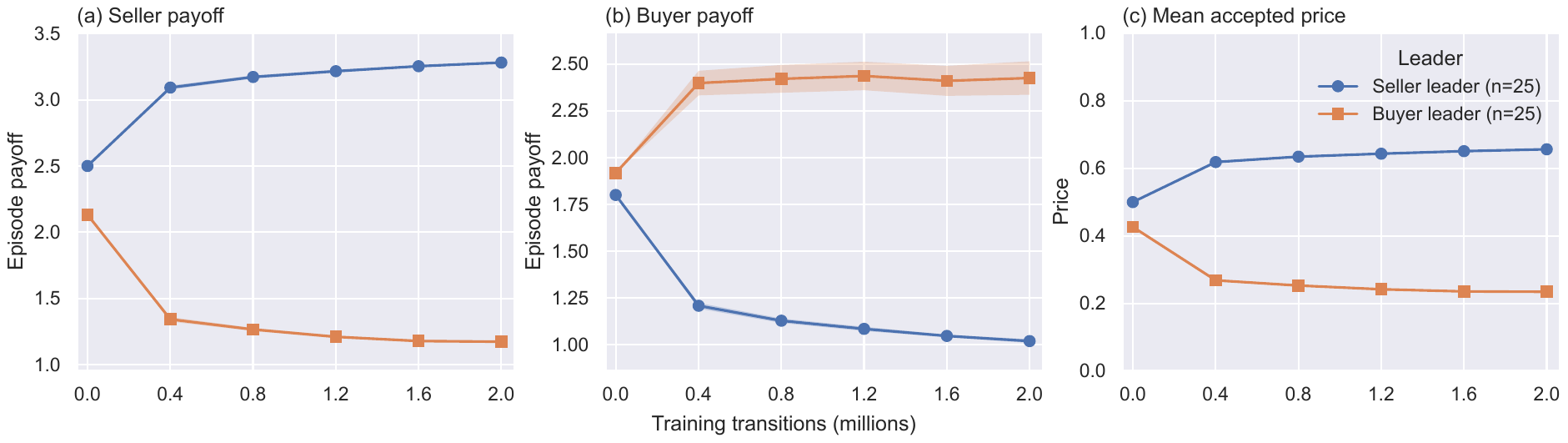}
\caption{StackPOMDP training under role reversal, including exact pre-update evaluation. Panels report (a) seller payoff, (b) buyer payoff, and (c) mean accepted transaction price. Curves average 25 training-seed policy means per role evaluated on the same 20 held-out episodes; bands show one standard error across seed-level means. At initialization, every seed-specific reconstruction has the same deterministic economic output, $0.5$. The next four points are prespecified training checkpoints and the last is the actual post-training policy. Seller leadership raises accepted prices and seller payoff, whereas buyer leadership lowers prices and shifts payoff toward the buyer.}
\Description{Three Atari training curves showing that seller leadership increases prices and seller payoff, while buyer leadership lowers prices and increases buyer payoff.}
\label{fig:atari_meta_stackpomdp}
\end{figure}

During the fresh bilateral game that follows, the frozen meta-follower receives $\omega$ alongside its basic POMG observation. The five economic actions remain fixed throughout the episode, while gameplay actions remain stochastic.

\paragraph{Follower-response validation.}
Before optimizing either leader, we verify that the frozen meta-followers respond in the expected direction to the committed economic vectors. Figure~\ref{fig:atari_response_diagnostics} reports the fixed-context diagnostics. Panel~(a) shows that the meta-buyer purchases and fires all five bullets at prices up to $0.5$, then reduces both quantities as the price rises, reaching zero at a price of $1$. Panel~(b) shows the complementary seller response. At buyer thresholds of zero and $0.1$, the meta-seller quotes above the acceptance boundary, retains the bullets, and fires all five; from a threshold of $0.2$ onward, it quotes below the boundary, sells all five bullets, and raises its price with the buyer's willingness to pay. These responses translate directly into the payoff patterns in panel~(c): buyer payoff falls with the seller's price, whereas seller payoff rises with the buyer's threshold. Together, the panels show that each meta-follower uses the commitment context rather than reproducing a fixed trade policy. 

\paragraph{Role-reversal results.}
Having established both responses, we train 25 independent StackPOMDP leaders per role for approximately two million physical transitions, reusing the same frozen meta-followers. Figure~\ref{fig:atari_meta_stackpomdp} reports the exact pre-update evaluation, four prespecified checkpoints, and the final policy. Both leaders initially output $0.5$ at every economic query, yielding mean accepted prices of $0.500$ under seller leadership and $0.426$ under buyer leadership. By the first checkpoint at around $400$k transitions, mean prices reach $0.618$ and $0.268$; terminal means are $0.656$ and $0.235$, respectively. Thus, most separation emerges early and is subsequently refined. At termination, seller leaders obtain a mean payoff of $3.281$, leaving buyers $1.019$; under buyer leadership, sellers obtain $1.173$ and buyers $2.425$. The three terminal role differences---accepted price, seller payoff, and buyer payoff---are statistically supported by the tests in Appendix~\ref{app:statistical_comparisons}. Buyer payoff is more sensitive to gameplay variation because gameplay reward enters at full weight, compared with $0.1$ for the seller. Under buyer leadership, this variation also reflects differences across independently trained gameplay policies. A separate final evaluation on 100 fresh emulator episodes per policy confirms the price separation ($0.656$ versus $0.235$) and payoff ordering; the buyer-led buyer payoff averages $2.489$ in this additional cohort.

\section{Conclusion}
\label{sec:conclusion}

This paper introduced a reinforcement-learning formulation with a provable Stackelberg optimality guarantee. When follower behavior can be instantiated as a policy-interactive response algorithm, embedding that algorithm in the leader's learning environment yields a single-agent Stackelberg POMDP. An optimal memoryless policy in this POMDP recovers an optimal leader commitment under the specified follower-response procedure (Theorem~\ref{thm:main}). For practical optimization, we use PPO with a centralized critic that receives privileged information for value estimation, while the deployed leader policy uses only its own observations.

Across the applications, the experiments show how accounting for follower adaptation can help the leader discover effective commitments. The economic gains depend on how commitments shape followers' incentives and subsequent behavior, making the response procedure central to the optimization problem. The framework accommodates both certified approximate equilibrium responses and learned behavior in complex environments, connecting a common formulation to different assumptions about how followers react.

Further work can build on each of these applications. In indirect mechanism design, MSPMs with several messaging rounds could request new reports between successive offers and adapt later prices and offer order to the information received. This would let us study when additional communication improves welfare under strategic reporting, and how to compute certified follower responses as the interaction becomes richer.

In platform design, learned interventions offer a lens on the behavior of pricing algorithms. Comparing the rules learned against different seller algorithms, and testing whether they transfer, could reveal which features of adaptation sustain collusion and which incentives reliably disrupt it. This would connect intervention design with work on the collusive behavior of learning algorithms~\citep{cartea2022learning, brown2023competition}.

An important computational direction is meta-learning with multiple strategically interacting followers. Their response models must capture adaptation to the leader and strategic dependence on one another, raising questions about context design and joint training. Extending Atari trade to several interacting traders would provide a setting for testing whether such models keep response computation practical. Varying gameplay conditions, resource scarcity, and trade timing could also reveal how the environment shapes trading behavior and the value of commitment.

\paragraph{Broader Impact and Limitations.}
As with any simulator-based policy design, the conclusions depend on whether the simulator's data, behaviors, and preferences represent the deployment population. An optimized leader may exploit misspecification in the follower simulator and fail when real responses shift, motivating robustness tests, out-of-distribution evaluation, continued monitoring, and retraining. The framework optimizes the specified leader objective, which need not align with every participant's interests. Platform interventions targeting consumer surplus may reduce seller profits or restrict access, while our bilateral-trade experiments show that assigning leadership changes how gains are distributed between buyer and seller. The normative objective and any welfare or access constraints should therefore be chosen explicitly and audited with affected stakeholders before deployment.

\section*{Acknowledgements}
This research is funded in part by Defense Advanced Research Projects Agency under Cooperative Agreement HR00111920029. The content of the information does not necessarily reflect the position or the policy of the Government, and no official endorsement should be inferred. This is approved for public release; distribution is unlimited. The work of G. Brero was also supported by the SNSF (Swiss National Science Foundation) under Fellowship P2ZHP1\_191253.

\printbibliography

@inproceedings{thoma2024contextual,
  title={Contextual Bilevel Reinforcement Learning for Incentive Alignment},
  author={Thoma, Vinzenz and P{\'a}sztor, Barna and Krause, Andreas and Ramponi, Giorgia and Hu, Yifan},
  booktitle={Advances in Neural Information Processing Systems},
  volume={37},
  pages={127369--127435},
  year={2024},
  publisher={Curran Associates, Inc.},
  doi={10.52202/079017-4046},
  url={https://proceedings.neurips.cc/paper_files/paper/2024/hash/e66309ead63bc1410d2df261a28f602d-Abstract-Conference.html}
}

@inproceedings{frankel2025finite,
  title={Finite-Time Convergence Rates in Stochastic {Stackelberg} Games with Smooth Algorithmic Agents},
  author={Frankel, Eric and Kulkarni, Kshitij and Drusvyatskiy, Dmitriy and Oh, Sewoong and Ratliff, Lillian J.},
  booktitle={Proceedings of the 42nd International Conference on Machine Learning},
  volume={267},
  series={Proceedings of Machine Learning Research},
  pages={17575--17656},
  year={2025},
  publisher={PMLR},
  url={https://proceedings.mlr.press/v267/frankel25b.html}
}

@misc{ivanov2024principal,
  title={Principal-Agent Reinforcement Learning: Orchestrating {AI} Agents with Contracts},
  author={Ivanov, Dima and D{\"u}tting, Paul and Talgam-Cohen, Inbal and Wang, Tonghan and Parkes, David C.},
  howpublished={arXiv preprint arXiv:2407.18074},
  year={2024},
  url={https://arxiv.org/abs/2407.18074}
}

@inproceedings{kudo2025policy,
  title={Policy Iteration for Two-Player General-Sum Stochastic {Stackelberg} Games},
  author={Kudo, Mikoto and Akimoto, Youhei},
  booktitle={Proceedings of the 17th Asian Conference on Machine Learning},
  volume={304},
  series={Proceedings of Machine Learning Research},
  pages={718--733},
  year={2025},
  publisher={PMLR},
  url={https://proceedings.mlr.press/v304/kudo26a.html}
}

@inproceedings{zeng2025meanfield,
  title={Learning in {Stackelberg} Mean Field Games: A Non-Asymptotic Analysis},
  author={Zeng, Sihan and Evans, Benjamin Patrick and Bhatt, Sujay and Ardon, Leo and Ganesh, Sumitra and Koppel, Alec},
  booktitle={Advances in Neural Information Processing Systems},
  volume={38},
  pages={170504--170554},
  year={2025},
  publisher={Curran Associates, Inc.},
  doi={10.52202/085713-5675},
  url={https://proceedings.neurips.cc/paper_files/paper/2025/hash/f8f5329ee619d1cb398b4b6233e093d4-Abstract-Conference.html}
}

@inproceedings{maddux2025noregret,
  title={No-Regret Learning in {Stackelberg} Games with an Application to Electric Ride-Hailing},
  author={Maddux, Anna and Maljkovic, Marko and Geroliminis, Nikolas and Kamgarpour, Maryam},
  booktitle={2025 IEEE 64th Conference on Decision and Control (CDC)},
  pages={4009--4014},
  year={2025},
  publisher={IEEE},
  doi={10.1109/CDC57313.2025.11312071},
  url={https://arxiv.org/abs/2504.03745}
}

@inproceedings{personnat2026learning,
  title={Learning to Play Multi-Follower {Bayesian Stackelberg} Games},
  author={Personnat, Gerson and Lin, Tao and Hossain, Safwan and Parkes, David C.},
  booktitle={International Conference on Learning Representations},
  year={2026},
  url={https://openreview.net/forum?id=8hMaqBagPd}
}

@inproceedings{zhang2024social,
  title={Position: Social Environment Design Should be Further Developed for {AI}-based Policy-Making},
  author={Zhang, Edwin and Zhao, Sadie and Wang, Tonghan and Hossain, Safwan and Gasztowtt, Henry and Zheng, Stephan and Parkes, David C. and Tambe, Milind and Chen, Yiling},
  booktitle={Proceedings of the 41st International Conference on Machine Learning},
  volume={235},
  series={Proceedings of Machine Learning Research},
  pages={60527--60540},
  year={2024},
  publisher={PMLR},
  url={https://proceedings.mlr.press/v235/zhang24cl.html}
}

@misc{thoma2026deep,
  title={Deep Incentive Design with Differentiable Equilibrium Blocks},
  author={Thoma, Vinzenz and Piliouras, Georgios and Marris, Luke},
  howpublished={arXiv preprint arXiv:2603.07705},
  year={2026},
  url={https://arxiv.org/abs/2603.07705}
}

@book{fudenberg1991game,
  title={Game Theory},
  author={Fudenberg, Drew and Tirole, Jean},
  year={1991},
  publisher={MIT Press},
  address={Cambridge, MA}
}

@book{stackelberg1934marktform,
  author    = {von Stackelberg, Heinrich},
  title     = {Marktform und Gleichgewicht},
  year      = {1934},
  publisher = {Julius Springer},
  address   = {Wien and Berlin}
}

@article{johnson2023platform,
  title={Platform design when sellers use pricing algorithms},
  author={Johnson, Justin P and Rhodes, Andrew and Wildenbeest, Matthijs},
  journal={Econometrica},
  volume={91},
  number={5},
  pages={1841--1879},
  year={2023},
}

@article{aumann1974subjectivity,
  title={Subjectivity and correlation in randomized strategies},
  author={Aumann, Robert J},
  journal={Journal of Mathematical Economics},
  volume={1},
  number={1},
  pages={67--96},
  year={1974},
}

@article{arora2012multiplicative,
  title={The Multiplicative Weights Update Method: A Meta-Algorithm and Applications},
  author={Arora, Sanjeev and Hazan, Elad and Kale, Satyen},
  journal={Theory of Computing},
  volume={8},
  number={6},
  pages={121--164},
  year={2012},
  doi={10.4086/toc.2012.v008a006}
}

@book{cesabianchi2006prediction,
  title={Prediction, Learning, and Games},
  author={Cesa-Bianchi, Nicol{\`o} and Lugosi, G{\'a}bor},
  year={2006},
  publisher={Cambridge University Press},
  address={Cambridge, UK},
  doi={10.1017/CBO9780511546921}
}

@article{williams1992simple,
  title={Simple statistical gradient-following algorithms for connectionist reinforcement learning},
  author={Williams, Ronald J.},
  journal={Machine Learning},
  volume={8},
  pages={229--256},
  year={1992},
  doi={10.1007/BF00992696}
}

@book{sutton2018reinforcement,
  title={Reinforcement Learning: An Introduction},
  author={Sutton, Richard S. and Barto, Andrew G.},
  year={2018},
  edition={2nd},
  publisher={MIT Press},
  address={Cambridge, MA},
  isbn={978-0262039246}
}

@misc{duan2016rl,
  title={RL$^2$: Fast reinforcement learning via slow reinforcement learning},
  author={Duan, Yan and Schulman, John and Chen, Xi and Bartlett, Peter and Sutskever, Ilya and Abbeel, Pieter},
  howpublished={arXiv preprint arXiv:1611.02779},
  year={2016},
  url={https://arxiv.org/abs/1611.02779}
}

@inproceedings{ng2000algorithms,
  title={Algorithms for inverse reinforcement learning},
  author={Ng, Andrew Y. and Russell, Stuart},
  booktitle={Proceedings of the 17th International Conference on Machine Learning},
  pages={663--670},
  year={2000},
  publisher={Morgan Kaufmann Publishers Inc.},
  address={San Francisco, CA, USA}
}

@article{wang2024survey,
  title={A Survey on Large Language Model Based Autonomous Agents},
  author={Wang, Lei and Ma, Chen and Feng, Xueyang and Zhang, Zeyu and Yang, Hao and Zhang, Jingsen and Chen, Zhiyuan and Tang, Jiakai and Chen, Xu and Lin, Yankai and Zhao, Wayne Xin and Wei, Zhewei and Wen, Jirong},
  journal={Frontiers of Computer Science},
  volume={18},
  pages={186345},
  year={2024},
  doi={10.1007/s11704-024-40231-1},
  url={https://doi.org/10.1007/s11704-024-40231-1}
}

@misc{fish2024algorithmic,
  title={Algorithmic Collusion by Large Language Models},
  author={Fish, Sara and Gonczarowski, Yannai A. and Shorrer, Ran I.},
  howpublished={arXiv preprint arXiv:2404.00806},
  year={2024},
  url={https://arxiv.org/abs/2404.00806}
}

@inproceedings{hausknecht2015drqn,
  title={Deep Recurrent Q-Learning for Partially Observable {MDP}s},
  author={Hausknecht, Matthew and Stone, Peter},
  booktitle={AAAI Fall Symposium on Sequential Decision Making for Intelligent Agents},
  series={AAAI Technical Report},
  number={FS-15-06},
  pages={29--37},
  year={2015},
  publisher={AAAI Press},
  address={Palo Alto, CA, USA},
  url={https://aaai.org/papers/11673-11673-deep-recurrent-q-learning-for-partially-observable-mdps/}
}

@article{wang2022coordinating,
  title={Coordinating Followers to Reach Better Equilibria: End-to-End Gradient Descent for {Stackelberg} Games},
  author={Wang, Kai and Xu, Lily and Perrault, Andrew and Reiter, Michael K and Tambe, Milind},
  journal={Proceedings of the AAAI Conference on Artificial Intelligence},
  volume={36},
  number={5},
  pages={5219--5227},
  year={2022},
  doi={10.1609/aaai.v36i5.20457},
  url={https://doi.org/10.1609/aaai.v36i5.20457}
}

@article{lovejoy1991computational,
  title   = {A Survey of Algorithmic Methods for Partially Observed {M}arkov Decision Processes},
  author  = {William S. Lovejoy},
  journal = {Annals of Operations Research},
  year    = {1991},
  volume  = {28},
  pages   = {47--65}
}

@misc{heess2015memory,
  title={Memory-Based Control with Recurrent Neural Networks},
  author={Heess, Nicolas and Hunt, Jonathan J. and Lillicrap, Timothy P. and Silver, David},
  howpublished={arXiv preprint arXiv:1512.04455},
  year={2015},
  url={https://arxiv.org/abs/1512.04455}
}

@article{smallwood1973optimal,
  title   = {The Optimal Control of Partially Observable {M}arkov Processes over a Finite Horizon},
  author  = {Richard D. Smallwood and Edward J. Sondik},
  journal = {Operations Research},
  year    = {1973},
  volume  = {21},
  number  = {5},
  pages   = {1071--1088}
}

@inproceedings{letchford2009learning,
  title={Learning and approximating the optimal strategy to commit to},
  author={Letchford, Joshua and Conitzer, Vincent and Munagala, Kamesh},
  booktitle={Algorithmic Game Theory: Second International Symposium, SAGT 2009},
  series={Lecture Notes in Computer Science},
  volume={5814},
  pages={250--262},
  year={2009},
  publisher={Springer},
  address={Berlin, Heidelberg},
  doi={10.1007/978-3-642-04645-2_23},
  url={https://doi.org/10.1007/978-3-642-04645-2_23}
}

@article{zheng2022ai,
  title={The {AI Economist: Taxation} policy design via two-level deep multiagent reinforcement learning},
  author={Zheng, Stephan and Trott, Alexander and Srinivasa, Sunil and Parkes, David C and Socher, Richard},
  journal={Science Advances},
  volume=8,
 number = {18},
pages = {eabk2607},
  year={2022}
}

@misc{balaguer2022good,
  title={The Good Shepherd: An Oracle Agent for Mechanism Design},
  author={Balaguer, Jan and Koster, Raphael and Summerfield, Christopher and Tacchetti, Andrea},
  howpublished={ICLR 2022 Workshop on Gamification and Multiagent Solutions},
  year={2022},
  url={https://openreview.net/forum?id=HGNLs51peq}
}

@inproceedings{yang2022adaptive,
  title={Adaptive Incentive Design with Multi-Agent Meta-Gradient Reinforcement Learning},
  author={Yang, Jiachen and Wang, Ethan and Trivedi, Rakshit and Zhao, Tuo and Zha, Hongyuan},
  booktitle={Proceedings of the 21st International Conference on Autonomous Agents and Multiagent Systems},
  pages={1436--1445},
  year={2022},
  publisher={International Foundation for Autonomous Agents and Multiagent Systems},
  address={Richland, SC, USA}
}

@article{peng2019learning,
  title={Learning optimal strategies to commit to},
  author={Peng, Binghui and Shen, Weiran and Tang, Pingzhong and Zuo, Song},
  journal={Proceedings of the AAAI Conference on Artificial Intelligence},
  volume={33},
  number={1},
  pages={2149--2156},
  year={2019},
  doi={10.1609/aaai.v33i01.33012149},
  url={https://doi.org/10.1609/aaai.v33i01.33012149}
}

@inproceedings{paruchuri2008playing,
  title={Playing games for security: {An} efficient exact algorithm for solving {Bayesian Stackelberg} games},
  author={Paruchuri, Praveen and Pearce, Jonathan P and Marecki, Janusz and Tambe, Milind and Ordonez, Fernando and Kraus, Sarit},
  booktitle={Proceedings of the 7th International Joint Conference on Autonomous Agents and Multiagent Systems},
  pages={895--902},
  year={2008},
  publisher={International Foundation for Autonomous Agents and Multiagent Systems},
  address={Richland, SC, USA}
}

@article{conitzer2011commitment,
  title={Commitment to correlated strategies},
  author={Conitzer, Vincent and Korzhyk, Dmytro},
  journal={Proceedings of the AAAI Conference on Artificial Intelligence},
  volume={25},
  number={1},
  pages={632--637},
  year={2011},
  doi={10.1609/aaai.v25i1.7875},
  url={https://doi.org/10.1609/aaai.v25i1.7875}
}

@article{coniglio2020computing,
  title={Computing a pessimistic stackelberg equilibrium with multiple followers: The mixed-pure case},
  author={Coniglio, Stefano and Gatti, Nicola and Marchesi, Alberto},
  journal={Algorithmica},
  volume={82},
  number={5},
  pages={1189--1238},
  year={2020},
}

@techreport{von2004leadership,
  title={Leadership with commitment to mixed strategies},
  author={von Stengel, Bernhard and Zamir, Shmuel},
  year={2004},
  institution={Technical Report LSE-CDAM-2004-01, CDAM Research Report}
}

@book{NationalAcademies2024Options,
  author    = {{National Academies of Sciences, Engineering, and Medicine}},
  title     = {Options for a National Plan for Smart Manufacturing},
  year      = {2024},
  publisher = {The National Academies Press},
  address   = {Washington, DC},
  doi       = {10.17226/27260},
  url       = {https://doi.org/10.17226/27260}
}

@misc{BoschSDS2024LightsOut,
  author = {Pasha, Abdulla},
  title = {Lights-Out Manufacturing: Revolutionizing the Factory Floor with Automation},
  howpublished = {White paper, Bosch Software and Digital Solutions},
  year = {2024},
  url = {https://bosch-sds.com/wp-content/uploads/2024/09/Lights-Out-Manufacturing_-Revolutionizing-the-Factory-Floor-with-Automation.pdf}
}

@article{calvano2020protecting,
  title={Protecting consumers from collusive prices due to AI},
  author={Calvano, Emilio and Calzolari, Giacomo and Denicol{\`o}, Vincenzo and Harrington, Joseph E and Pastorello, Sergio},
  journal={Science},
  volume={370},
  number={6520},
  pages={1040--1042},
  year={2020},
}

@misc{FTC,
  author = {{U.S. Federal Trade Commission}},
  title = {\emph{The Competition and Consumer Protection Issues of Algorithms, Artificial Intelligence, and Predictive Analytics}, Hearing on Competition and Consumer Protection in the 21st Century, U.S. Federal Trade Commission},
  journal = {},
  year = {2018},
  url = {https://www.ftc.gov/policy/hearings-competition-consumer-protection}
}

@misc{OECD,
  author = {{European Union}},
  title = {Algorithms and Collusion---Note from the European Union},
  howpublished = {OECD Competition Committee, DAF/COMP/WD(2017)12},
  year = {2017},
  url = {https://one.oecd.org/document/DAF/COMP/WD(2017)12/en/pdf}
}

@article{assad2020algorithmic,
  author={Assad, Stephanie and Clark, Robert and Ershov, Daniel and Xu, Lei},
  title={Algorithmic Pricing and Competition: Empirical Evidence from the German Retail Gasoline Market},
  journal={Journal of Political Economy},
  volume={132},
  number={3},
  pages={723--771},
  year={2024},
  doi={10.1086/726906},
  url={https://doi.org/10.1086/726906}
}

@misc{javaid2022supply,
  author       = {Javaid, Shehmir},
  title        = {Top 12 Supply Chain AI Use Cases in 2024},
  year         = {2024},
  month        = 1,
  day          = {3},
  note         = {AIMultiple. Updated January 3, 2024; archived January 6, 2024},
  url          = {https://web.archive.org/web/20240106062348/https://research.aimultiple.com/supply-chain-ai/}
}

@misc{alicke2021covid,
  author       = {Alicke, Knut and Barriball, Elisabeth and Trautwein, Vera},
  title        = {How COVID-19 is reshaping supply chains},
  year         = {2021},
  month        = 11,
  note         = {McKinsey \& Company},
  url          = {https://www.mckinsey.com/capabilities/operations/our-insights/how-covid-19-is-reshaping-supply-chains}
}

@inproceedings{rakelly2019efficient,
  title={Efficient off-policy meta-reinforcement learning via probabilistic context variables},
  author={Rakelly, Kate and Zhou, Aurick and Finn, Chelsea and Levine, Sergey and Quillen, Deirdre},
  booktitle={Proceedings of the 36th International Conference on Machine Learning},
  series={Proceedings of Machine Learning Research},
  volume={97},
  pages={5331--5340},
  year={2019},
  publisher={PMLR},
  address={Long Beach, CA, USA},
  url={https://proceedings.mlr.press/v97/rakelly19a.html}
}

@misc{GlcBrero_StackelbergPOMDP,
  author = {Gerstgrasser, Matthias and Brero, Gianluca and Eden, Alon and Chakrabarti, Darshan and Lepore, Nicolas and Li, Vincent and Mibuari, Eric and Greenwald, Amy and Parkes, David C.},
  title = {{Stackelberg POMDP: Learning to Lead via Reinforcement Learning} (code release)},
  year = {2026},
  organization = {GitHub},
  version = {jair-2026-v1.0.0},
  howpublished = {\url{https://github.com/GlcBrero/StackelbergPOMDP/releases/tag/jair-2026-v1.0.0}}
}

@article{calvano2020artificial,
  title={Artificial intelligence, algorithmic pricing, and collusion},
  author={Calvano, Emilio and Calzolari, Giacomo and Denicolo, Vincenzo and Pastorello, Sergio},
  journal={American Economic Review},
  volume={110},
  number={10},
  pages={3267--3297},
  year={2020},
}

@article{abada2023artificial,
  title={Artificial intelligence: Can seemingly collusive outcomes be avoided?},
  author={Abada, Ibrahim and Lambin, Xavier},
  journal={Management Science},
  volume={69},
  number={9},
  pages={5042--5065},
  year={2023},
}

@inproceedings{zhang2023steering,
  title={Steering No-Regret Learners to a Desired Equilibrium},
  author={Zhang, Brian Hu and Farina, Gabriele and Anagnostides, Ioannis and Cacciamani, Federico and McAleer, Stephen Marcus and Haupt, Andreas Alexander and Celli, Andrea and Gatti, Nicola and Conitzer, Vincent and Sandholm, Tuomas},
  booktitle={Proceedings of the 25th ACM Conference on Economics and Computation (EC '24)},
  pages={73--74},
  year={2024},
  publisher={Association for Computing Machinery},
  address={New York, NY, USA},
  doi={10.1145/3670865.3673536},
  url={https://doi.org/10.1145/3670865.3673536}
}

@inproceedings{sandholm2006sequences,
  title={Sequences of take-it-or-leave-it offers: Near-optimal auctions without full valuation revelation},
  author={Sandholm, Tuomas and Gilpin, Andrew},
  booktitle={Proceedings of the 5th International Joint Conference on Autonomous Agents and Multiagent Systems},
  pages={1127--1134},
  year={2006},
  publisher={Association for Computing Machinery},
  address={New York, NY, USA}
}

@misc{brero2021learning,
  title={Learning {Stackelberg} equilibria in sequential price mechanisms},
  author={Brero, Gianluca and Chakrabarti, Darshan and Eden, Alon and Gerstgrasser, Matthias and Li, Vincent and Parkes, David C},
  howpublished={ICML Workshop on Reinforcement Learning Theory},
  year={2021}
}

@inproceedings{schulman2015stochastic,
  author={Schulman, John and Heess, Nicolas and Weber, Theophane and Abbeel, Pieter},
  title={Gradient Estimation Using Stochastic Computation Graphs},
  booktitle={Advances in Neural Information Processing Systems},
  volume={28},
  year={2015},
  url={https://proceedings.neurips.cc/paper/2015/hash/de03beffeed9da5f3639a621bcab5dd4-Abstract.html}
}

@misc{schulman2017proximal,
  author={Schulman, John and Wolski, Filip and Dhariwal, Prafulla and Radford, Alec and Klimov, Oleg},
  title={Proximal Policy Optimization Algorithms},
  howpublished={arXiv preprint arXiv:1707.06347},
  year={2017},
  url={https://arxiv.org/abs/1707.06347}
}

@article{zhong2021can,
  title={Can Reinforcement Learning Find {Stackelberg-Nash} Equilibria in General-Sum Markov Games with Myopically Rational Followers?},
  author={Zhong, Han and Yang, Zhuoran and Wang, Zhaoran and Jordan, Michael I},
  journal={Journal of Machine Learning Research},
  volume={24},
  number={35},
  pages={1--52},
  year={2023},
  url={https://www.jmlr.org/papers/v24/22-0203.html}
}

@misc{ShuT19,
  author    = {Tianmin Shu and
               Yuandong Tian},
  title     = {{$M^3$RL}: Mind-aware Multi-agent Management Reinforcement Learning},
  howpublished = {7th International Conference on Learning Representations},
  year      = {2019},
  url       = {https://openreview.net/forum?id=BkzeUiRcY7}
}

@misc{ShiYWWZLA20,
  author    = {Zhenyu Shi and
               Runsheng Yu and
               Xinrun Wang and
               Rundong Wang and
               Youzhi Zhang and
               Hanjiang Lai and
               Bo An},
  title     = {Learning Expensive Coordination: An Event-Based Deep {RL} Approach},
  howpublished = {8th International Conference on Learning Representations},
  year      = {2020},
  url       = {https://openreview.net/forum?id=ryeG924twB}
}

@inproceedings{blum2014learning,
  title={Learning optimal commitment to overcome insecurity},
  author={Blum, Avrim and Haghtalab, Nika and Procaccia, Ariel D},
  booktitle={Advances in Neural Information Processing Systems},
  volume={27},
  pages={1826--1834},
  year={2014},
  publisher={Curran Associates, Inc.},
  address={Red Hook, NY, USA},
  url={https://proceedings.neurips.cc/paper_files/paper/2014/hash/d3eea90e8e9cff58fc84a4c40d22d95a-Abstract.html}
}

@inproceedings{BreroLMP22,
  title={Learning to Mitigate {AI} Collusion on Economic Platforms},
  author={Brero, Gianluca and Mibuari, Eric and Lepore, Nicolas and Parkes, David C.},
  booktitle={Advances in Neural Information Processing Systems},
  volume={35},
  pages={37892--37904},
  year={2022},
  publisher={Curran Associates, Inc.},
  address={Red Hook, NY, USA}
}

@inproceedings{basilico2017methods,
  title={Methods for finding leader--follower equilibria with multiple followers},
  author={Basilico, Nicola and Coniglio, Stefano and Gatti, Nicola},
  booktitle={Proceedings of the 15th International Conference on Autonomous Agents and Multiagent Systems},
  pages={1363--1364},
  year={2016},
  publisher={International Foundation for Autonomous Agents and Multiagent Systems},
  address={Richland, SC, USA},
  url={https://www.ifaamas.org/Proceedings/aamas2016/pdfs/p1363.pdf}
}

@article{bellemare2013arcade,
  title={The arcade learning environment: An evaluation platform for general agents},
  author={Bellemare, Marc G and Naddaf, Yavar and Veness, Joel and Bowling, Michael},
  journal={Journal of artificial intelligence research},
  volume={47},
  pages={253--279},
  year={2013}
}

@inproceedings{ConitzerS06,
  author    = {Vincent Conitzer and
               Tuomas Sandholm},
  title     = {Computing the optimal strategy to commit to},
  booktitle = {Proceedings of the 7th ACM Conference on Electronic Commerce},
  pages     = {82--90},
  year      = {2006},
  publisher = {Association for Computing Machinery},
  address   = {New York, NY, USA},
  doi       = {10.1145/1134707.1134717},
  url       = {https://doi.org/10.1145/1134707.1134717}
}

@inproceedings{HansenBZ04,
  author    = {Eric A. Hansen and
               Daniel S. Bernstein and
               Shlomo Zilberstein},
  title     = {Dynamic Programming for Partially Observable Stochastic Games},
  booktitle = {Proceedings of the Nineteenth National Conference on Artificial Intelligence},
  pages     = {709--715},
  publisher = {{AAAI} Press / The {MIT} Press},
  address   = {Menlo Park, CA, USA},
  year      = {2004}
}

@inproceedings{BaiJWX21,
  author    = {Yu Bai and
               Chi Jin and
               Huan Wang and
               Caiming Xiong},
  title     = {Sample-Efficient Learning of {Stackelberg} Equilibria in General-Sum
               Games},
  booktitle = {Advances in Neural Information Processing Systems 34},
  pages     = {25799--25811},
  year      = {2021},
  publisher = {Curran Associates, Inc.},
  address   = {Red Hook, NY, USA}
}

@article{wang1999commodity,
  title={Commodity taxes under fiscal competition: {Stackelberg} equilibrium and optimality},
  author={Wang, You-Qiang},
  journal={American Economic Review},
  volume={89},
  number={4},
  pages={974--981},
  year={1999}
}

@article{fang2016deploying,
  title={Deploying {PAWS: Field} Optimization of the Protection Assistant for Wildlife Security},
  author={Fang, Fei and Nguyen, Thanh Hong and Pickles, Rob and Lam, Wai Y and Clements, Gopalasamy R and An, Bo and Singh, Amandeep and Tambe, Milind and Lemieux, Andrew},
  journal={Proceedings of the AAAI Conference on Artificial Intelligence},
  volume={30},
  number={2},
  pages={3966--3973},
  year={2016},
  doi={10.1609/aaai.v30i2.19070},
  url={https://doi.org/10.1609/aaai.v30i2.19070}
}

@inproceedings{sinha2018stackelberg,
  title={Stackelberg Security Games: Looking Beyond a Decade of Success},
  author={Sinha, Arunesh and Fang, Fei and An, Bo and Kiekintveld, Christopher and Tambe, Milind},
  booktitle={Proceedings of the Twenty-Seventh International Joint Conference on Artificial Intelligence},
  pages={5494--5501},
  year={2018},
  publisher={International Joint Conferences on Artificial Intelligence Organization},
  address={Stockholm, Sweden},
  doi={10.24963/ijcai.2018/775},
  url={https://www.ijcai.org/proceedings/2018/775}
}

@article{zhou2019would,
  title={Would an increasing block carbon tax be better? {A} comparative study within the {Stackelberg} Game framework},
  author={Zhou, Dequn and An, Yunfei and Zha, Donglan and Wu, Fei and Wang, Qunwei},
  journal={Journal of environmental management},
  volume={235},
  pages={328--341},
  year={2019},
}

@inproceedings{LoweWTHAM17,
  author    = {Ryan Lowe and
               Yi Wu and
               Aviv Tamar and
               Jean Harb and
               Pieter Abbeel and
               Igor Mordatch},
  title     = {Multi-Agent Actor-Critic for Mixed Cooperative-Competitive Environments},
  booktitle = {Advances in Neural Information Processing Systems 30},
  pages     = {6379--6390},
  year      = {2017},
  publisher = {Curran Associates, Inc.},
  address   = {Red Hook, NY, USA}
}

@inproceedings{hartline2015no,
  title={{No-regret learning in Bayesian games}},
  author={Hartline, Jason and Syrgkanis, Vasilis and Tardos, Eva},
  booktitle={Advances in Neural Information Processing Systems},
  volume={28},
  pages={3061--3069},
  year={2015},
  publisher={Curran Associates, Inc.},
  address={Red Hook, NY, USA},
  url={https://proceedings.neurips.cc/paper/2015/hash/3e7e0224018ab3cf51abb96464d518cd-Abstract.html}
}

@article{stable-baselines3,
  author  = {Antonin Raffin and Ashley Hill and Adam Gleave and Anssi Kanervisto and Maximilian Ernestus and Noah Dormann},
  title   = {Stable-Baselines3: Reliable Reinforcement Learning Implementations},
  journal = {Journal of Machine Learning Research},
  year    = {2021},
  volume  = {22},
  number  = {268},
  pages   = {1-8},
}

@article{brero2021reinforcement,
  author    = {Gianluca Brero and
               Alon Eden and
               Matthias Gerstgrasser and
               David C. Parkes and
               Duncan Rheingans{-}Yoo},
  title     = {Reinforcement Learning of Sequential Price Mechanisms},
  journal   = {Proceedings of the AAAI Conference on Artificial Intelligence},
  volume    = {35},
  number    = {6},
  pages     = {5219--5227},
  year      = {2021},
  doi       = {10.1609/aaai.v35i6.16659},
  url       = {https://doi.org/10.1609/aaai.v35i6.16659}
}

@inproceedings{AgrawalSZ20,
	Author = {Shipra Agrawal and Jay Sethuraman and Xingyu Zhang},
	Booktitle = {Proceedings of the 21st ACM Conference on Economics and Computation},
	Pages = {187--188},
	Title = {On Optimal Ordering in the Optimal Stopping Problem},
	Year = {2020},
	Publisher = {Association for Computing Machinery},
	Address = {New York, NY, USA},
	Doi = {10.1145/3391403.3399484},
	Url = {https://doi.org/10.1145/3391403.3399484}}

@article{roughgarden2015intrinsic,
  doi = {10.1145/2806883},
  year = {2015},
  month = nov,
  volume = {62},
  number = {5},
  pages = {1--42},
  author = {Tim Roughgarden},
  title = {Intrinsic Robustness of the Price of Anarchy},
  journal = {Journal of the {ACM}}
}

@inproceedings{gerstgrasser2023oracles,
  title={Oracles \& followers: Stackelberg equilibria in deep multi-agent reinforcement learning},
  author={Gerstgrasser, Matthias and Parkes, David C},
  booktitle={Proceedings of the 40th International Conference on Machine Learning},
  series={Proceedings of Machine Learning Research},
  volume={202},
  pages={11213--11236},
  year={2023},
  publisher={PMLR},
  address={Honolulu, HI, USA},
  url={https://proceedings.mlr.press/v202/gerstgrasser23a.html}
}

@article{cartea2022learning,
  title={Algorithmic Collusion and a Folk Theorem from Learning with Bounded Rationality},
  author={Cartea, {\'A}lvaro and Chang, Patrick and Penalva, Jos{\'e} and Waldon, Harrison},
  journal={Games and Economic Behavior},
  volume={157},
  pages={1--21},
  year={2026},
  doi={10.1016/j.geb.2025.11.012},
  url={https://doi.org/10.1016/j.geb.2025.11.012}
}

@article{brown2023competition,
  title={Competition in pricing algorithms},
  author={Brown, Zach Y and MacKay, Alexander},
  journal={American Economic Journal: Microeconomics},
  volume={15},
  number={2},
  pages={109--156},
  year={2023},
}

@article{rolf2023review,
  title={A review on reinforcement learning algorithms and applications in supply chain management},
  author={Rolf, Benjamin and Jackson, Ilya and M{\"u}ller, Marcel and Lang, Sebastian and Reggelin, Tobias and Ivanov, Dmitry},
  journal={International Journal of Production Research},
  volume={61},
  number={20},
  pages={7151--7179},
  year={2023},
}

@book{train2009discrete,
  author = {Train, Kenneth E.},
  title = {Discrete Choice Methods with Simulation},
  edition = {2},
  year = {2009},
  publisher = {Cambridge University Press},
  url = {https://eml.berkeley.edu/books/choice2.html}
}

@article{owen2007pigeonhole,
  author = {Owen, Art B.},
  title = {The pigeonhole bootstrap},
  journal = {The Annals of Applied Statistics},
  volume = {1},
  number = {2},
  pages = {386--411},
  year = {2007},
  doi = {10.1214/07-AOAS122}
}

\clearpage
\appendix

\section{Theoretical Foundations}
\label{app:theory}

This appendix formalizes induced subgames, response-phase credit assignment, no-regret follower responses, and cached commitments.

\subsection{Formal Definition of Subgame of a POMG}\label{app:subgame}
\begin{definition}[Subgame of a POMG]
Given a POMG $\mathcal G$ on agent set $N$, a proper subset $I\subsetneq N$, and fixed policies $\pi_I$, let $M=N\setminus I$ be the remaining agents. The induced POMG $\mathcal G_{\pi_I}$ has the same horizon and action and observation spaces for agents in $M$. Its construction is as follows:
\begin{itemize}
    \item An augmented state is $x=(s,h_I,o_M)$, where $h_I$ records the fixed agents' observation histories and $o_M$ records the remaining agents' current observations. Agent $i\in M$ observes only its component $o_i$. Recording these observations preserves their dependence on the hidden state and the fixed agents' sampled actions without revealing additional information.
    \item Given $x$ and the remaining agents' action profile $a_M$, sample $a_i\sim\pi_i(\cdot\mid h_i)$ independently for $i\in I$ and run the original game's transition and observation rules with the full profile $a=(a_I,a_M)$. Append each new observation $o_i'$ to $h_i$ for $i\in I$, and retain $o_M'$ in the next augmented state. Thus transitions average over the fixed agents' sampled actions and the joint observation draw.
    \item More explicitly, $K(s',o'\mid s,a)=\tran(s,a,s')\mathcal O(s',a,o')$ is the joint next-state and next-observation law in Definition~\ref{def:pomg}. For $x'=(s',h_I',o_M')$, the induced transition probability is
    \begin{equation*}
    \tran_{\pi_I}(x,a_M,x')=
    \sum_{a_I}\sum_{o_I'}
    \left[\prod_{i\in I}\pi_i(a_i\mid h_i)\right]
    K\bigl(s',(o_I',o_M')\mid s,(a_I,a_M)\bigr)
    \mathbf 1\{h_I'=h_I\mathbin{\|}o_I'\},
    \end{equation*}
    where $\mathbin{\|}$ appends the observations componentwise. The observation rule returns each agent's recorded component of the current augmented state.
    \item For $i\in M$, the induced reward function is the conditional expected original reward,
    \begin{equation*}
    R_{\pi_I,i}(x,a_M)=\sum_{a_I}
    \left[\prod_{j\in I}\pi_j(a_j\mid h_j)\right]
    R_i\bigl(s,(a_I,a_M)\bigr).
    \end{equation*}
    Sampling the original reward during simulation has this same conditional expectation.
    \item The initial augmented-state distribution is induced by $(s_0,o_0)\sim\rho_0$, recording the initial observations in $h_I$ and $o_M$.
\end{itemize}
The policies $\pi_I$ remain fixed throughout play; their prescribed actions may still be random.
\end{definition}

\subsection{Response-Phase Credit Assignment}
\label{app:response_credit}

The formal statement below makes the sampling assumptions explicit. Its proof separates the likelihood contributions of response queries and gameplay, and the subsequent example shows why omitting the former can change the sign of the expected update.

\paragraph{Formal statement of Proposition~\ref{prop:response_credit}.}
Consider a Stackelberg POMDP with finite state, action, and observation spaces and fixed horizons $E,T$, initialization, and transition kernels. Let $\pi_\vartheta\in\Pi_{\mathcal S}^{0}$ have continuously differentiable, strictly positive action probabilities on an open parameter set. Hold $\vartheta$ fixed within each episode and draw a fresh action from $\pi_\vartheta(\cdot\mid o_t)$ at each call in both phases.

Write a full Stackelberg-POMDP trajectory as
\begin{equation*}
\xi=(s_0,o_0,a_0,\ldots,s_{E+T-1},o_{E+T-1},a_{E+T-1},s_{E+T}),
\end{equation*}
where $s_t$ is the POMDP state, $o_t$ the leader's observation, and $a_t$ its sampled action at step $t$. The first $E$ observation--action pairs are response queries and answers; the remaining $T$ are gameplay. Let $R(\xi)$ denote the leader's realized reward-phase return; response-phase rewards are zero. Define the response-query and gameplay score vectors
\begin{equation*}
\begin{aligned}
S_{\mathrm{resp}}(\xi)&=\sum_{e=0}^{E-1}\nabla_\vartheta\log\pi_\vartheta(a_e\mid o_e),\\
S_{\mathrm{play}}(\xi)&=\sum_{t=E}^{E+T-1}\nabla_\vartheta\log\pi_\vartheta(a_t\mid o_t).
\end{aligned}
\end{equation*}

The expected return $J(\vartheta)=\E_\vartheta[R(\xi)]$ is the Stackelberg objective of the corresponding leader commitment under the specified follower-response oracle $\mathcal B$ (Definition~\ref{def:leaderProblem}). Let $\widehat g_{\mathrm{play}}(\xi)=R(\xi)S_{\mathrm{play}}(\xi)$ denote the gameplay-only REINFORCE gradient estimator. Then
\begin{equation}
\nabla_\vartheta J(\vartheta)
=\E_\vartheta[\widehat g_{\mathrm{play}}(\xi)]
+\E_\vartheta[R(\xi)S_{\mathrm{resp}}(\xi)].
\label{eq:response_gradient}
\end{equation}
Thus $\widehat g_{\mathrm{play}}$ is unbiased exactly when the second expectation vanishes.

\begin{proof}[Proof of Proposition~\ref{prop:response_credit}]
For a full trajectory $\xi$, its probability factors as
\begin{equation*}
\Pr_\vartheta(\xi)=c(\xi)\prod_{t=0}^{E+T-1}\pi_\vartheta(a_t\mid o_t),
\end{equation*}
where $c(\xi)$ contains the fixed initialization, observation, and transition probabilities, including the response algorithm's internal randomness. The selected follower profile varies across trajectories but is fixed conditional on the final response-algorithm state recorded in $\xi$. The PI interface ensures that $c(\xi)$ has no additional dependence on $\vartheta$. Holding the sampled trajectory fixed, the product rule gives
\begin{equation*}
\nabla_\vartheta\Pr_\vartheta(\xi)
=\Pr_\vartheta(\xi)\bigl(S_{\mathrm{resp}}(\xi)+S_{\mathrm{play}}(\xi)\bigr).
\end{equation*}
The finite spaces and horizons make the trajectory sum finite. Differentiating $J(\vartheta)=\sum_\xi\Pr_\vartheta(\xi)R(\xi)$ term by term yields Equation~\eqref{eq:response_gradient}, and subtracting $\nabla_\vartheta J$ from $\E_\vartheta[\widehat g_{\mathrm{play}}]$ gives bias $-\E_\vartheta[R S_{\mathrm{resp}}]$.
\end{proof}

The reward-to-go version has the same expectation: rewards preceding an action have zero expected product with that action's conditional score. Hence replacing full returns by reward-to-go in the gameplay terms does not recover the omitted response contribution.

\paragraph{A one-query illustration.}
Consider a one-stage game with one leader and one follower, each choosing an action in $\{0,1\}$ at a single initial state. For leader action $a$ and follower action $b$, the payoffs are
\begin{equation*}
R_\ell(a,b)=2b-a,
\qquad R_f(a,b)=\mathbf{1}\{a=b\}.
\end{equation*}
The leader commits to $\pi_p(1\mid o)=p$, with $0<p<1$. The follower-response procedure queries this policy once at the initial observation $o$, obtains $X\sim\mathrm{Bernoulli}(p)$, and returns the deterministic follower policy $b=X$. This is the best response to the one-sample empirical distribution of the leader's action.

The response procedure is reset before every episode. Its query receives zero leader reward, and its returned follower policy is frozen for gameplay. The leader then receives the same observation $o$ and draws a fresh action $A\sim\mathrm{Bernoulli}(p)$, independently of $X$. Thus the Stackelberg POMDP has $E=1$ and $T=1$, and the same memoryless policy supplies both draws. Under this response procedure,
\begin{equation*}
R=2X-A,
\qquad J(p)=\E[2X-A]=p,
\qquad J'(p)=1.
\end{equation*}
For a Bernoulli draw $Y$, the score is $\partial_p\log\Pr_p(Y)=(Y-p)/(p(1-p))$. Independence of $X$ and $A$ therefore gives
\begin{equation*}
\begin{aligned}
\E[\widehat g_{\mathrm{play}}]
&=\frac{\E[(2X-A)(A-p)]}{p(1-p)}=-1,\\
\E[R S_{\mathrm{resp}}]
&=\frac{\E[(2X-A)(X-p)]}{p(1-p)}=2.
\end{aligned}
\end{equation*}
The expected gameplay-only update points downhill at every interior $p$, while including the query score gives the correct derivative $-1+2=1$. With a logistic parameterization $p=\sigma(z)=1/(1+e^{-z})$, all three quantities are multiplied by $p(1-p)>0$, so the sign reversal persists.

The corresponding controlled learning experiment and protocol appear in Appendix~\ref{sec:response_credit_ablation}.

\subsection{No-Regret Follower Responses}
\label{app:mw_alg}

Fixing the leader's mechanism in Section~\ref{subsec:indirect_mechanism_noregret} induces a Bayesian game over follower messages: the committed mechanism determines the utility of each message profile. Multiplicative weights (MW) learns type-dependent reporting strategies in this induced game. We define the equilibrium condition used to certify MSPM responses and describe the MW procedure used in the MSPM experiments. Experimental parameters and evaluation protocols appear in Appendix~\ref{app:mechanism_details}.

\paragraph{Equilibrium target.}
Consider $n$ players with finite type spaces $\Theta_i$, finite action spaces $A_i$, and utility functions $u_i(\theta;a)$. The type profile $\theta$ is drawn from $\mathcal D$; each player privately observes $\theta_i$ and the players choose their actions simultaneously. A correlated strategy draws a profile of type-contingent action rules independently of $\theta$. Write $q(a\mid\theta)$ for the resulting conditional joint action distribution. In the mechanism-design domains, $a_i=\mu_i$ and $u_i(\theta;a)=u_i^{\pi_\ell}(\theta;\mu)$, where $\pi_\ell$ is the fixed leader policy.

\begin{definition}[Bayesian coarse correlated equilibrium]
\label{def:bcce}
For $\varepsilon\geq 0$, $q$ is an $\varepsilon$-Bayesian coarse correlated equilibrium ($\varepsilon$-BCCE) if, for every player $i\in\{1,\ldots,n\}$, every type $\theta_i\in\Theta_i$ with positive probability, and every alternative action $a_i'\in A_i$,
\[
\mathbb{E}_{\substack{\theta_{-i}\sim\mathcal{D}(\cdot\mid\theta_i)\\
a\sim q(\cdot\mid\theta)}}
\left[
u_i\!\left(\theta;(a_i',a_{-i})\right)
-
u_i(\theta;a)
\right]
\leq\varepsilon.
\]
A BCCE is the exact case $\varepsilon=0$.
\end{definition}

The deviation $a_i'$ may depend on player $i$'s type but is fixed before the action profile is sampled. Thus, no type can gain more than $\varepsilon$ in expected utility by committing to an alternative action. With one type per player, the definition reduces to coarse correlated equilibrium (CCE) and its $\varepsilon$-approximation.

\paragraph{No-regret learning and convergence.}
No-regret learning means that, for each player type, the largest positive average utility gain from replacing its past choices with a fixed action chosen in hindsight vanishes. For independent types, when every player type has vanishing external regret, the empirical distribution of play approaches the BCCE set \citep{hartline2015no}. This convergence concerns the empirical distribution over play; the current strategy profile need not converge.

\paragraph{Multiplicative-weights response.}
MW is one standard no-regret algorithm: for bounded utilities, a suitable learning-rate choice yields vanishing average external regret \citep{arora2012multiplicative}. We use MW to generate the follower responses in the MSPM experiments.

Algorithm~\ref{alg:MW} uses an update parameter $\eta>0$, distinct from the equilibrium tolerance $\varepsilon$, and runs for $M$ iterations. Each player maintains normalized type-contingent weights initialized uniformly. Each iteration samples one type profile and evaluates every joint action profile. The resulting utilities are averaged exactly against the opponents' current mixed strategies before the sampled type rows are updated. Thus, one complete MW update evaluates $\prod_i |A_i|$ joint-action profiles.

\begingroup
\setlength{\algomargin}{0pt}
\begin{algorithm}[H]
\caption{The Multiplicative-Weights Algorithm}
\Description{Pseudocode for multiplicative weights in a Bayesian simultaneous-move game with type-contingent action weights.}
\label{alg:MW}
\begin{algorithmic}[1]
    \REQUIRE utility functions $u_i$, type distribution $\mathcal D$, update parameter $\eta$, iterations $M$
    \STATE Initialize $w_i[\type_i][a_i] = 1/|A_i|$ for every player, type, and action \label{alg:weightsInitialization}
    \FOR {$M$ times} \label{alg:startMWIterations}
        \STATE Sample type profile $\types=(\type_1,\ldots,\type_n)$ from distribution $\mathcal D$ \label{alg:sampleTypes}
        \STATE Freeze $p_i(a_i)=w_i[\type_i][a_i]$ for all players and actions
        \STATE Evaluate $u_i(\types;a)$ for every joint action $a\in\prod_j A_j$ and player $i$ \label{alg:sampleActions}
        \FOR {each player $i\in\{1,\ldots,n\}$} \label{alg:beginForAgents}
            \FOR {each action $a_i'\in A_i$} \label{alg:beginForActions}
                \STATE $\bar u_i(a_i')\leftarrow\sum_{a_{-i}}\bigl(\prod_{j\ne i}p_j(a_j)\bigr)u_i(\types;(a_i',a_{-i}))$
                \STATE $w_i[\type_i][a_i']\leftarrow w_i[\type_i][a_i'](1+\eta)^{\bar u_i(a_i')}$ \label{alg:weightsUpdate}
            \ENDFOR \label{alg:endForActions}
            \STATE Normalize the updated row $w_i[\type_i]$ to sum to one
        \ENDFOR \label{alg:endForAgents}
    \ENDFOR
\end{algorithmic}
\end{algorithm}
\endgroup

Each joint-message evaluation runs the committed mechanism with the selected type and message profiles. The leader is accessed only through its actions at the resulting observation histories, so the MW response can be implemented through the policy-interactive interface. The domain-specific constructions in Appendix~\ref{app:mechanism_details} specify the response state, budget, and rule for selecting the returned reporting strategy.

\paragraph{Finite-response certification.}
The experiments use finite runs with a fixed MW update rate. We therefore verify the BCCE condition directly for candidate MSPM responses.
Specializing Definition~\ref{def:bcce} to messages $a_i=\mu_i$ and utilities $u_i^{\pi_\ell}$, the interim coarse-deviation gap of a candidate reporting strategy $q(\mu\mid\theta)$ is
\begin{equation}
\Delta_{\mathrm{BCCE}}(q;\pi_\ell)
=
\max\!\left\{0,\,
\max_{i,\theta_i,\mu_i'}
\mathbb{E}_{\substack{\theta_{-i}\sim\mathcal{D}(\cdot\mid\theta_i)\\
\mu\sim q(\cdot\mid\theta)}}
\left[
u_i^{\pi_\ell}\!\left(\theta;(\mu_i',\mu_{-i})\right)
-
u_i^{\pi_\ell}(\theta;\mu)
\right]
\right\}.
\label{eq:bcce_gap}
\end{equation}
The maximum ranges over followers, their positive-probability types, and alternative messages, exactly as in Definition~\ref{def:bcce}. Hence, $q$ is an $\varepsilon$-BCCE if and only if $\Delta_{\mathrm{BCCE}}(q;\pi_\ell)\leq\varepsilon$. This gap allows finite responses to be certified directly, without relying on an asymptotic no-regret guarantee.

\subsection{Cached Commitments}
\label{app:cached_commitments}

Given a PI response algorithm $\eqalg$ with query bound $E$, include an initially empty answer cache in its auxiliary state. The cached algorithm simulates $\eqalg$ internally until it returns or requests a previously unseen history. In the latter case, it queries the leader, stores the fresh answer, and resumes the simulation; repeated requests are answered from the cache. Initialization and internal randomness are otherwise those of $\eqalg$, so the algorithm's components remain independent of the supplied leader policy. Each new call corresponds to a query of the original algorithm, giving at most $E$ calls before the cached algorithm returns; ignored padding can be added as in Definition~\ref{def:PIAlgorithm}. Thus the cached algorithm is also PI.

For a finite query space, an episode can sample a deterministic commitment while revealing its actions only as needed. Let $\mathcal Q$ be the finite set of leader query inputs and hold the actor parameters $\vartheta$ fixed throughout an episode. The augmented environment state contains an initially empty action cache. On the first request for $o\in\mathcal Q$, the actor draws an action from $\pi_\vartheta(\cdot\mid o)$ and stores it; later requests reuse that answer. The response and evaluation procedures access the actor only through these requests. Each new draw depends only on its query input and uses randomness independent of other draws and of the procedures' internal randomness. The cache persists into reward evaluation, including entries first requested there, and is cleared between episodes.

\begin{lemma}[Sampling a commitment on first use]
\label{lem:cached_commitment}
Let $q_\vartheta$ be the distribution over deterministic commitments obtained by independently drawing $f(o)\sim\pi_\vartheta(\cdot\mid o)$ for every $o\in\mathcal Q$. Sampling these entries on first use has the same joint distribution of follower response and reward evaluation as sampling the complete commitment $f$ before the episode and revealing its entries as requested. Let $V(f)$ be the expected return under the specified response and evaluation procedure for $f$. If that procedure terminates almost surely with an integrable return, then
\begin{equation}
J_{\mathrm{cache}}(\vartheta)=\mathbb E_{f\sim q_\vartheta}[V(f)].
\label{eq:cached_objective}
\end{equation}
\end{lemma}

\begin{proof}
Couple the procedures by drawing the complete $f$ in advance and using the same internal response and evaluation randomness. At every first request, reveal the corresponding entry of $f$. Conditional on all previously revealed entries, an unqueried entry still has distribution $\pi_\vartheta(\cdot\mid o)$. Thus the procedures give identical answers and follow identical trajectories, even when queries are chosen adaptively. Entries never requested have no effect. Averaging the conditional return over $f$ gives Equation~\eqref{eq:cached_objective}.
\end{proof}

When the response procedure is a PI algorithm and, conditional on the returned follower profile, reward evaluation has the expected return of the fresh game in Definition~\ref{def:StackPOMDP}, Theorem~\ref{thm:main} identifies $V(f)$ with the deterministic commitment's value under $\mathcal B$.

Proposition~\ref{prop:response_credit} analyzes fresh action sampling in both phases. In the cached construction, policy-gradient score terms correspond to actual action draws; cache reuse contributes no additional score term.

\section{Environment Specifications and Experimental Details}
\label{app:environments}

For each of the three applications in Section~\ref{sec:experiments}, we describe the environment and information structure, the follower response and StackPOMDP construction, and the training and evaluation protocol.

\paragraph{Computing environment.}
Experiments used CPU computation on the Unity HPC cluster, running Ubuntu~24.04 LTS, and on a local Apple M3 Max with 16 CPU cores and 128\,GB of unified memory, running macOS~15.1. The local machine trained the shared Atari gameplay and meta-follower checkpoints. The final Atari leader cohort used Unity nodes with AMD EPYC 7502 or 7543 processors, or Intel Xeon E5-2620 v3, Gold 5220, Gold 6148, Gold 6238R, or Platinum 8352Y processors. These dual-socket nodes provide 120--500\,GB of RAM; each leader run was allocated five CPU cores and 24\,GB.\footnote{Unity documents its operating system at \url{https://unityhpc.org/about/} and node hardware at \url{https://unityhpc.org/documentation/cluster_specs/nodes/} (accessed October 5, 2026).} Cluster runs used Python~3.9.25, while local Atari pretraining used Python~3.9.16. Both environments used NumPy~1.23.0, PyTorch~2.0.0, Gym~0.21.0, and Stable-Baselines3~1.8.0~\citep{stable-baselines3}; Atari additionally used \texttt{multi-agent-ale-py}~0.1.10 and OpenCV~4.9.0.80. The reproducibility package provides the environment specification and run configurations~\citep{GlcBrero_StackelbergPOMDP}.

\paragraph{Shared training setup.}
For all PPO training in these applications, we set the discount factor and generalized advantage estimation parameter to $\gamma=\lambda=1$ and collect rollouts of complete episodes. In leader training, these are complete Stackelberg POMDP episodes; computations performed internally by the response algorithm are accounted for separately from leader transitions. Unless stated otherwise, PPO hyperparameters follow the library defaults.

\subsection{Indirect Mechanism Design with No-Regret Followers}
\label{app:mechanism_details}

\subsubsection{Environment and POMG Specification}

The MSPM environment has the following POMG components:
\begin{itemize}
    \item \textbf{Agents.} One mechanism designer (leader) and two buyers (followers). Buyer $i\in\{1,2\}$ has a private value $\theta_i\in\Theta_i\subseteq[0,1]$ for the item, where $\Theta_i$ is finite.
    \item \textbf{State space $S$.} The buyers' values and messages, current allocation and payment, and the set of buyers not yet visited.
    \item \textbf{Action spaces $A_i$.} Buyer $i$ first sends a message $\mu_i\in\{1,2\}$ and, when visited, accepts exactly when $\theta_i>p_i$. The leader chooses the next unvisited buyer and a posted price $p_i\in[0,1]$.
    \item \textbf{Observation spaces $O_i$.} Each buyer observes its own value and any price offered to it. The leader observes the messages, allocation status, and remaining buyers, but not their values.
    \item \textbf{Transitions $\tran$.} Values are sampled and messages sent before sequential offers begin. Acceptance allocates the item and ends the mechanism. After rejection, the buyer leaves the remaining set and the mechanism may visit the other buyer.
    \item \textbf{Rewards $R_i$.} A buyer receives value minus payment if allocated and zero otherwise. Following \citet{brero2021reinforcement}, the leader's training reward is allocated welfare minus ex post efficient welfare, $W-W^*$.
    \item \textbf{Horizon $T$.} One messaging stage followed by at most two sequential offer decisions.
\end{itemize}

\subsubsection{Follower Response and StackPOMDP Construction}
\label{app:stack_mspm}

Algorithm~\ref{alg:MW} treats messages as follower actions, with update parameter $\eta=0.1$ and initially uniform weights. The response algorithm records the type-contingent weights, update and joint-profile counters, sampled values, and accumulated counterfactual payoffs.

\paragraph{Queries and updates.}
Response queries use values and messages from the MW sweep and receive zero leader reward. Each episode requires at most eight distinct leader-policy queries. Their answers are cached and reused while multiplicative weights continues until the response is certified as a $0.05$-BCCE. Each update evaluates all four joint-message profiles under the leader commitment, computes utilities against the other buyer's mixed strategy, updates the sampled type rows, and draws the next value profile.

\paragraph{Stopping and returned response.}
After each complete update, we compute the gap in Equation~\eqref{eq:bcce_gap} exactly for three candidates, in order: the current mixed strategy, the deterministic strategy selecting each type's highest-weight message, and the empirical average. We accept the first candidate with gap at most $0.05$. The empirical average is motivated by no-regret convergence; the current mixed and pure strategies can yield earlier certification as they are not affected by weights initializations. This tolerance bounds expected-utility gains on the normalized payoff scale. We allow at most 10,000 additional MW updates after the first; if no candidate is certified within this budget, the run terminates with an error.

\paragraph{StackPOMDP construction.}
The StackPOMDP augments the mechanism state with the internal MW state, action cache, and response/reward phase. Because the type and message spaces are small, the reward phase enumerates all profiles, weighted by the value distribution $\mathcal D$ and the certified follower response. This computes the mechanism's expected reward exactly, removing sampling noise from reward evaluation. The actor receives the same mechanism observation in both phases: messages, allocation status, and remaining buyers. Buyer values, the MW state, and the phase are hidden from the actor and available to the centralized critic.

\subsubsection{Implementation, Training, and Evaluation}
\label{subapp:noregret_impl_mspm}

Figure~\ref{fig:mspm_results} uses 25 fresh seeds (96--120) for each type count, giving 125 completed runs. The configuration was frozen before observing these outcomes. Each episode samples a deterministic cached leader commitment; the PPO buffer retains one action draw per distinct observation. We use a full-state critic. Rollouts contain 40 episodes. We selected an initial action standard deviation of $0.5$ after reviewing earlier runs with that value and a separate pilot with standard deviation $1.0$. The learning rate is $3\times10^{-4}$ through episode 3,000, then decays linearly to zero. Each run trains for 10,000 complete episodes. Values follow Example~\ref{ex:mspm-two-type} for two types and the independent uniform grids in the main text for types 3--6.

\paragraph{Evaluation.}
We use the Gaussian policy's mean action at each observation, clipped to the action bounds, and evaluate exactly at episode zero and every 200 episodes through episode 10,000. Figure~\ref{fig:mspm_results}(a) plots unsmoothed current-policy welfare with one-standard-error shading. For each training run, we select the checkpoint with the highest exactly evaluated welfare and report the average across runs, with sample standard errors, in Figure~\ref{fig:mspm_results}(b). All scheduled evaluations passed the $0.05$-BCCE certificate within the stopping budget. A checkpoint attains ex post efficient welfare when $W=W^*$. Complete training histories are included in the reproducibility data.

\paragraph{Training dynamics.}
Figure~\ref{fig:mspm_multitype_training} complements the best-checkpoint summaries in Figure~\ref{fig:mspm_results}(b) by showing current-policy welfare throughout training on the three- through six-type grids. We plot the gap $W-W^*$ so that zero marks ex post efficient welfare in every panel. Mean welfare improves beyond the exact SPM benchmark on every grid, while the five- and six-type instances retain visible gaps to ex post efficiency at the end of training.

\begin{figure}[!tbp]
\centering
\includegraphics[width=0.95\linewidth]{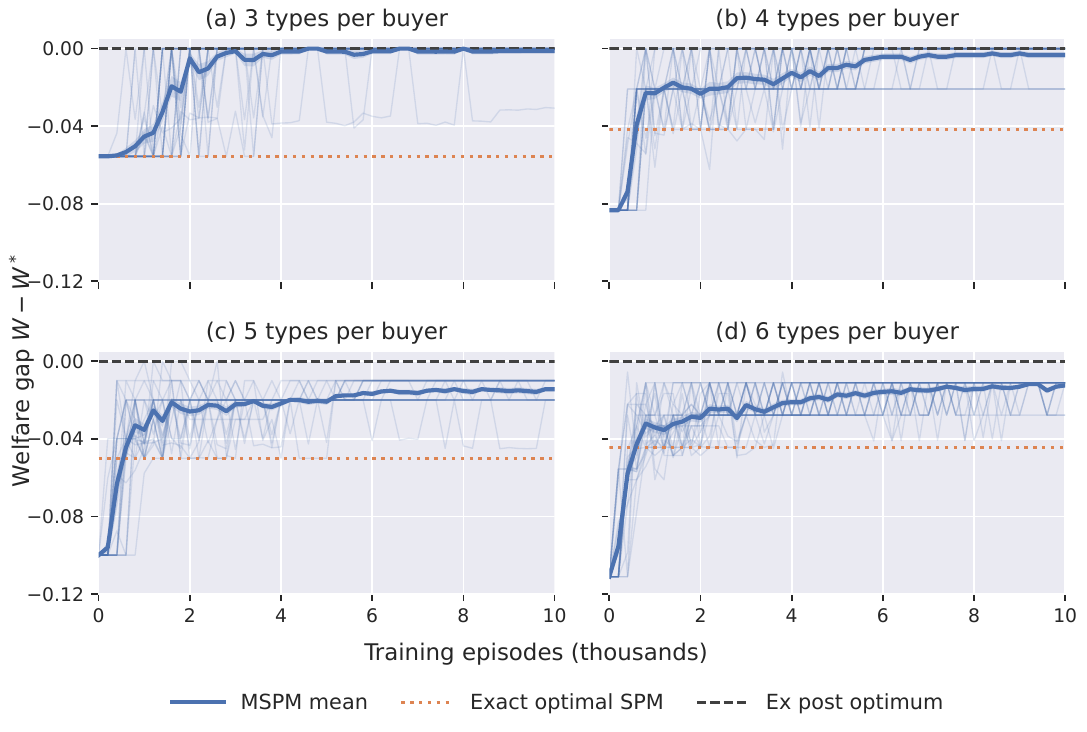}
\caption{MSPM training on the uniform value grids with three through six types per buyer. Welfare gaps $W-W^*$ are evaluated exactly every 200 episodes, without smoothing; zero denotes ex post efficient welfare. Shading is one standard error. Dotted lines mark the exact optimal SPM benchmarks.}
\Description{Four panels show training with three, four, five, and six types per buyer on common axes. Faint blue lines show individual runs, the thicker blue line their mean, and shading one standard error. Welfare gaps rise toward zero from below. Orange dotted lines mark exact optimal SPM welfare gaps and gray dashed lines mark the ex post optimum at zero.}
\label{fig:mspm_multitype_training}
\end{figure}

\paragraph{Exact SPM benchmark.}
On each symmetric grid, either buyer may be visited first. If $k$ is the lowest value index accepting the first price, followed after rejection by a zero price to the other buyer, expected welfare loss is
\[
\frac{\binom{\bar\theta-k+1}{3}+\binom{k+1}{3}}
{\bar\theta^2(\bar\theta-1)}.
\]
Minimizing over $k$ and subtracting from $W^*$ gives exact SPM welfare $2/3$, $2/3$, $13/20$, and $13/20$ for types 3--6. The two-type benchmark follows Example~\ref{ex:mspm-two-type}.

\subsection{Platform Intervention with Hidden Costs}
\label{app:collusion_details}

We extend the demand model of \citet{calvano2020artificial}---with refinements by \citet{johnson2023platform}---using hidden common costs, and a learned symmetric display rule.

\subsubsection{Environment and POMG Specification}
\label{subapp:econ_model}

At submitted prices $p_t=(p_{1,t},p_{2,t})$, the platform displays a set $B_t\subseteq\{1,2\}$. For each displayed seller $i\in B_t$, logit demand and profit are
\begin{equation}
D_i(p_t;B_t)=\frac{\exp((a_i-p_{i,t})/\mu)}{\exp(a_0/\mu)+\sum_{j\in B_t}\exp((a_j-p_{j,t})/\mu)},
\qquad \pi_{i,t}=(p_{i,t}-c)D_i(p_t;B_t),
\label{eq:logit_demand_app}
\end{equation}
with $D_i=0$ for excluded sellers. Here, $a_i$ and $a_0$ denote product and outside-option qualities, $\mu$ is the logit scale, and $c$ is the common marginal cost.

Under this logit model, consumer surplus is the utility from a consumer's preferred displayed product or the outside option, net of any payment. Averaging over consumer tastes and omitting an additive constant independent of prices and the displayed set gives \citep[Section~3.5]{train2009discrete}:
\begin{equation}
U(p_t;B_t)=\mu\log\!\left(\exp(a_0/\mu)+\sum_{j\in B_t}\exp((a_j-p_{j,t})/\mu)\right).
\label{eq:consumer_surplus}
\end{equation}
Thus, $U$ captures the benefit of lower prices and, at fixed prices, a wider choice of products. We use it as the platform's reward.

We use $a_1=a_2=2$, $a_0=0$, and $\mu=0.25$; displaying neither seller then gives $U=0$. Sellers choose among 15 equally spaced prices from $1.0$ to $2.1$:
\[
\mathcal P=\{p^{(0)},\ldots,p^{(14)}\},\qquad
p^{(k)}=1.0+\frac{1.1k}{14}.
\]
The POMG has the following components:
\begin{itemize}
    \item \textbf{Agents.} One platform (leader) and two sellers (followers), each offering one differentiated product.
    \item \textbf{State space $S$.} The common marginal cost $c\in\{1.0,1.5\}$ and decision stage (pricing or display), together with the submitted price pair $p_t\in\mathcal P^2$ at the display stage. The two costs are equally likely; cost stays fixed throughout play.
    \item \textbf{Action spaces $A_i$.} Seller $i$ chooses $p_{i,t}\in\mathcal P$. The platform chooses an action in $\{0,1,2\}$ determining the displayed set $B_t$, as defined below. The same price grid is available at both costs, including prices below cost.
    \item \textbf{Observation spaces $O_i$.} Sellers observe the common cost and each completed round's price pair, retaining past prices in their observation histories. The platform observes the current submitted pair, represented in unordered form to enforce symmetry; cost is hidden from it.
    \item \textbf{Transitions $\tran$.} Sellers submit prices simultaneously, moving the game to the display stage. The platform selects the displayed set; after payoffs and seller observations are generated, the game returns to the pricing stage.
    \item \textbf{Rewards $R_i$.} Price submission gives zero reward. At the display stage, sellers receive profits $\pi_{i,t}$ and the platform receives consumer surplus $U(p_t;B_t)$. A displayed seller pricing below cost incurs a loss.
\end{itemize}

\paragraph{Admission decisions.}
\label{subapp:collusion_interventions}
The leader commits to a symmetric rule specifying the displayed set $B(p)$ at every price pair. Its three actions select
\[
B_0(p)=\{i:p_i=p^{(0)}\},\qquad
B_1(p)=\operatorname*{arg\,min}_{i\in\{1,2\}}p_i,\qquad
B_2(p)=\{1,2\}.
\]
Above the minimum grid price, action 0 displays neither seller. Action 1 displays the cheaper seller, or both when prices tie; action 2 always displays both. Sellers at the minimum grid price are always included. Actions with the same effect remain available as aliases; no action mask is used. At $(p^{(0)},p^{(0)})$, all three actions display both sellers. These choices span exactly the display outcomes obtainable with symmetric price ceilings on the specified grid.

During validation, we compare all 15 fixed thresholds and the 120 clipped-minimum rules
\[
\tau_{l,u}(p)=\min\{p^{(u)},\max\{p^{(l)},\min(p_1,p_2)\}\},\qquad 0\leq l\leq u\leq14.
\]
The rule takes the lower submitted price as its ceiling, subject to a lower floor and an upper cap. It displays both sellers when both prices are at most the floor, rewards lower prices between the bounds, and excludes both sellers when both exceed the cap.
Validation selects the fixed threshold $p^{(7)}=1.55$ and the clipped rule $(l,u)=(4,7)$, whose lower bound is $p^{(4)}\simeq1.314286$.

Figure~\ref{fig:hidden_cost_learning} also reports no intervention, equivalent to having both sellers always displayed, and the policy selected by random search. All plotted benchmarks use the same two costs, response stopping criterion, learning parameters, Q initialization, and validation seller seeds as PPO.

\subsubsection{Follower Response and StackPOMDP Construction}
\label{subapp:stack_instantiation}

Each seller has a $225\times15$ Q-table, indexed by the previous price pair and its own current price. The response algorithm records the Q-tables, preceding prices, and exploration and stability counters. Every response restarts both Q-tables and the exploration clock and draws an initial preceding price pair uniformly from $\mathcal P^2$.

Following \citet{calvano2020artificial,johnson2023platform}, we initialize Q-values from discounted profits against a uniformly random rival. We retain the no-intervention benchmark of \citet{calvano2020artificial} across candidate display rules, giving each the same initial Q-values at a given cost:
\[
Q_{i,0}(s,p_i)=\frac{1}{(1-\delta)|\mathcal P|}\sum_{p_j\in\mathcal P}(p_i-c)\widetilde D_i(p_i,p_j).
\]
Here, $\widetilde D_i$ is full-market logit demand with both products displayed. Subsequent profits and Q-updates use the committed display rule.

\paragraph{Queries and updates.}
The response algorithm runs seller Q-learning internally. Whenever it requires the leader's action for an unordered price pair, it reuses the cached answer if available; otherwise, the StackPOMDP queries the leader policy, caches the answer, and resumes seller learning. Answers remain fixed throughout response computation and reward evaluation. Each episode requires at most 120 distinct queries.

With Q-learning states $s_t=p_{t-1}$ and $s_{t+1}=p_t$, each seller updates its chosen entry by
\begin{equation}
Q_{i,t+1}(s_t,p_{i,t})=(1-\eta)Q_{i,t}(s_t,p_{i,t})+
\eta\left[\pi_{i,t}+\delta\max_{p\in\mathcal P}Q_{i,t}(s_{t+1},p)\right],
\label{eq:collusion_q_update}
\end{equation}
where $\eta=0.15$ and $\delta=0.95$. Each seller explores uniformly with probability $\epsilon_t=\exp(-4\times10^{-6}t)$ and otherwise takes the first maximizer in grid order.

\paragraph{Stopping and returned response.}
Following \citet{calvano2020artificial,johnson2023platform}, each seller response stops when neither greedy policy has changed at any of the 225 states for 100,000 consecutive rounds, or after one billion rounds. Both greedy policies are then frozen and returned together with the terminal price state. Stability defines a stopping rule, not an equilibrium certificate. All reported fresh seller responses met the stability criterion.

\paragraph{StackPOMDP construction.}
The StackPOMDP augments the market state with the seller-learning state, action cache, and response/reward phase. Leader queries receive zero reward during the response phase. The same cached display rule is used in the reward phase, with any previously unseen price pair answered and cached on first encounter. Each episode ends with a terminal reward transition; PPO assigns actor credit to the original action draws.

The response phase runs seller Q-learning separately at costs 1.0 and 1.5, using the same cached commitment and seller random seed. With frozen deterministic policies on the finite price grid, the trajectory from each response's terminal price state eventually repeats. The reward phase computes exact long-run consumer surplus by averaging Equation~\eqref{eq:consumer_surplus} over this recurrent cycle, excluding the transient prefix, and then averages equally across the two costs. This evaluation is a variant of the fresh, finite-horizon reward phase in Definition~\ref{def:StackPOMDP}.

The actor receives only the two price coordinates. The critic additionally receives the sampled cost context, partially specified display rule, query progress, phase, and returned seller-price marginals.

\subsubsection{Implementation, Training, and Evaluation}
\label{subapp:collusion_hyperparams}

Rollouts contain 40 episodes. The learning rate is $3\times10^{-4}$ for the first quarter of training, then decays linearly to zero. Each of the 25 leader seeds (801--825) trains for 10,000 sampled commitments, totaling 20,000 seller responses.

Price coordinates are normalized to $[-1,1]$. Seller Q-updates and cache lookups run internally; only queries requiring a new leader action and the terminal reward transition contribute PPO transitions.

\paragraph{Evaluation.}
Deterministic evaluation takes the most probable display action at each query. During training, the current policy is evaluated every 400 seller responses on five fixed validation seller seeds per cost. Figure~\ref{fig:hidden_cost_learning} uses this validation cohort for the unsmoothed curves and benchmark lines. For the final comparison, Table~\ref{tab:hidden_cost_platform} evaluates each final PPO policy and selected benchmark on \PlatformSellerN{} fresh seller seeds per cost. PPO and the threshold benchmarks share a fresh cohort; random search uses a separate fresh cohort. PPO means, medians, and standard errors summarize run-level scores after averaging seller responses within each policy. For the equal-cost summary, the two costs are averaged within each run before computing these statistics.

\paragraph{Random search.}
The random-search protocol was specified after earlier PPO outcomes were available and fixed before this run. It samples 10,000 distinct complete display rules uniformly by choosing independently among distinct display outcomes at each unordered price pair and rejecting repeated rules. Each rule uses a new seller-learning seed, shared across the two costs, giving 20,000 seller responses in total. We select the rule with the highest equal-cost mean, retaining the earliest candidate in a tie. The selected rule is frozen before evaluation on 25 fresh seller seeds per cost. The reproducibility protocol records the random seeds and selection rule, and the archive retains all sampled rules and responses.

\paragraph{Illustration and scope.}
The illustrated policy is final seed 801, selected among all 25 runs by highest final validation score with the smallest seed breaking ties. Every fresh response at each cost reaches the corresponding marked price pair. We compute unilateral one-period profits directly from Equation~\eqref{eq:logit_demand_app}, holding the other price fixed. Enumerating all 225 price pairs and all unilateral deviations verifies the stated pure-strategy equilibria of the clipped rule. These calculations do not certify repeated-game optimality or convergence from every initialization.

\subsection{Atari 2600 Bilateral Trade}
\label{app:atari_pomg}
\raggedbottom

\subsubsection{Environment and POMG Specification}

The bilateral-trade POMG has the following components:
\begin{itemize}
    \item \textbf{Agents.} One seller $S$ and one buyer $B$, each controlling an independently resetting single-player \textit{Space Invaders} instance. Both begin without ammunition.
    \item \textbf{State space $S$.} The state records both emulator states, frame stacks, and reset histories; ammunition inventories and active projectiles; the game clock, five-event schedule, and current event; and accumulated game rewards, transfers, and payments. Each inventory holds between zero and five bullets.
    \item \textbf{Action spaces $A_i$.} Each agent outputs $(a_i^{\mathrm{game}},a_i^{\mathrm{econ}})\in\{0,\ldots,5\}\times[0,1]$. The first coordinate selects \texttt{NOOP}, \texttt{FIRE}, \texttt{RIGHT}, \texttt{LEFT}, \texttt{RIGHTFIRE}, or \texttt{LEFTFIRE}; the second gives the seller's price $p_j$ or buyer's threshold $t_j$. Only the coordinate appropriate to the current transition is executed.
    \item \textbf{Observation spaces $O_i$.} Agent $i$ observes $o_{i,t}=(x_{i,t},u_{i,t})$. Here $x_{i,t}\in[0,1]^{4\times84\times84}$ is its normalized four-frame Atari stack, and
    \[
        u_{i,t}=(b_{i,t},z_{i,t},\tau_t,d_t,e_t)\in[0,1]^9
    \]
    contains normalized ammunition, a projectile-active indicator, normalized game time, a trade-pause indicator, and a five-entry one-hot encoding of the next or current event. At a trade, $x_{i,t}=\mathbf 0$; during gameplay, $d_t=0$. The opponent's emulator state and future trade times are hidden.
    \item \textbf{Transitions $\tran$.} Gameplay advances both emulators independently; a life loss or game-over affects only that emulator and preserves the outer ledger. Five distinct trade indices are sampled uniformly without replacement from $\{0,\ldots,199\}$. At each trade, both emulators pause, one bullet arrives to the seller, and the bullet and payment transfer exactly when $p_j\leq t_j$.
    \item \textbf{Rewards $R_i$.} Gameplay transition $t$ gives rewards $(0.1\operatorname{sign}(r^{\mathrm{ALE}}_{S,t}),\operatorname{sign}(r^{\mathrm{ALE}}_{B,t}))$, where $r^{\mathrm{ALE}}_{i,t}$ is the ALE score reward. Trade $j$ gives $(p_jI_j,-p_jI_j)$, where $I_j=\mathbf 1[p_j\leq t_j]$. With total clipped gameplay reward $G_i=\sum_{t=1}^{200}\operatorname{sign}(r^{\mathrm{ALE}}_{i,t})$, episode returns follow Equation~\eqref{eq:atari_payoffs}.
    \item \textbf{Reward-phase length $T_{\max}$.} There are 205 transitions: 200 gameplay decisions and five paused trades.
\end{itemize}

\subsubsection{Follower Response and StackPOMDP Construction}
The Meta-RL oracle of Proposition~\ref{prop:meta_rl_pi} records the next event to query and the economic actions already returned by the leader. It starts at the first event with no recorded actions.

\paragraph{Queries and updates.}
The actor interface augments the basic POMG fields with the contextual commitment $\omega$,
\begin{equation}
    \xi_t=\bigl(u_t^\top,\omega^\top\bigr)^\top
    =\bigl(b_t,z_t,\tau_t,d_t,e_t^\top,\omega^\top\bigr)^\top
    \in[0,1]^{14}.
    \label{eq:atari_actor_state}
\end{equation}
The context entries are zero in the basic POMG. For each $j\in\{1,\ldots,5\}$, the oracle queries the leader using the complete canonical input
\begin{equation}
    x_j^Q=\mathbf 0,\qquad
    u_j^Q=(0,0,0,1,e_j),\qquad
    \omega^Q=\mathbf 0,\qquad
    \xi_j^Q=((u_j^Q)^\top,(\omega^Q)^\top)^\top.
    \label{eq:atari_canonical_query}
\end{equation}
Thus, $x_j^Q$ is the dummy image and only the event-identity component of the low-dimensional input varies. Each returned economic action is cached for its corresponding trade in the reward game.

\paragraph{Stopping and returned response.}
After the fifth query, the returned economic actions form the context
\[
    \omega=(a^{L,\mathrm{econ}}_1,\ldots,a^{L,\mathrm{econ}}_5)\in[0,1]^5,
\]
and the oracle returns the frozen meta-follower conditioned on this vector.

\paragraph{StackPOMDP construction.}
During the response phase, the StackPOMDP state records query progress and cached economic actions; each query receives zero reward, and neither emulator advances. Once the response is fixed, both emulators reset and the reward phase follows the bilateral-trade POMG. The leader receives live $(x_t,u_t)$ observations during gameplay and the same canonical query input at trade $j$, where its cached economic action is implemented. The actor receives no response/reward phase flag or accumulated response history; these fields are available to the centralized critic together with the remaining training state. The five-query response phase followed by the 205-transition reward game yields a 210-transition StackPOMDP episode.

Conditioning on the sampled economic vector $\omega$ fixes the five economic actions while retaining stochastic gameplay. The conditional expected return still depends on the actor parameters through the gameplay policy, and the training objective averages this return over the sampled vectors.

\subsubsection{Implementation, Training, and Evaluation}
The environment uses at most 30 reset noops, episodic life, max-and-skip 4, $84\times84$ warped grayscale frames, a four-frame stack, clipped game reward $\operatorname{sign}(r)$, and zero ALE repeat-action probability. An ordinary life loss rebuilds that side's frame stack and continues the same ALE game after one \texttt{NOOP}; a true game-over starts a fresh game for that side. Both cases preserve the outer ammunition ledger, common clock, and five-event schedule, and never reset the other emulator.

An emulator-RAM wrapper detects active projectiles and effective launches. The external ammunition ledger credits arrivals and accepted transfers, debits only RAM-confirmed launches, and caps either inventory at five. Because the original game ignores \texttt{FIRE} while a player projectile is already in flight, \texttt{FIRE}-containing actions are masked when ammunition is zero or a projectile is active. This prevents a null input from being charged against scarce ammunition. Since both agents start empty and one bullet arrives to the seller at each event, the evaluator checks
\begin{equation}
5 = n_S^{\mathrm{shot}} + n^{\mathrm{trade}} + b_S^{\mathrm{final}},
\qquad
n^{\mathrm{trade}} = n_B^{\mathrm{shot}} + b_B^{\mathrm{final}}.
\label{eq:atari_bullet_accounting}
\end{equation}
where $b_i^{\mathrm{final}}$ denotes terminal inventory.

\paragraph{Shared composite actor.}
Panel~(a) of Figure~\ref{fig:nn_architecture} shows the generic composite actor that serves as the starting point for all four policies. Its visual path is the standard DQN-style Atari encoder implemented by Stable-Baselines3 as \texttt{NatureCNN}. It comprises three convolutional ReLU layers,
\[
\texttt{Conv2d}(4,32,8,4),\qquad
\texttt{Conv2d}(32,64,4,2),\qquad
\texttt{Conv2d}(64,64,3,1),
\]
followed by flattening and a 512-unit ReLU layer. In parallel, a 64-unit ReLU layer encodes $\xi_t$. Concatenating both paths gives 576 features to a linear six-logit Atari head, allowing gameplay to depend on ammunition, time, event identity, and---for the meta-buyer---the seller-price context. Fresh frames represent an independent emulator reset; no reset-status flag is supplied. The generic economic head applies two 64-unit tanh layers to the encoded state and maps them to $\alpha=\operatorname{softplus}(q_\alpha)+10^{-4}$ and $\beta=\operatorname{softplus}(q_\beta)+10^{-4}$. Training samples the economic action from this Beta distribution; evaluation uses its mean $\alpha/(\alpha+\beta)$.

\paragraph{Specialized meta-seller price head.}
The meta-buyer and both leaders use the generic architecture in Panel~(a). The meta-seller poses a harder conditioning problem: it must respond to a complete five-threshold commitment, and each quote affects both whether trade occurs and, if accepted, the payment. In preliminary runs, applying the generic economic head in this setting did not yield prices that reliably varied with the buyer's threshold vector. We therefore freeze the pretrained deterministic gameplay actor and replace only its economic head with the specialized architecture in Panel~(b). A live-state branch maps $u_t=(b_t,z_t,\tau_t,d_t,e_t)$ to a base price logit and concentration. A second branch maps the normalized commitment $\widetilde\omega=2\omega-\mathbf{1}$ and event identity to a residual, while the event-specific value $e_t^\top\widetilde\omega$ also shifts the Beta mean. The context output and event slope start at zero, so initial prices are invariant to $\omega$. The specialized head is trained with a new critic and is not transferred to either leader; both leaders instead initialize generic event-only Beta heads.

\begin{figure}[tb]
\centering
\begin{tikzpicture}[
    x=1cm,
    y=1cm,
    font=\scriptsize,
    flow/.style={-{Stealth[length=2.1mm]}, line width=0.55pt,
        shorten <=0.7pt, shorten >=1.2pt},
    module/.style={draw, rounded corners=2pt, fill=white,
        align=center, minimum height=0.72cm, inner sep=3pt}
]
\node[anchor=west, font=\bfseries\small] at (0,8.15)
    {(a) Shared composite actor};

\node[module, text width=1.55cm] (frames) at (0.95,7.15)
    {$4\!\times\!84\!\times\!84$\\frames $x_t$};
\node[module, text width=2.90cm] (state) at (1.50,5.55)
    {$\xi_t\in[0,1]^{14}$\\
     $(b_t,z_t,\tau_t,d_t,e_t,\omega)$};

\node[module, text width=1.65cm] (cnn) at (3.25,7.15)
    {SB3\\NatureCNN\\$512$};
\node[module, text width=1.75cm] (stateenc) at (4.15,5.55)
    {state encoder\\FC $64$, ReLU};
\node[module, text width=1.55cm] (gameconcat) at (6.35,6.45)
    {concatenate\\$576$};
\node[module, text width=1.60cm] (gamehead) at (8.55,6.45)
    {game head\\$6$ logits};
\node[module, text width=2.05cm] (gamemask) at (11.00,6.45)
    {feasibility mask\\Atari action\\$a_t^{\mathrm{game}}$};

\draw[flow] (frames) -- (cnn);
\draw[flow] (state) -- (stateenc);
\draw[flow, rounded corners=2pt]
    (cnn.east) -- ++(0.55,0)
    |- ($(gameconcat.west)+(0,0.16)$);
\draw[flow, rounded corners=2pt]
    ($(stateenc.east)+(0,0.14)$) -- ++(0.18,0)
    |- ($(gameconcat.west)+(0,-0.16)$);
\draw[flow] (gameconcat) -- (gamehead);
\draw[flow] (gamehead) -- (gamemask);

\node[module, text width=2.20cm] (econhead) at (7.00,4.55)
    {generic\\economic head\\$64$ tanh\\$\to64$ tanh\\$\to(q_\alpha,q_\beta)$};
\node[module, text width=1.95cm] (beta) at (9.85,4.55)
    {$\alpha,\beta=$\\$\operatorname{softplus}(q)+\varepsilon$\\Beta action\\$a_t^{\mathrm{econ}}$};
\draw[flow, rounded corners=2pt]
    ($(stateenc.east)+(0,-0.14)$) -- ++(0.42,0)
    |- (econhead.west);
\draw[flow] (econhead) -- (beta);

\node[module, text width=2.15cm] (fullaction) at (14.20,5.48)
    {full action\\$(a_t^{\mathrm{game}},a_t^{\mathrm{econ}})$};
\draw[flow, rounded corners=2pt]
    (gamemask.east) -- (12.70,6.45)
    -- (12.70,5.64)
    -- ($(fullaction.west)+(0,0.16)$);
\draw[flow, rounded corners=2pt]
    (beta.east) -- (12.70,4.55)
    -- (12.70,5.32)
    -- ($(fullaction.west)+(0,-0.16)$);

\draw[gray!55, line width=0.45pt] (0,3.55) -- (15.35,3.55);
\node[anchor=west, font=\bfseries\small] at (0,3.15)
    {(b) Meta-seller economic-head replacement};

\node[module, text width=2.35cm] (liveinput) at (1.25,2.00)
    {live state\\$u_t=(b_t,z_t,\tau_t,d_t,e_t)$};
\node[module, text width=2.35cm] (contextinput) at (1.25,0.65)
    {buyer commitment\\$\widetilde\omega=2\omega-\mathbf 1$, $e_t$};
\node[module, text width=2.00cm] (livebranch) at (4.35,2.00)
    {live branch\\$32$ tanh\\$\to(\ell_t,\rho_t)$};
\node[module, text width=2.00cm] (contextbranch) at (4.35,0.65)
    {context branch\\$32$ tanh\\$\to r_t$};
\node[module, text width=3.30cm] (combine) at (8.10,1.35)
    {$q_t=e_t^\top\widetilde\omega$\\
     $\mu_t=\varepsilon+(1-2\varepsilon)$\\
     $\phantom{\mu_t={}}\cdot\sigma(\ell_t+r_t+s q_t)$\\
     $\kappa_t=\operatorname{softplus}(\rho_t)+\varepsilon$};
\node[module, text width=2.25cm] (sellerbeta) at (11.50,1.35)
    {$\alpha_t=\mu_t\kappa_t$\\
     $\beta_t=(1-\mu_t)\kappa_t$};
\node[module, text width=2.30cm] (price) at (14.25,1.35)
    {Beta-distributed\\quoted price $p_t$};

\draw[flow] (liveinput) -- (livebranch);
\draw[flow] (contextinput) -- (contextbranch);
\draw[flow, rounded corners=2pt]
    (livebranch.east) -- ++(0.42,0) |- ($(combine.west)+(0,0.22)$);
\draw[flow, rounded corners=2pt]
    (contextbranch.east) -- ++(0.42,0) |- ($(combine.west)+(0,-0.22)$);
\draw[flow] (combine) -- (sellerbeta);
\draw[flow] (sellerbeta) -- (price);
\end{tikzpicture}
\caption{Atari policy architecture. \textup{(a)} The shared composite actor uses image and low-dimensional state features for gameplay and a generic Beta head for the economic action. The meta-buyer supplies the complete contextual state to both branches; leaders restrict the economic branch to event identity. \textup{(b)} The meta-seller alone replaces the generic economic head with live-state and commitment branches that produce a Beta-distributed price.}
\Description{Two architecture panels. The first shows a Nature CNN and state encoder feeding the shared gameplay and generic economic heads. The second shows the meta-seller's specialized price head, with separate live-state and commitment branches combined into a Beta distribution.}
\label{fig:nn_architecture}
\end{figure}

\paragraph{Training stages.}
Training proceeds through four stages:
\begin{enumerate}
    \item \textbf{Fixed-ammunition gameplay:} the agent learns the 200-decision game with five bullets supplied at reset and no trades.
    \item \textbf{Delayed-ammunition gameplay:} the agent starts empty and receives one free bullet at each of five random paused events, introducing the final event and reset dynamics without economic decisions.
    \item \textbf{Contextual meta-followers:} both train over randomized event schedules and randomized five-entry opponent commitments. The meta-buyer conditions on seller-price vectors and fine-tunes its gameplay path; the meta-seller conditions on buyer-threshold vectors and keeps its deterministic gameplay actor frozen.
    \item \textbf{StackPOMDP leaders:} each leader trains against the frozen opposite-role meta-follower through complete 210-transition outer episodes.
\end{enumerate}

\paragraph{Critic and PPO training.}
Each of the four training stages initializes a private critic with architecture $608\to256$ ReLU, $256\to256$ ReLU, $256\to1$. It concatenates the 576 actor features with a zero-padded 32-scalar vector containing the stage-appropriate phase, schedule, inventory, commitment, query-cache, reward, transfer, payment, and transition-count fields. These privileged fields supply additional state information for value estimation but never enter the actor. The policy code in the reproducibility package defines their exact ordering and verifies their dimensions.

All stages use PPO with four vectorized environments, four optimization epochs, clip range $0.1$, and entropy coefficient $0.01$. Table~\ref{tab:atari_training_configuration} gives the stage-specific learning rates and module transfers. Each rollout collects one episode per environment; all $4n_{\mathrm{steps}}$ policy-action rows form one minibatch in each optimization epoch. A leader episode has 210 physical transitions but only 205 action draws: five response queries and 200 gameplay actions. These draws contribute policy-gradient score terms; the five real trades reuse the cached economic actions and generate no additional draws or score terms. Each final leader run therefore collects 1,952,420 action rows over 2,000,040 physical transitions and applies 9,524 optimizer epochs. For policies with both gameplay and economic decisions, the actor loss averages their contributions separately before summing, so the 200 gameplay decisions do not overwhelm the five economic decisions.

\begin{table}[tb]
\centering
\caption{Atari training configuration. Every stage trains for approximately two million environment transitions. Here $n_{\mathrm{steps}}$ counts policy-action rows per complete episode, excluding cached actions.}
\label{tab:atari_training_configuration}
\small
\begin{tabular}{>{\raggedright\arraybackslash}p{3.25cm}rp{10.15cm}}
\toprule
\textbf{Training stage} & $n_{\mathrm{steps}}$ & \textbf{Learning rates, initialization, and trainable modules} \\
\midrule
Fixed-ammunition gameplay & 200 & Base rate $2.5\times10^{-4}$; new visual and state encoders, Atari head, and private critic; economic action inactive. \\
Delayed-ammunition gameplay & 205 & Base rate $2.5\times10^{-4}$; transferred visual encoder and Atari head use $2.5\times10^{-5}$, while the state encoder and new critic use the base rate; economic action inactive. \\
Contextual meta-buyer & 205 & Base rate $1\times10^{-4}$; transferred visual encoder and Atari head use $1\times10^{-5}$, while the state encoder, generic Beta head, and new critic use the base rate. The Beta head is initialized with mean $0.95$ and concentration $10$. \\
Contextual meta-seller & 205 & Gameplay actor frozen; live branch uses $5\times10^{-4}$, context branch and event slope use $2\times10^{-3}$, and the new critic uses $1\times10^{-4}$; the three groups are clipped separately at $0.5$. The Beta head is initialized with mean $0.5$ and concentration $2$. \\
Each StackPOMDP leader & 205 & Base rate $1\times10^{-4}$; transferred visual encoder and Atari head use $1\times10^{-5}$, while the state encoder, new event-only Beta head, and new critic use the base rate. The Beta head is initialized with mean $0.5$ and concentration $2$. \\
\bottomrule
\end{tabular}
\end{table}

\paragraph{Reproducibility and scope.}

The experiments require an externally supplied Space Invaders ROM, which is not redistributed. Setup instructions and the required SHA-256 digest are provided in the replication repository.

The meta-follower models are shared upstream checkpoints trained with seed one. Starting from these shared models, the final leader cohort uses 25 independent training seeds, $1,\ldots,25$, for each leader role. We reuse the shared gameplay and meta-follower checkpoints without repeating their pretraining. Response-curve uncertainty is measured across emulator episodes for one trained meta-follower. Leader uncertainty instead treats the independently trained policy as the unit and does not pool emulator episodes.

All six scheduled evaluations use the same 20 emulator seeds, $4{,}000{,}001,\ldots,4{,}000{,}020$, across policies and roles. The physical-transition checkpoints are 0, 400,680, 800,520, 1,200,360, 1,600,200, and 2,000,040. Every final policy is also evaluated on 100 fresh emulator seeds, $5{,}000{,}001,\ldots,5{,}000{,}100$. No checkpoint is selected using either evaluation cohort.

The final 50 leader runs accumulated 299.68 wall-clock hours summed across runs, 1,498.42 allocated core-hours, and 458.13 SLURM \texttt{TotalCPU} hours. These totals include training and the scheduled and fresh final evaluations.

The experiments cover one Atari trading environment, and the learned buyer response is not demonstrably timing-optimal. The replication pipeline supports further tests of resource dynamics, information, and market design.


\raggedbottom
\section{Controlled Ablations}
\label{app:results_extra}
\label{sec:ap:limitations}

This appendix collects three controlled ablations of the framework's design choices: response-phase credit assignment, commitment consistency across phases, and reward evaluation after follower adaptation. Appendix~\ref{app:statistical_comparisons} reports tests of their final outcomes. The maintained implementations and run-level outputs are publicly available in the reproducibility package~\citep{GlcBrero_StackelbergPOMDP}.

\subsection{Response-Phase Credit Assignment}
\label{sec:response_credit_ablation}

To isolate response-phase credit assignment, we use the one-query game from Appendix~\ref{app:response_credit}. The leader uses the same Bernoulli policy with probability $p=\sigma(z)$ in both phases. A response query draws $X\sim\mathrm{Bernoulli}(p)$, and the follower then commits to action $b=X$. Gameplay draws an independent leader action $A\sim\mathrm{Bernoulli}(p)$ and yields reward $R=2X-A$. The leader's exact value is therefore $J(p)=p$: increasing $p$ improves the objective. The follower responds optimally to its single observed action, not to the full stochastic policy.

Both variants execute this same response procedure and gameplay. The full Stackelberg POMDP update credits both sampled actions, whereas the \emph{Basic POMDP} update treats the follower response as part of the environment and credits only gameplay:
\begin{equation}
\widehat g_{\mathrm{full}}=R\bigl[(X-p)+(A-p)\bigr],
\qquad
\widehat g_{\mathrm{basic}}=R(A-p).
\label{eq:response_ablation_estimators}
\end{equation}
Their expectations are $p(1-p)$ and $-p(1-p)$, respectively. Thus, at every interior policy, omitting query credit reverses the expected update direction. Here actions are sampled freshly, so both terms in the full update correspond to actual random draws.

\paragraph{Learning protocol.}
We initialize the single scalar logit at $z=0$. Each update holds $z$ fixed while drawing 64 complete episodes, then performs plain stochastic gradient ascent with learning rate $0.1$ using the mean of the applicable estimator in Equation~\eqref{eq:response_ablation_estimators}. Both arms divide by 64, even though the full estimator contains two score terms per episode. There is no critic, baseline, entropy penalty, advantage normalization, gradient clipping, or learning-rate schedule. Every query has zero immediate reward; each score receives the final gameplay return.

We fix seeds 1--25 and $1{,}000$ updates before training, giving $64{,}000$ outer episodes and $128{,}000$ actual transitions per method and seed. Within a seed, the two methods share query and gameplay uniform variates, thresholded at their respective current probabilities; query and gameplay variates are independent. Different seeds use independent streams. Evaluation computes $J(p)=p$ exactly after every update, without replacing the stochastic policy by its greedy action. The primary comparison uses the fixed final checkpoint; all seeds are included. The protocol, per-update counts of the four possible $(X,A)$ outcomes, and source hashes are included in the reproducibility bundle. We verify the gradient by enumerating those outcomes and by finite differences, and reconstruct every recorded update from its event counts.

\begin{figure}[t]
\centering
\includegraphics[width=\linewidth]{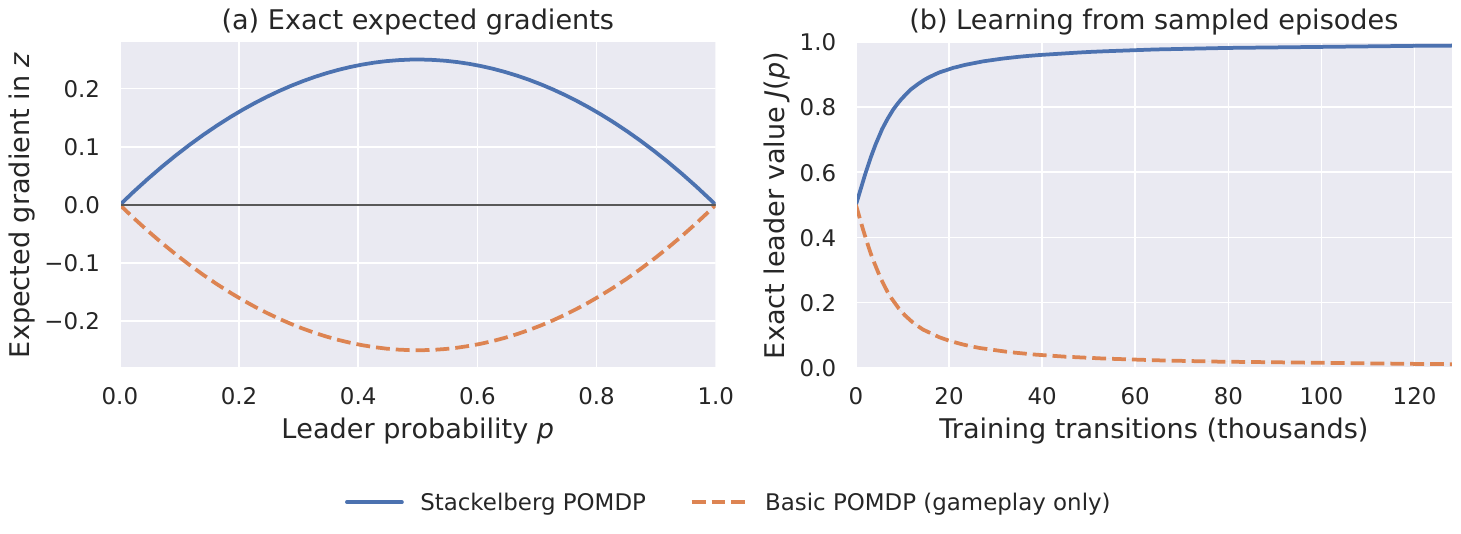}
\caption{Controlled response-credit ablation. \textup{(a)} Exact expected gradients with respect to the leader logit: full credit gives $p(1-p)$ and gameplay-only credit gives $-p(1-p)$. \textup{(b)} Sampled REINFORCE training from $p=1/2$, using 25 matched seeds, 64 complete episodes per update, learning rate $0.1$, and $1{,}000$ updates. Both methods execute $128{,}000$ actual transitions, including response queries. Curves show the mean exact value $J(p)=p$; shading is one standard error across seeds and is narrow at this scale.}
\Description{The left panel shows opposite expected gradient signs. In the right panel, full response credit raises leader value from 0.5 to about 0.989, while gameplay-only credit reduces it to about 0.011.}
\label{fig:response_credit_ablation}
\end{figure}

Figure~\ref{fig:response_credit_ablation}(a) shows the exact expected gradients, while Figure~\ref{fig:response_credit_ablation}(b) shows learning from sampled episodes. Both methods start from the same untrained scalar policy, use the same optimizer and episode budget, and normalize gradients by the same number of complete episodes. At the fixed training horizon, mean exact value is $0.9891$ with full credit and $0.0109$ with gameplay-only credit; the paired mean difference is $0.9782$ with standard error $0.00014$.

This diagnostic illustrates the bias identified in Proposition~\ref{prop:response_credit}. It isolates response-query credit without a critic, action caching, or a pretrained follower. It does not establish that retaining query transitions is necessary for every optimization method.

\subsection{One Policy Across Phases: Commitment Consistency}
\label{sec:ap:limits:noninvariant}

This diagnostic tests whether phase information lets the leader induce a response to a policy it does not actually play. Theorem~\ref{thm:main} requires the same action distribution for the same game history during response queries and gameplay.

We use iterated Prisoner's Dilemma with $T=5$ rounds. In each round, the leader (row player) and follower (column player) choose whether to cooperate ($C$) or defect ($D$). Entries below list their respective stage payoffs:
\[
\begin{array}{c|cc}
 & C & D \\ \hline
C & (-1,-1) & (-3,0) \\
D & (0,-3) & (-2,-2)
\end{array}.
\]
Mutual cooperation benefits both players relative to mutual defection, but each has a stage-game incentive to defect. In the plotted runs, each player's game observation belongs to $\{\texttt{Start},C,D\}$, where $C$ and $D$ record the opponent's previous action; the game policies use this one-round memory.

The follower uses a contextual policy pre-trained by policy-gradient RL across leader commitments. Before gameplay, it receives the leader's answers to three queries, one for each game observation. For each seed, the follower trains over all eight deterministic commitment tables, with a uniformly sampled table per episode. REINFORCE uses a linear categorical policy, Adam with learning rate $0.02$, and batches of 13 complete episodes. We train for at least 500 updates, then check every 100 updates and stop when the maximum return gap relative to the best stationary follower policy is at most $0.25$, with a cap of 10,000 updates. Exhaustive evaluation over the eight tables verifies this gate for every frozen follower. This is a diagnostic response guarantee within the stated one-round-memory policy class. We compare a leader receiving only these game observations with a phase-aware variant that also receives a binary indicator distinguishing response queries from gameplay.

The leader uses policy-gradient learning with a linear categorical policy, Adam, $\gamma=1$, and batches of 13 complete outer episodes. Each of the 25 seeds (0--24) shares its frozen follower across the six conditions: learning rates $0.008$ and $0.015$ without phase information, and $0.004$, $0.008$, $0.015$, and $0.03$ with it. Training lasts 1,924 updates, or 200,096 physical transitions. Scheduled evaluation uses 20 held-out episodes every 2,000 transitions; final evaluation uses 100 fresh episodes. Performance is the average leader payoff per gameplay round.

As Figure~\ref{fig:memory} shows, the phase-aware leader obtains higher reward, but not through a valid Stackelberg commitment. Every final phase-aware policy receives $0$ per round, while every phase-hidden policy receives $-0.8$. Its response-query answers induce follower cooperation, while the leader defects during the reward game. The experiment therefore illustrates why phase information must remain unavailable to the deployed leader policy: otherwise the response is induced by one policy and evaluated against another.

\begin{figure}[t]
\centering
\includegraphics[width=0.75\linewidth]{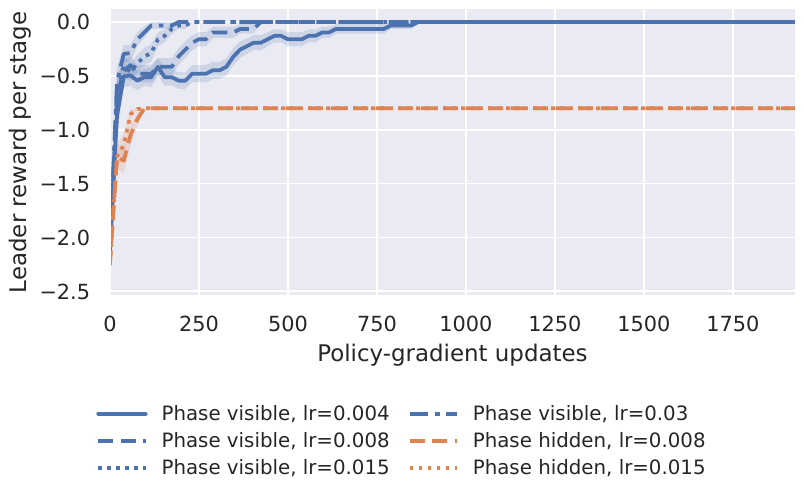}
\caption{Commitment consistency across phases. In iterated Prisoner's Dilemma, a phase-aware leader obtains higher reward by inducing cooperation during response queries and defecting during reward play, but this is not one Stackelberg commitment. Each learning-rate curve averages 25 independent seed-level policy means, using 20 held-out episodes per scheduled evaluation. Shaded regions show one sample standard error across runs.}
\Description{Learning curves showing invalid gains from using phase-dependent commitments.}
\label{fig:memory}
\end{figure}

\subsection{Reward After the Response: Objective Alignment}
\label{sec:ap:limits:reward}

This diagnostic concerns the commitment objective in Definition~\ref{def:leaderProblem}, which evaluates a committed leader policy against its follower response. Adding leader reward from the response phase changes that objective by also valuing play against partially adapted followers.

We use a single-shot Battle of the Sexes variant. The leader chooses a row and the follower a column, each selecting $A$ or $B$. Entries give the leader and follower payoffs, respectively:
\[
\begin{array}{c|cc}
 & A & B \\ \hline
A & (2,0.001) & (-c,0) \\
B & (-c,0) & (1,2)
\end{array},
\qquad c\in\{0,5\}.
\]
Both players prefer coordination, but the leader prefers $(A,A)$ and the follower prefers $(B,B)$. The small follower payoff $0.001$ at $(A,A)$ is chosen to slow adaptation toward the leader-preferred outcome. The two variants use either no coordination penalty ($c=0$) or a leader penalty of $5$ for mismatched actions ($c=5$); the follower receives zero in either case.

The follower uses ten tabular Q-learning episodes to respond before the final reward game. We compare excluding and including the leader's rewards from those learning episodes; the plotted performance always measures reward in the final game alone. The leader uses a linear Q-network with replay, episode-level parameter-noise exploration, and learning rate $0.1$, with $\gamma=1$; the follower uses tabular Q-learning with learning rate $0.2$ and exploration probability $0.1$. Each condition uses 25 matched seeds (0--24) and 22,000 physical transitions. Scheduled evaluations use 20 held-out episodes every 2,000 transitions; final evaluation uses 100 fresh episodes per policy. The leader replay buffer holds 50,000 transitions, minibatches contain 1,024 transitions, and the target network is updated every 500 transitions.

Figure~\ref{fig:bots_leaderreward} shows that excluding response-phase reward yields the leader-preferred payoff $2$ in all 25 final policies in both games. Including it yields that payoff in only 5 of 25 policies without a coordination penalty, for a mean of $1.20\pm0.08$; the remaining policies obtain $1$. With the penalty, all 25 included-reward policies obtain $1$, the other coordination outcome. Thus, response-phase rewards reduce the reliability of optimizing the stated final-game objective.

This failure mode is distinct from a modeling error: optimizing reward throughout follower learning is a valid objective when the leader cares about transient dynamics, as in \citet{balaguer2022good}'s mechanisms for human participants. However, this objective need not yield a Stackelberg-optimal commitment under the specified follower-response model.

\begin{figure}[H]
\centering
\includegraphics[width=0.9\linewidth]{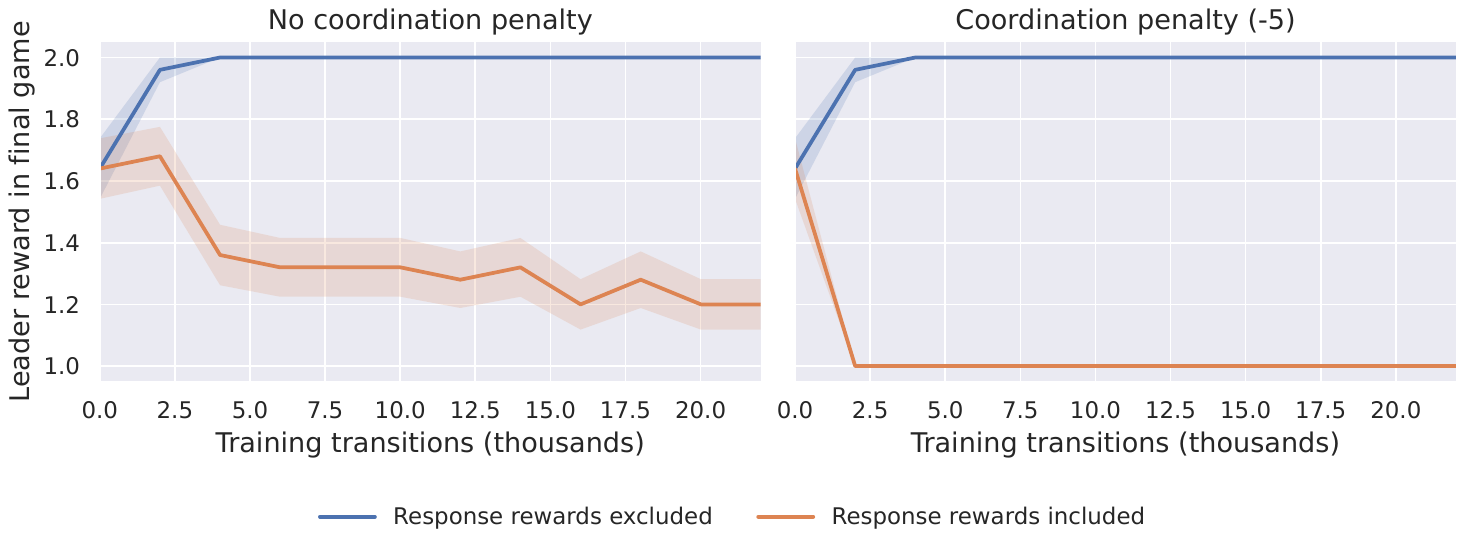}
\caption{Leader reward with and without reward during follower Q-learning in Battle of the Sexes. Plots show reward during actual play only (the relevant quantity for Stackelberg equilibria). Coordination failures give the leader payoff 0 in the left panel and $-5$ in the right panel; the follower payoff is 0 in both. Reward exclusion yields final-game payoff 2 in all 25 final policies in each panel. With response rewards included, 5 of 25 attain it without a penalty and none attain it with the penalty. Each curve averages 25 independent seed-level policy means evaluated on 20 held-out episodes; shaded regions show one standard error.}
\Description{Two Battle-of-the-Sexes learning curves showing that rewarding the leader during follower adaptation prevents reliable convergence to the Stackelberg payoff.}
\label{fig:bots_leaderreward}
\end{figure}

\FloatBarrier
\section{Statistical Comparisons}
\label{app:statistical_comparisons}

Table~\ref{tab:statistical_comparisons} reports retrospective tests specified after observing the descriptive results. All specified comparisons are shown. Tests are two-sided, with Holm correction separately within each of the four displayed families. Effects are mean differences in the original outcome units, with standard errors for those differences.

\paragraph{MSPM and ablations.}
We use exact sign tests across independent training seeds, pairing ablation conditions within each seed. The null is equal probability of positive and negative differences, conditional on a nonzero difference. These tests accommodate discrete and constant observed outcomes. Numerical ties within $10^{-10}$ are omitted from the sign count but retained in the mean and standard error. MSPM compares each run's best exactly evaluated welfare within the reported budget, including the separate final evaluation, with optimal SPM welfare. Exact evaluation removes checkpoint welfare sampling noise; the test concerns variation across training runs. Response credit uses the final update; matrix ablations use the final fresh 100-episode policy means. Commitment comparisons use only the two learning rates shared by both conditions.

\paragraph{Platform.}
Final PPO evaluation crosses 25 training seeds with 25 seller seeds. We use the pigeonhole bootstrap, independently resampling training-seed rows and seller-seed columns~\citep{owen2007pigeonhole}. Each cell averages the two costs. The same resampled seller columns are used for PPO and each threshold comparator; the selected random policy's separate seller-seed cohort is resampled independently. With $B=10^6$ draws and fixed bootstrap seed 20261005, we estimate a two-sided $p$-value as $(1+\#\{|\Delta_b^*-\widehat\Delta|\geq|\widehat\Delta|\})/(B+1)$, where $\widehat\Delta$ is the observed mean difference and $\Delta_b^*$ its bootstrap replicate. Bootstrap standard errors account for both sources of variation; the descriptive standard errors in Table~\ref{tab:hidden_cost_platform} summarize variation across policy means on the fixed seller cohort. Inference is conditional on the selected comparator policies. The single random-search run does not estimate variability across repeated searches.

\paragraph{Atari.}
Welch unequal-variance two-sample $t$-tests compare the 25 independent training-seed policy means per role at the final scheduled evaluation, averaging 20 emulator episodes per policy. Inference is conditional on the shared upstream checkpoints and this evaluation cohort.

\begin{table}[htbp]
\centering
\caption{Retrospective statistical comparisons. Differences follow the order in the first column. S: exact sign test; B: crossed bootstrap; W: Welch test. Adjusted $p$-values use Holm correction within each displayed family. There are 25 training seeds per condition; the platform also resamples seller seeds. Zero SE records identical observed outcomes across seeds. The random-policy comparison reaches the bootstrap's Monte Carlo resolution before adjustment.}
\label{tab:statistical_comparisons}
\small
\setlength{\tabcolsep}{4pt}
\renewcommand{\arraystretch}{1.1}
\begin{tabular}{@{}lccc@{}}
\toprule
Comparison & Mean difference (SE) & Test & Adjusted $p$ \\
\midrule
\multicolumn{4}{@{}l}{\textit{MSPM}} \\
MSPM minus SPM, 2 types & $0.0600\;(0.0000)$ & S & $3.0\times10^{-7}$ \\
MSPM minus SPM, 3 types & $0.0556\;(0.0000)$ & S & $3.0\times10^{-7}$ \\
MSPM minus SPM, 4 types & $0.0392\;(0.0014)$ & S & $3.0\times10^{-7}$ \\
MSPM minus SPM, 5 types & $0.0396\;(0.0014)$ & S & $3.0\times10^{-7}$ \\
MSPM minus SPM, 6 types & $0.0338\;(0.0003)$ & S & $3.0\times10^{-7}$ \\
\addlinespace
\multicolumn{4}{@{}l}{\textit{Platform}} \\
PPO minus fixed cap & $0.0540\;(0.0151)$ & B & $0.0010$ \\
PPO minus clipped rule & $-0.0384\;(0.0172)$ & B & $0.0251$ \\
PPO minus selected random policy & $0.2335\;(0.0261)$ & B & $3.0\times10^{-6}$ \\
\addlinespace
\multicolumn{4}{@{}l}{\textit{Atari: seller-led minus buyer-led}} \\
Accepted price & $0.4216\;(0.0026)$ & W & $3.0\times10^{-66}$ \\
Seller payoff & $2.1082\;(0.0132)$ & W & $3.0\times10^{-66}$ \\
Buyer payoff & $-1.4062\;(0.0900)$ & W & $2.9\times10^{-14}$ \\
\addlinespace
\multicolumn{4}{@{}l}{\textit{Ablations}} \\
Full minus gameplay-only credit & $0.9782\;(0.0001)$ & S & $3.0\times10^{-7}$ \\
Phase visible minus hidden, lr=0.008 & $0.8000\;(0.0000)$ & S & $3.0\times10^{-7}$ \\
Phase visible minus hidden, lr=0.015 & $0.8000\;(0.0000)$ & S & $3.0\times10^{-7}$ \\
Reward exclusion minus inclusion, c=0 & $0.8000\;(0.0816)$ & S & $1.9\times10^{-6}$ \\
Reward exclusion minus inclusion, c=5 & $1.0000\;(0.0000)$ & S & $3.0\times10^{-7}$ \\
\bottomrule
\end{tabular}
\Description{Sixteen comparisons of MSPM welfare, platform consumer surplus, Atari outcomes, and controlled ablations, with mean differences, standard errors, test methods, and adjusted p-values. PPO's mean consumer surplus exceeds the fixed cap and selected random policy but is below the clipped rule.}
\end{table}

All comparisons reject their respective nulls at level $0.05$ within their Holm families. PPO's mean consumer surplus exceeds the fixed cap and selected random policy but falls below the clipped rule. All 25 MSPM runs per type exceed optimal SPM welfare. The phase-visibility gain reflects the invalid commitment in Appendix~\ref{sec:ap:limits:noninvariant}. The released analysis records seed-level outcomes, test settings, and unadjusted and adjusted $p$-values.

\FloatBarrier
\flushbottom

\end{document}